\documentclass[11pt,reqno]{amsart}
\usepackage[foot]{amsaddr}
\usepackage{amsmath, amssymb}
\usepackage{amsthm}
\usepackage[a4paper,margin=1.25in]{geometry}
\usepackage{fullpage} 
\usepackage[english]{babel}
\usepackage{hyperref}
\usepackage{multirow}

\usepackage{accents}

\newcommand\ubar[1]{%
  \underaccent{\bar}{#1}}

\usepackage{mathdots}

\usepackage{tikz}
\usepackage{tikz-cd}
\usetikzlibrary{calc}
\usepackage{color, mathtools,epsfig, graphicx}
\usetikzlibrary{positioning}

\newtheorem{theorem}{Theorem}[section]

\newtheorem{proposition}[theorem]{Proposition}

\newtheorem{corollary}[theorem]{Corollary}

\newtheorem{example}[theorem]{Example}

\newtheorem{lemma}[theorem]{Lemma}

\newtheorem{definition}[theorem]{Definition}

\newtheorem{conjecture}[theorem]{Conjecture}

\theoremstyle{remark}

\newtheorem{remark}[theorem]{Remark}

\addto\extrasenglish{

}

\numberwithin{equation}{section}
\newcommand{\C}{\mathbb{C}}
\newcommand{\R}{\mathbb{R}}
\newcommand{\Q}{\mathbb{Q}}
\newcommand{\Z}{\mathbb{Z}}
\newcommand{\p}{\mathbb{P}}
\newcommand{\Pic}{\operatorname{Pic}}
\newcommand{\red}{\operatorname{red}}
\newcommand{\Div}{\operatorname{Div}}
\renewcommand{\div}{\operatorname{div}}
\newcommand{\E}{\mathcal{E}}
\newcommand{\F}{\mathcal{F}}
\newcommand{\A}{\mathcal{A}}
\newcommand{\B}{\mathcal{B}}
\newcommand{\D}{\mathcal{D}}
\newcommand{\K}{\mathcal{K}}
\newcommand{\X}{\mathcal{X}}

\newcommand{\defeq}{\vcentcolon=}

\begin{document}

\title[Deautonomisation by singularity confinement and degree growth]{Deautonomisation by singularity confinement and degree growth}


\author[AS]{Alexander Stokes$^{1,2,\dagger}$}

\address{$^1$Waseda Institute for Advanced Study, Waseda University, 1-21-1 Nishi Waseda Shinjuku-ku Tokyo 169-0051, Japan.}

\author[TM]{Takafumi Mase$^2$}


\author[RW]{Ralph Willox$^2$}

\address{$^2$Graduate School of Mathematical Sciences, The University of Tokyo, 3-8-1 Komaba Meguro-ku Tokyo 153--8914, Japan.}

\author[BG]{Basile Grammaticos$^3$}
\address{$^3$Universit\'e Paris-Saclay, CNRS/IN2P3, IJCLab, 91405 Orsay, France and Universit\'e
de Paris, IJCLab, 91405 Orsay, France}

\address{$^{\dagger}$corresponding author}
\email[Corresponding author]{stokes@aoni.waseda.jp}

\maketitle

\begin{abstract}
In this paper we give an explanation of a number of observations relating to degree growth of birational mappings of the plane and their deautonomisation by singularity confinement.
These observations are of a link between two a priori unrelated notions: firstly the dynamical degree of the mapping and secondly the evolution of parameters required for its singularity structure to remain unchanged under a sufficiently general deautonomisation.
We explain this correspondence for a large class of birational mappings of the plane via the spaces of initial conditions for their deautonomised versions.
We show that even for non-integrable mappings in this class, the surfaces forming these spaces have effective anticanonical divisors and one can define a period map parametrising them, similar to that in the theory of rational surfaces associated with discrete Painlev\'e equations.
This provides a bridge between the evolution of coefficients in the deautonomised mapping and the induced dynamics on the Picard lattice which encode the dynamical degree.
\end{abstract}




\section{Introduction}  \label{section1}


The degree growth of a discrete dynamical system defined by a birational mapping 
$
\varphi : \p^2 \dashrightarrow \p^2,
$
is measured by the (first) dynamical degree 
\begin{equation}
\lambda_1(\varphi) = \lim_{k\to \infty} \left( \deg \varphi^k \right)^{1/k}.
\end{equation}
The quantity $\lim_{k\to\infty} \frac{1}{k} \log \left(\deg \varphi^k \right) = \log \lambda_1(\varphi)$ is used in the integrable systems community under the name `algebraic entropy' \cite{algebraicentropy}, and the system defined by $\varphi$ is said to be integrable if its dynamical degree is 1 or equivalently if its algebraic entropy is 0. 
This means that its degree growth is sub-exponential, which follows the idea that integrability in the discrete case is associated with slow growth in complexity under the dynamics. This idea has roots in Arnol'd's notion of intersection complexity \cite{arnold1, arnold2}, Veselov's ideas of polynomial growth \cite{veselovdiscrete}, and Friedland entropy \cite{friedland, friedlandmilnor}.

The motivation for this paper stems from observations of a remarkable phenomenon in the study of degree growth of birational mappings that have the singularity confinement property.
Singularity confinement, as a property of the mapping,  means that  $\varphi$ is birationally conjugate to an automorphism of a rational surface $X$. 
There is a procedure called deautonomisation by singularity confinement, which produces a non-autonomous mapping, i.e. a sequence of mappings
$\varphi_n : \p^2 \dashrightarrow \p^2$,
such that the structure of destabilising orbits of $\varphi_n$ is the same as for $\varphi$.
This involves embedding $\varphi$ into a family of maps $\varphi_{\mathcal{A}} = \left\{ \varphi_{\boldsymbol{a}} ~|~ \boldsymbol{a} \in \mathcal{A} \right\}$, where $\mathcal{A}$ is some space of parameters, and deriving the parameter evolution $\boldsymbol{a}_n \mapsto \boldsymbol{a}_{n+1}$ such that the singularity structure of $\varphi_n = \varphi_{\boldsymbol{a}_n}$ is the same as $\varphi$.
In the language of singularity confinement, the recurrences governing the parameter evolution are called confinement conditions.


The phenomenon in question can then be summarised as follows. \\

{\it If one performs deautonomisation by singularity confinement on a mapping $\varphi$ using a sufficiently general family $\varphi_{\mathcal{A}}$, then the confinement conditions reduce to a linear recurrence whose characteristic polynomial has $\lambda_1(\varphi)$ as its largest root.}\\

This relationship between dynamical degree and confinement conditions was observed in \cite{justification}  and conjectured to hold  as long as $\varphi$ could be embedded in a sufficiently general family.
Based on this conjectured correspondence, the method of full deautonomisation by singularity confinement was introduced in \cite{fulldeauto, redemption} as a way of determining the degree growth of an autonomous mapping with the singularity confinement property based only on calculations on the level of equations.
This has been shown to work even in cases where other methods \cite{rod, SCasIC, rodexpress} based on singularity confinement on the level of equations fail.

Phrased on the level of mappings, this method involves embedding $\varphi$ into a family $\varphi_{\mathcal{A}}$ of mappings which is taken heuristically as large as it can be in order for deautonomisation  preserving the structure of singularities exactly to be possible.
Then in deautonomising $\varphi$ using $\varphi_{\mathcal{A}}$, if one finds an eigenvalue of magnitude greater than $1$ from the confinement conditions one concludes that this is $\lambda_1(\varphi)$ and the mapping is not integrable. 

This unexpected correspondence between parameter evolution and degree growth requires explanation, and it is this problem that we address in this paper.

\subsection{Background}


\subsubsection{Singularity confinement}
Singularity confinement was proposed in \cite{singularityconfinement} as a discrete counterpart to the Painlev\'e property of ordinary differential equations.
The Painlev\'e property requires that if a solution develops a movable singularity, i.e. one whose location in the complex plane depends on initial conditions, then this does not induce multivaluedness which would obstruct the meaningful definition of a general solution for the equation. 
In the discrete case, one is interested in singularities of mappings that disappear after finitely many iterations. 
Singularity confinement has been used as a heuristic integrability criterion and to isolate non-autonomous versions of autonomous integrable systems. 
In particular, many discrete Painlev\'e equations were found \cite{RGH91,GR93} by applying the method of deautonomisation by singularity confinement to so-called Quispel-Roberts-Thompson (QRT) mappings \cite{QRT1,QRT2}. 

It is important to note, however, that singularity confinement is not a necessary criterion for integrability, in that there exist mappings which are integrable through linearisation and which do not have the singularity confinement property. 
An analogous fact is true in the continuous case in that there are linearisable differential equations which do not have the Painlev\'e property \cite{integrablewithoutpainleve}.
But most importantly, it is also not a sufficient condition since there exist examples of mappings in which singularity confinement holds but the dynamics are chaotic, the most famous example being that due to Hietarinta and Viallet \cite{HVequation}. 

In the years since the realisation that singularity confinement is not sufficient for integrability, thankfully, a complete classification of birational mappings of the plane was obtained, according to their degree growth and whether they have the singularity confinement property, both in the autonomous case by Diller and Favre \cite{dillerfavre} and in the non-autonomous case by one of the authors \cite{MASE}. 
In particular if a non-autonomous mapping with the singularity confinement property has unbounded degree growth, then it must either fit into the framework of discrete Painlev\'e equations given by Sakai \cite{SAKAI2001} or be non-integrable.

We will illustrate singularity confinement first on the level of the equation itself, as it originally was in \cite{singularityconfinement}.
Consider an autonomous second-order discrete equation
\begin{equation} \label{discreteeqgeneral}
(x_{n+1}, y_{n+1} ) = ( f(x_n,y_n), g(x_n,y_n) ),
\end{equation}
where $f$, $g$ are rational functions of their arguments with coefficients independent of $n$, such that the mapping $(x_n,y_n) \mapsto ( f(x_n,y_n), g(x_n,y_n) )$ is birational. 
Note that scalar three-point equations $x_{n+1}= F( x_n, x_{n-1})$ with $F$ homographic in $x_{n-1}$ can be brought to this form by setting $y_n = x_{n-1}$.
By taking $(x_n,y_n)$ and $(x_{n+1},y_{n+1})$ as affine charts on two copies of $\p^2$ the equation \eqref{discreteeqgeneral} defines a birational mapping 
$
\varphi : \p^2 \dashrightarrow \p^2.
$
It is sometimes convenient to instead recast the equation \eqref{discreteeqgeneral} as a birational mapping of $\p^1 \times \p^1$ by taking the variables as affine coordinates in the two $\p^1$ factors, in which case the resulting mapping will be birationally conjugate to $\varphi$ via the standard elementary birational transformation $\p^2 \dashrightarrow \p^1 \times \p^1$.

A singularity is constituted by an iteration from some initial condition that results in a loss of a degree of freedom, in the sense that the iterate takes a value at which the inverse mapping is undefined. 
As already mentioned above, `singularity confinement' refers to the situation where finitely many further iterations lead the singularity to disappear, with the lost degree of freedom being recovered.

\begin{example}[singularity confinement] \label{ex:singconf}
Consider the discrete equation
\begin{equation} \label{0inf0threepoint}
x_{n+1} + x_{n-1} = \frac{1}{x_n^{2}},
\end{equation}
which corresponds to a QRT mapping belonging to Class I in the classification of \cite{RCGOautolimits} (see also \cite{canonicalforms} for a derivation and \cite[Appendix]{ancillary} for an easily consultable list). 


Regarding the iterates $x_n$ as elements of $\p^1$, if while iterating the equation \eqref{0inf0threepoint} some iterate $x_n$ takes the value 0 while $x_{n-1} = u$ is finite, then we have $x_{n+1} = \infty$ irrespective of $u$.
This constitutes a loss of degree of freedom in the sense that the one-parameter family of pairs $(x_{n-1},x_n) = (u, 0)$ all lead to the same value $(x_{n},x_{n+1}) = (0, \infty)$. 
In terms of the mapping $(x_{n-1},x_{n}) \mapsto (x_{n}, x_{n+1})$ of $\p^1 \times \p^1$, the curve given by $x_n=0$ is contracted to a point, i.e. blown down.

Computing further we find that $x_{n+2} = 0$ but then, crucially, see that the value of the next iterate $x_{n+3} = \frac{1}{x_{n+2}^2} - x_{n+1}$ is not defined, 
with the point $(x_{n+1},x_{n+2}) = (\infty,0)$ on $\p^1 \times \p^1$ being an indeterminacy of the mapping.
The way in which singularity confinement is verified, is to take $x_n = \varepsilon$ for some parameter $\varepsilon$ which leads to 
\begin{equation}
x_{n-1} = u, \quad x_n = \varepsilon, \quad x_{n+1} = \frac{1}{\varepsilon^2} -u, \quad x_{n+2} = - \varepsilon + O(\varepsilon^{2}), \quad x_{n+3} =u+ O(\varepsilon).
\end{equation}
Then, defining the values of iterates as the limits of the above as $\varepsilon \to 0$, we have
\begin{equation}
x_{n-1} = u, \quad x_n = 0, \quad x_{n+1} = \infty^2, \quad x_{n+2} = 0, \quad x_{n+3} =u,
\end{equation}
in which $x_{n+1} = \infty^2$ means the iterate is of order $1/\varepsilon^{2}$ as above.
In singularity confinement parlance we say that the lost degree of freedom is recovered through the reappearance of $u$ in the value of $x_{n+3}$, and that the singularity is confined. The results of the above calculation are then summarised by saying that equation \eqref{0inf0threepoint} admits the (confined) singularity pattern $\{0, \infty^2, 0\}$.
\end{example}

\subsubsection{Deautonomisation by singularity confinement}

The procedure of deautonomisation by singularity confinement is based on the idea of obtaining a non-autonomous version of $\varphi$ which, roughly speaking, has the same structure of singularities. 
We again illustrate this first on the level of equations but will make this precise on the level of mappings in Section \ref{section2}.

\begin{example}[deautonomisation by singularity confinement] \label{ex:deauto}
Consider the following generalisation of equation \eqref{0inf0threepoint}:
\begin{equation} \label{dP1general}
x_{n+1} + x_{n-1} = \frac{a_n  + b_n x_n }{x_n^{2}},
\end{equation}
where the sequences $a_n, b_n \in \C$ are to be determined by requiring that the singularity structure of \eqref{dP1general} is the same as equation \eqref{0inf0threepoint}.
As above, some iterate $x_n$ taking the value 0 constitutes a singularity of the mapping defined by \eqref{dP1general}, so we let $x_n = \varepsilon$ and $x_{n-1} = u$, and compute
\begin{equation}
\begin{aligned}
x_{n-1} &= u, \\
x_n &= \varepsilon, \\
x_{n+1} &= \frac{a_n}{\varepsilon^2} + \frac{b_n}{\varepsilon} -u, \\
x_{n+2} &= - \varepsilon + \frac{b_{n+1}}{a_n} \varepsilon^2 + O(\varepsilon^{3}), \\
x_{n+3} &=\frac{a_{n+2} - a_n}{\varepsilon^2} - \left(b_{n+2} -  \frac{2 a_{n+2} b_{n+1}}{a_n} + b_n \right)\frac{1}{\varepsilon} + F(u) + O(\varepsilon),
\end{aligned}
\end{equation}
where $F(u)$ is a known function of $u$ as well as parameters.
In order for the singularity to be confined in the same way as for equation \eqref{0inf0threepoint}, i.e. with the same pattern $\left\{0, \infty^2, 0\right\}$, the singular part of the expansion of $x_{n+3}$ in $\varepsilon$ must vanish. 
This requires 
\begin{equation} \label{confconddP1}
a_{n+2} = a_n
\end{equation}
for all $n$, as well as
\begin{equation} \label{confconddP}
b_{n+2} - 2 b_{n+1} + b_n = 0.
\end{equation}
Performing these calculations again with the confinement conditions \eqref{confconddP1} and \eqref{confconddP} imposed we find  
\begin{equation}
\begin{aligned}
x_{n-1} &= u, \\
x_n &= \varepsilon, \\
x_{n+1} &= \frac{a_n}{\varepsilon^2} + \frac{b_n}{\varepsilon} -u, \\
x_{n+2} &= - \varepsilon + \frac{b_{n+1}}{a_n} \varepsilon^2 + O(\varepsilon^{3}), \\
x_{n+3} &= u + \frac{b_{n+1}(b_{n+1} - b_n)}{a_n} + O(\varepsilon),
\end{aligned}
\end{equation}
and the singularity pattern $\left\{0, \infty^2, 0\right\}$ from the autonomous case \eqref{0inf0threepoint} persists.
Solving the recurrences \eqref{confconddP1} and \eqref{confconddP} leads to 
\begin{equation} \label{dP1}
x_{n+1} + x_{n-1} =  \frac{\alpha+ \beta n}{x_n} + \frac{(-1)^n \gamma + \delta}{x_n^2}  ,
\end{equation}
for some constants $\alpha, \beta, \gamma, \delta$, which reduces to a known example of a discrete Painlev\'e equation when $\beta \neq 0$ \cite{RG96}.
Note that the characteristic polynomial of the linear recurrence \eqref{confconddP} is 
$t^2 - 2 t + 1$,
which has no roots larger than one.
\end{example}

\subsubsection{Dynamical degree reflected in parameter evolution}

The first observations of the dynamical degree being reflected in confinement conditions can be traced back to late-confining versions of discrete Painlev\'e equations  \cite{HVdP1}, where singularities are confined but not at the first opportunity. 
For non-autonomous mappings obtained through this kind of late confinement, as opposed to genuine discrete Painlev\'e equations, the dependence of the coefficients in the mapping on the independent variable $n$ ceases to be additive, multiplicative or elliptic.
Rather, the coefficients satisfy a recurrence relation whose characteristic polynomial has the dynamical degree of the mapping as a root.
Outside the context of late confinement, the same kind of appearance of the dynamical degree in the parameter evolution has been observed in many examples \cite{justification, redemption, fulldeauto}.
However sometimes it is not sufficient to simply allow parameters present in the original mapping to evolve and it is necessary to generalise the mapping further, often by adding extra terms which do not affect the singularity structure, as illustrated in the following example.

\begin{example}[full deautonomisation by singularity confinement] \label{ex:fulldeauto}
Consider the mapping $\varphi$ defined by the discrete equation
\begin{equation} \label{0inf2k0}
x_{n+1} + x_{n-1} = \frac{1}{x_n^{2m}},
\end{equation}
for some integer $m \geq 1$, which can be shown \cite{fulldeauto,rodexpress} to have dynamical degree 
\begin{equation} \label{ex13dynamicaldegree}
\lambda_1(\varphi) = m + \sqrt{m^2-1}.
\end{equation}
Similar analysis to that in example \ref{ex:singconf} shows that this mapping admits the confined singularity pattern $\left\{ 0, \infty^{2m}, 0 \right\}$:
\begin{equation}
x_{n-1} = u, \quad x_n = \varepsilon, \quad x_{n+1} = \frac{1}{\varepsilon^{2m}} -u, \quad x_{n+2} = - \varepsilon + O(\varepsilon^{4m^2}), \quad x_{n+3} =u+ O(\varepsilon).
\end{equation}
 If we proceed to deautonomise, by taking
\begin{equation} \label{0inf2k0deauto}
x_{n+1} + x_{n-1} = \frac{a_n}{x_n^{2m}},
\end{equation}
for $a_n\in \C$ to be determined, we see by similar calculations to those for the $m=1$ case in example \ref{ex:deauto} that
\begin{equation} \label{confcalcs0inf0}
\begin{aligned}
x_{n-1} &= u, \\
x_n &= \varepsilon, \\
x_{n+1} &= \frac{a_n}{\varepsilon^{2m}} -u, \\
x_{n+2} &= - \varepsilon + O(\varepsilon^2), \\
x_{n+3} &=\frac{a_{n+2} - a_n}{\varepsilon^{2m}} + u + O(\varepsilon).
\end{aligned}
\end{equation}
In order for the pattern $\left\{ 0, \infty^{2m}, 0 \right\}$ to persist we must require the confinement condition $a_{n+2} = a_n$. 
Note that this condition is the same regardless of the value of $m$ and, most importantly, does not detect the dynamical degree \eqref{ex13dynamicaldegree} of the mapping.
However, in performing these calculations one notices that adding terms in $\frac{1}{x_n^i}$, $i=1,\dots,2m-1$ to the right-hand side of \eqref{0inf2k0} does not affect the leading behaviours $x_{n+1}\sim \frac{1}{\varepsilon^{2m}}$, $x_{n+2} \sim \varepsilon$, and there is still an opportunity for the singularity to be confined by ensuring that the singular part of the expansion of $x_{n+3}$ in $\varepsilon$ vanishes.
Considering the equation
\begin{equation} \label{0inf0fulldeautoeq}
x_{n+1} + x_{n-1} = \frac{a_n^{(2m)}}{x_n^{2m}} + \sum_{i=2}^{2m-1} \frac{a_n^{(i)}}{x_n^{i}} + \frac{b_n}{ x_{n}},
\end{equation}
with $a_n^{(i)} \in \C$, $i=2,\dots,2m$ and $b_n \in \C$, performing calculations along the same lines as in \eqref{confcalcs0inf0} gives
\begin{equation} \label{confcalcs0inf0general}
\begin{aligned}
x_{n-1} &= u, \\
x_n &= \varepsilon, \\
x_{n+1} &= \frac{a_n^{(2m)}}{\varepsilon^{2m}} + \sum_{i=2}^{2m-1} \frac{a_n^{(i)}}{\varepsilon^{i}} + \frac{b_n}{ \varepsilon} -u, \\
x_{n+2} &= - \varepsilon + \frac{b_{n+1}}{a_n^{(2m)}} \varepsilon^{2m} + O(\varepsilon^{2m+1}), \\
x_{n+3} &= \sum_{i=2}^{2m} \frac{(-1)^i a_{n+2}^{(i)} - a_n^{(i)}}{\varepsilon^{i}}  - \left(b_{n+2} - 2m \frac{a_{n+2}^{(2m)}}{a_{n}^{(2m)}} b_{n+1} + b_n\right) \frac{1}{\varepsilon} +G(u) + O(\varepsilon),
\end{aligned}
\end{equation}
where $G$ is a known rational function of $u$ and parameters.
The confinement conditions on $a^{(i)}_n$ and $b_n$ required for this mapping to admit the singularity pattern $\left\{ 0, \infty^{2m}, 0 \right\}$ are then 
\begin{equation}\label{ichiichigo}
a_{n+2}^{(i)} = (-1)^i a_{n}^{(i)}, ~i=2,\dots,2m, \qquad b_{n+2} - 2 m b_{n+1} + b_{n} = 0.
\end{equation}
With the conditions \eqref{ichiichigo} imposed the singularity is confined, with the lost degree of freedom reappearing via $x_{n+3} = u - \frac{a_n^{(2m-1)}}{a_n^{(2m)}} b_{n+1} + O(\varepsilon).$
The characteristic polynomial for the recurrence satisfied by $b_n$ in \eqref{ichiichigo} is 
\begin{equation}
t^2 - 2 m t + 1,
\end{equation}
whose largest root coincides with the dynamical degree \eqref{ex13dynamicaldegree}.
\end{example}
The term `full deautonomisation by singularity confinement', rather than simply deautonomisation by singularity confinement, was introduced in \cite{fulldeauto, redemption} to refer to the kind of process in example \ref{ex:fulldeauto} by which equation \eqref{0inf0fulldeautoeq} with parameter evolution \eqref{ichiichigo} is obtained by generalising \eqref{0inf2k0} heuristically as much as possible, while still being able to obtain nontrivial parameter evolution such that the singularity patterns of the original autonomous mapping are preserved.

\subsubsection{Spaces of initial conditions for autonomous and non-autonomous mappings of the plane}

The remarkable connection that is observed here between degree growth and parameter evolution is what we aim to explain in this paper.
Our explanation is based on singularity confinement being equivalent to the existence of a space of initial conditions.
For an autonomous system of discrete equations defining a birational mapping $\varphi : \p^2 \dashrightarrow \p^2$, this means that there exists a smooth rational surface $X$ and a birational map $\pi : X \dashrightarrow \p^2$ such that $\tilde{\varphi} \defeq \pi^{-1} \circ \varphi \circ \pi$ is an automorphism of $X$. 

In this situation we have the induced actions on the Picard group $\Pic(X) \cong H^2(X,\Z)$ by pushforward $\varphi_*$ and pullback $\varphi^*$, and $\lambda_1(\varphi)$ can be computed as the spectral radius of $\varphi_*$. 
We remark that to calculate $\lambda_1(\varphi)$ via the induced action on $\Pic(X)$ it is sufficient for $\varphi$ to be conjugate to an algebraically stable birational self-map of some $X$ rather than an automorphism, which goes back to the work of Sibony \cite{SIBONY}.
Diller and Favre showed \cite{dillerfavre} that such $X$ can be constructed in an algorithmic way for any birational mapping $\varphi : \p^2 \dashrightarrow \p^2$, and the problem of calculating the dynamical degree of an autonomous mapping is settled by their procedure, as long as one is prepared to perform the blow-ups involved. 
Practically speaking, in the case of a mapping with the singularity confinement property this amounts to constructing the surface $X$ of which $\varphi$ becomes an automorphism.

In the non-autonomous setting with a sequence of birational mappings $\varphi_n : \p^2 \dashrightarrow \p^2$, one can similarly define the notion of a space of initial conditions. 
This is done in terms of a sequence of rational surfaces $X_n$ and a sequence of birational mappings $\pi_n : X_n \dashrightarrow \p^2$ under which $\varphi_n$ conjugate to isomorphisms $\tilde{\varphi}_n \defeq \pi_{n+1}^{-1} \circ \varphi_n \circ \pi_n : X_n \rightarrow X_{n+1}$, subject to some additional conditions (see Definitions \ref{defspaceofICsNA} and \ref{spaceofICsdefblowdowns}).
These additional conditions are necessary since allowing arbitrary $n$-dependence in the mappings and surfaces leads to pathological examples and prevents the development of a general theory, as discussed in detail in \cite{MASE}.

If $\varphi_n$ has a space of initial conditions in this sense, then all $\Pic(X_n)$ can be identified in a canonical way into a single lattice $\Pic(\X)$ (see Definition \ref{def:picardlattice}) of which the pushforwards $\varphi_n$ induce an automorphism $\Phi$ (see Definition \ref{def:actiononpicardlattice}).
In this case the dynamical degree $\lambda_1(\varphi_n)$ is well-defined as a limit of degrees of compositions $\varphi^{(k)} \defeq \varphi_{n+k-1} \circ \cdots \circ \varphi_{n+1} \circ \varphi_n$ and is computed as the spectral radius of $\Phi$ (see Lemma \ref{lem:takenawa}, due to Takenawa \cite{takenawaDDS}).

\subsubsection{The period map for generalised Halphen surfaces}

In the case of discrete Painlev\'e equations, there is a known bridge between parameter evolution and the induced automorphism $\Phi$ of $\Pic(\X)$.
This is part of the theory of rational surfaces forming the spaces of initial conditions for discrete Painlev\'e equations due to Sakai \cite{SAKAI2001}, and is based on a period map defined on these surfaces, whose construction in special cases goes back to the work of Looijenga \cite{looijenga}.

The space of initial conditions for a discrete Painlev\'e equation is formed of generalised Halphen surfaces (see Definition \ref{def:generalisedhalphensurface}).
In particular each surface has an effective anticanonical divisor $D = \sum_i m_i D_i$, with the classes of the irreducible components $D_i$ spanning a sublattice $Q \subset \Pic(\X)$, which is isomorphic to the root lattice of an affine root system.
The orthogonal complement $Q^{\perp}$ is the root lattice of another affine root system.
Sakai's geometric framework for discrete Painlev\'e equations is based on a classification of the surfaces forming their spaces of initial conditions according to the types of the root systems associated with $Q$ and $Q^{\perp}$ as well as $\operatorname{rank} H_1(D_{\operatorname{red}} , \Z)$, where $D_{\operatorname{red}} = \sum_i D_i$, which determines whether the associated equations are of additive, multiplicative or elliptic type.

Then for a Sakai surface $X$ one can define a period map in terms of a 2-form $\omega$ on $X$ with $-\operatorname{div} \omega = D$ as a function
$\chi : Q^{\perp} \rightarrow \C,$
(see Definition \ref{def:periodmap}), whose construction in the case when $D$ is a cycle of rational curves is due to Looijenga \cite{looijenga}.
Sakai showed \cite{SAKAI2001} that there is a basis of simple roots $\alpha_i$ for $Q^{\perp}$ (which we call a root basis for $Q^{\perp}$, see Definition \ref{rootbasisdef}) which correspond to differences of exceptional curves of the first kind such that $\chi(\alpha_i)$ can be computed in terms of the Poincar\'e residue of $\omega$.
Sakai further realised the set of isomorphism classes of surfaces of each type as a family of surfaces $X_{\boldsymbol{a}}$, with $\boldsymbol{a}\in \mathcal{A}$ in a parameter space consisting of the values of the period map on the root basis and, in cases when the surface type is in addition associated with a differential Painlev\'e equation, an `extra parameter' which becomes the independent variable for the differential equation.

Sakai constructed symmetries of surfaces of each type as an action of an extended affine Weyl group associated to $Q^{\perp}$ on the family $X_{\mathcal{A}}$. 
That is, for each $w$ in the symmetry group, there is an isomorphism $w : X_{\boldsymbol{a}} \rightarrow X_{w.\boldsymbol{a}}$, where the action on parameters $ \boldsymbol{a} \mapsto w. \boldsymbol{a}$ corresponds in a special way to the induced action on $\Pic(\X)$ by pushforward/pullback (see Lemma \ref{rootvarsparamevolution}). 
In particular, the action of a translation element of the symmetry group on $\Pic(\X)$ restricts to $Q^{\perp}$ in the form $\alpha_i \mapsto \alpha_i + c_i \delta$, where $\delta$ is the null root corresponding to the anticanonical divisor class and $c_i \in \Z$ are determined by the weight associated to the translation.
The action of such a translation element on $X_{\mathcal{A}}$ gives a discrete Painlev\'e equation with parameter evolution $a_i \mapsto a_i + c_i d$, where $d \in \C$ corresponds to the value of the period map on the anticanonical class. 
Depending on $\operatorname{rank} H_1(D_{\operatorname{red}} , \Z)$ this gives rise to either additive, multiplicative or elliptic evolution of the coefficients of the discrete Painlev\'e equation when written in coordinates.

\subsection{Outline of main results}

Our explanation of the observed correspondence between degree growth and parameter evolution is based on an extension of the bridge, outlined above, between parameter evolution and dynamics on $\Pic(\X)$ offered by the period map construction.
We extend this beyond the case of discrete Painlev\'e equations to the case of non-integrable mappings which possess spaces of initial conditions. 

This requires, amongst other things, the fact that the surfaces forming the space of initial conditions for a non-integrable mapping have effective anticanonical divisors.
In fact the question of the existence of effective anticanonical and, more generally, antipluricanonical divisors (here meaning a divisor representing some negative multiple of the canonical class) appears frequently in the study of automorphisms of algebraic surfaces and their dynamics \cite{bedfordkim, dillercremona, dillerjacksonsommese, dillerlin, gizatullinGsurfaces, harbourne, mcmullen, zhang}. 
In particular, a conjecture of Gizatullin (communicated to Harbourne by Dolgachev and Looijenga \cite{harbourne}) stated that if a rational surface $X$ has an automorphism $\varphi$ of infinite order, then $X$ must have an effective antipluricanonical divisor. 
In all counterexamples to this conjecture presented by Harbourne \cite{harbourne}, $\varphi$ has zero entropy and the pair $(X,\varphi)$ is non-minimal in the sense that there exists a birational morphism $\pi : X \rightarrow X'$ such that $\pi \circ \varphi \circ \pi^{-1}$ is an automorphism of $X'$ but $\pi$ is not an isomorphism.
McMullen then posed in \cite[Section 12]{mcmullen} a refined version of the conjecture, by asking whether for minimal $(X,\varphi)$, $\varphi$ being of infinite order is sufficient to guarantee that $X$ has an effective antipluricanonical divisor.
For certain classes of $\varphi$ coming from group actions the conjecture is true \cite{zhang}, but these also have zero entropy.
A negative answer to McMullen's question was given by Bedford and Kim \cite{bedfordkim} who constructed a family of mappings which give automorphisms of infinite order with nonzero entropy, on rational surfaces which admit no effective antipluricanonical divisor.

Our first result towards extending the period map construction to the case of spaces of initial conditions for non-integrable maps is Theorem \ref{effectivenesstheorem}. 
This states that in a large class of non-autonomous mappings, namely those that preserve rational 2-forms, even in non-integrable cases the space of initial conditions can be chosen so that the surfaces have effective anticanonical divisors.
This requires us to formulate a notion of minimality of a space of initial conditions in the non-autonomous case (see Definitions \ref{def:contraction} and \ref{def:minimal}), analogous to that mentioned above for the pair $(X,\varphi)$.
Theorem \ref{effectivenesstheorem} then states that if a non-integrable non-autonomous mapping $\varphi_n$ has a minimal space of initial conditions with surfaces $X_n$, and $\varphi_n$ preserves a sequence of rational 2-forms $\omega_n$ then each $X_n$ has an effective anticanonical divisor, under the assumption of some conditions on $n$-dependence in the same spirit as those in the definition of a space of initial conditions for a non-autonomous mapping.

In contrast to the finite list of surface types in the Sakai classification \cite{SAKAI2001}, there are infinitely many possible types of effective anticanonical divisors associated with non-integrable mappings. 
Another main result of this paper is Theorem \ref{theorem:classification} which classifies these based on the homology of $D_{\operatorname{red}}$ and the local intersection configuration of irreducible components.
In particular we see several of the possibilities realised in the examples in sections \ref{section5} and \ref{section6} and in the follow-up to the present paper \cite{part2}. 
In Remark \ref{rem:realisationintheliterature} we list some other types of anticanonical divisors for non-integrable mappings in the classification that have appeared in the literature.

Therefore for a non-integrable mapping with space of initial conditions formed of surfaces with effective anticanonical divisors, one has a sublattice $Q \subset \Pic(\X)$ given by the span of the classes of irreducible components of the anticanonical divisor. 
The next part of the period map construction we extend to the non-integrable case is related to calculation of the values of the period map on $Q^{\perp} \subset \Pic(\X)$ in terms of the parameters in the mapping.


Suppose $\varphi_n$ has a space of initial conditions formed of surfaces $X_n$ with effective anticanonical divisors given by 2-forms $\omega_n$. 
If the surfaces belong to some family as $X_n = X_{\boldsymbol{a}_n}$, with parameter space $\mathcal{A} \ni \boldsymbol{a}_n$, the period map associated to $\omega_n$ can then be defined in a way that gives a function $\chi$ on $Q^{\perp}$ whose values depend on $\mathcal{A}$ (see Definition \ref{def:periodmap}).
If $\chi$ is injective as a map from $Q^{\perp}$ to the space of functions of $\mathcal{A}$, then the parameter evolution $\boldsymbol{a}_n \mapsto \boldsymbol{a}_{n+1}$ will have eigenvalue $\lambda_1(\varphi_n)$. 
This follows from Lemma \ref{rootvarsparamevolution} relating actions on $\Pic(\X)$ and on parameters, as well as the fact that $Q^{\perp}$ contains in its $\R$-span the dominant eigenvector associated with the eigenvalue $\lambda_1(\varphi_n)$ of $\Phi$.

For parameter evolution to detect the dynamical degree via the period mapping, it is sufficient to relax the requirement of injectivity on the whole of $Q^{\perp}$ to a subset which is spanned by elements $\beta_i$ expressible as differences of exceptional curves of the first kind and which contains the dominant eigenvector of $\Phi$. 
We call this a sufficient subset (see Definition \ref{sufficientsubset}).
Our results on the dynamical degree being reflected in parameter evolution can then be summarised as follows.

\begin{theorem} \label{bigsummarytheorem}
Let $\varphi_n$ be a non-autonomous map admitting a space of initial conditions $(X_n,\tilde{\varphi}_n)$.
Suppose that $\tilde{\varphi}_n$ preserves a 2-form $\omega_n$ with an effective anticanonical divisor $D$ on $X_n$. 
If $\Pic(\X)$ admits a sufficient subset spanned by $\beta_i$, and the period map $\chi$ associated to $\omega_n$ is injective on $\operatorname{span} \left\{ \beta_i \right\}$, then $\lambda_1(\varphi_n)$ is an eigenvalue of the induced action of $\varphi_n$ on parameters $\chi(\operatorname{span} \left\{ \beta_i \right\})$.
\end{theorem}

The next main result is Theorem \ref{sufficientsubsetprop}, which states that for a minimal space of initial conditions for a non-integrable non-autonomous mapping, there always exists a sufficient subset of $\Pic(\X)$.

To use these results to study an autonomous map along the lines of full deautonomisation by singularity confinement, we construct a deautonomisation and check that it is sufficient in the sense that $\chi$ is injective on a sufficient subset (see Definition \ref{sufficientdeautonomisationdef}).
This is based on the following, which follows from Theorem \ref{bigsummarytheorem} as well the Definition \ref{deautodef} of deautonomisation. 
\begin{corollary} \label{bigsummarycorollary}
Let $\varphi : \p^2 \dashrightarrow \p^2$ be an autonomous mapping with a space of initial conditions and unbounded degree growth. 
Suppose there exists a rational 2-form $\omega$ on $\p^2$ such that $\varphi^* \omega =   c \, \omega$ for some $c \in \C^*$. 
Then if $\varphi_n = \varphi_{\boldsymbol{a}_n}$ is a sufficient deautonomisation of $\varphi$, the dynamical degree $\lambda_1(\varphi)$ is an eigenvalue of the parameter evolution $\boldsymbol{a}_n\mapsto \boldsymbol{a}_{n+1}$. 
\end{corollary}

The condition that $\varphi$ preserves a 2-form is not a very strong one given the assumption that $\varphi$ has a space of initial conditions, as evidenced by the fact that the only known counterexample is that due to Bedford and Kim \cite{bedfordkim}.
However verifying the injectivity hypothesis is delicate; this is the part of the construction that fails in, for example, the case of the deautonomisation \eqref{0inf2k0deauto} in example \ref{ex:fulldeauto}.
In a follow-up paper \cite{part2} we give sufficient deautonomisations for some more examples, and show that sufficient deautonomisations exist for all maps coming from scalar three-point equations of the the same form as Classes $\rm{I}$-$\rm{VI}$ in the classification of QRT maps in \cite{RCGOautolimits,canonicalforms}. 
This classification is of equations written in the form $F(x_{n+1},x_n,x_{n-1}) = G(x_n)$ for rational functions $F$ and $G$, according to the form that $F$ can take in order for the corresponding mapping to be of QRT type, but here we also consider possibly non-integrable examples of the same form. 
This leads us to conjecture the existence of sufficient deautonomisations in general for mappings with a space of initial conditions with effective anticanonical divisor. 

We also remark that finding the correspondence between the values of the period map and the coefficients in the mapping written in coordinates is sometimes highly nontrivial.
This is the case in, for example, elliptic discrete Painlev\'e equations or QRT-type mappings of Class $\rm{VIII}$. 
We must emphasise, though, that in these cases there are ways to overcome these difficulties and obtain an appropriate parametrisation of coefficients and therefore a linearisation of the confinement conditions such that full deautonomisation by singularity confinement still works \cite{ancillary}.

\subsection{Outline of the paper}

Section \ref{section2} contains the setup relating to spaces of initial conditions for non-autonomous mappings, including induced dynamics on the Picard lattice, degree growth, deautonomisation and our notion of minimality. 
In Section \ref{section3} we introduce the period map and the notions of root basis and sufficient subset, and prove our results relating to existence of effective anticanonical divisors for non-integrable mappings and construction of sufficient subsets. 
We also introduce the notion of sufficient deautonomisation on the level of the space of initial conditions. 
Section \ref{section4} contains the classification of anticanonical divisors for non-integrable mappings,
while Sections \ref{section5} and \ref{section6} illustrate the preceding results in two families of examples.

\section{Spaces of initial conditions for birational mappings of the plane} \label{section2}
In this Section we give the necessary setup in order to give a rigorous formulation and explanation of the mechanism by which the dynamical degree  is reflected in confinement conditions.
We first recall the definition of the singularity confinement property for a second-order mapping, in the sense of the existence of a space of initial conditions for the mapping, in both the autonomous and non-autonomous cases.
We then review the known classifications of autonomous and non-autonomous mappings according to their degree growth and according to whether they have a space of initial conditions.

Most of the setup is given in \cite[Appendix A]{MASE}, but we recall the relevant parts here in order to make the present paper self-contained. 
Throughout the paper we work over $\C$ and use the following notation and conventions. 
\begin{itemize}
\item $X$ : a surface. All surfaces in this paper will be smooth, irreducible, projective, rational surfaces. 
\item $X \dashrightarrow Y$ : a birational mapping between surfaces.
\item $X \rightarrow Y$ : a morphism of surfaces.
\item $\Div(X)$ : the group of divisors on $X$.
\item $\sim$ : linear equivalence of divisors.
\item $\Pic(X)$ : the Picard group of $X$, which is isomorphic to $\Div(X)/\sim$. We often use the same symbol for an element of $\Pic(X)$ as for its corresponding linear equivalence class of divisors and we write the group operation on $\Pic(X)$ additively.
\item $[D]$ : the linear equivalence class of $D \in \Div(X)$. 
We will generally use calligraphic script to indicate classes of divisors, i.e. $\mathcal{D} = [D]$. 
\item $| \F |$ : the linear system of $\F \in \Pic(X)$.
\item $\Pic^+(X) = \left\{ \F \in \Pic(X) ~|~  \F = [F] \text{ where } F \in \Div(X) \text{ is effective}\right\}$: the set of effective classes.
\item $D_1 \cdot D_2$ : the intersection number of the divisors $D_1$ and $D_2$. 
This is well-defined on linear equivalence classes and we use the same notation for the intersection number $\mathcal{D}_1 \cdot \mathcal{D}_2$ of $\mathcal{D}_1$, $\mathcal{D}_2 \in \Pic(X)$.
\item $(D)^2 = (\mathcal{D})^2$ : the self-intersection number of $D \in \Div(X)$ or $\mathcal{D} = [D] \in \Pic(X)$.
\item $\K_{X} \in \Pic(X)$ : the canonical bundle or canonical divisor class of $X$.
\item $\div (\omega) \in \Div(X)$ : the divisor of a rational 2-form $\omega$ on $X$, so $[\div (\omega)] = \K_X$.
\item $\mathcal{O}_{\p^n}(1) \in \Pic(\p^n)$ : the twisting sheaf corresponding to the class of a hyperplane in $\p^n$. 
\item $\Pic_{\Q}(X) = \Pic(X) \otimes \Q$, $\Pic_{\R}(X) = \Pic(X) \otimes \R$, $\Pic_{\C}(X) = \Pic(X) \otimes \C$.
\item $\rho(X)$ : the Picard number of $X$, which is equal to $\operatorname{rank} \Pic(X)$ since $X$ is rational.
\item $H^i(X,D)$ : the $i$-th cohomology group of the divisor $D$ on $X$.
\item $h^{i}(D) = h^i(X,D) = \operatorname{dim} H^i(X,D)$.
\item $b_1 = \operatorname{rank} H^1(X;\Z)$ : the first Betti number.
\item $h^{p,q} = \operatorname{dim} H^{q}(X, \Omega^{p})$ : the Hodge numbers.
\item $g_a(C) = \operatorname{dim}H^1(C, \mathcal{O}_{C})$ : the arithmetic genus of an irreducible curve $C$.
\end{itemize}
We will also use the genus formula $g_a(C) = 1 + \frac{1}{2} [C] \cdot \left( [C] + \mathcal{K}_{X} \right) $ for a curve $C$ on a surface $X$ \cite[Chapter V.1, Exercise 1.3]{hartshorne}.
By an exceptional curve of the first kind, we mean a rational curve of self-intersection -1 on a surface $X$. 
In particular for an exceptional curve $C$ of the first kind on $X$, by the genus formula we have $[C] \cdot \mathcal{K}_X = -1$. 
Such a curve is contractible to a point by a birational morphism $\pi : X \rightarrow X'$.

\subsection{Spaces of initial conditions}

The singularity confinement property of a second-order system of difference equations, defining a birational mapping of the plane, is equivalent to the existence of a space of initial conditions. 
The terminology `space of initial conditions' or `space of initial values' originates in Okamoto's work \cite{OKAMOTO1979} using blow-ups to construct augmented phase spaces on which Hamiltonian forms of the Painlev\'e differential equations are regularised. 
Solutions of these equations can be globally defined by analytic continuation from any point in the space, hence the term space of initial conditions.
The same terminology continues to be used in the discrete case to describe rational surfaces on which birational mappings are regularised in an analogous way. In the autonomous case this means that the birational mapping becomes an automorphism of the surface. 
In the non-autonomous case the $n$-dependence in the coefficients means that the equation becomes a family of mappings, for which a space of initial conditions is a family of rational surfaces between which the mappings become isomorphisms.

\subsubsection{Autonomous case}

Consider an autonomous second-order discrete equation
\begin{equation}
(x_{n+1}, y_{n+1} ) = ( f(x_n,y_n), g(x_n,y_n) ),
\end{equation}
where $f$, $g$ are rational functions of their arguments with coefficients independent of $n$, such that the mapping $(x_n,y_n) \mapsto ( f(x_n,y_n), g(x_n,y_n) )$ is birational. 
By taking $(x_n,y_n)$ as an affine chart we can extend this to a birational mapping 
\begin{equation}
\varphi : \p^2 \dashrightarrow \p^2.
\end{equation}

\begin{definition}[space of initial conditions for an autonomous mapping] \label{defspaceofICsAuto}
We say that an autonomous equation defining a mapping $\varphi$ as above has a space of initial conditions if there exists a surface $X$ and a birational map $\pi : X \dashrightarrow \p^2$ such that $\tilde{\varphi} \defeq \pi^{-1} \circ \varphi \circ \pi$ is an automorphism of $X$:
\begin{center}
\begin{tikzcd}
 X \arrow[r, "\tilde{\varphi}"]  \arrow[d, "\pi", dashed] 		& X  \arrow[d, "\pi", dashed]  \\
\p^2 \arrow[r, dashed, "\varphi"] 					&\p^2 .
\end{tikzcd}
\end{center}
\end{definition}

The automorphism $\tilde{\varphi}$ then induces, by pullback and pushforward, linear transformations of $\Pic(X)$ which we denote by $\tilde{\varphi}^*$ and $\tilde{\varphi}_* = (\tilde{\varphi}^*)^{-1}$ respectively.
These maps are lattice automorphisms of $\Pic(X)$, i.e. $\Z$-module automorphisms which preserve the symmetric bilinear form given by the intersection product, they fix $\K_X$ and preserve effectiveness of divisor classes.

\subsubsection{Non-autonomous case}
In the non-autonomous case, consider a second-order discrete system $(x_{n+1}, y_{n+1} ) = ( f_n(x_n,y_n), g_n(x_n,y_n) ),
$ defining a sequence of mappings
\begin{equation}
\varphi_n : \p^2 \dashrightarrow \p^2,
\end{equation}
 indexed by $n\in \Z$.
In addition to a sequence of surfaces $X_n$ on which $\varphi_n$ conjugates to an isomorphism $X_n \rightarrow X_{n+1}$, for each value of $n$, in the non-autonomous case it is necessary to include in the definition of a space of initial conditions some extra conditions in order to formulate a general theory.
This is because allowing arbitrary $n$-dependence in the mappings and surfaces leads to many pathological examples and prevents us from making any meaningful statements in general; see \cite{MASE} for a detailed discussion. 
In order to state the definition we will work with for non-autonomous mappings, we require the following.

\begin{definition}[basic rational surface]
A basic rational surface is a rational surface that admits $\p^2$ as a minimal model, i.e. there exists a birational morphism from $X$ to $\p^2$.
Such a surface can be obtained from $\p^2$ by a finite number of blow-ups, without the need for blow-downs.
\end{definition}

\begin{definition}[geometric basis] \label{geometricbasisdef}
Let $X$ be a basic rational surface and $(e^{(0)}, \dots, e^{(r)})$ a $\Z$-basis for $\Pic(X)$.
We call $(e^{(0)}, \dots, e^{(r)})$ a geometric basis if there exists a composition of blow-ups $\pi = \pi^{(1)} \circ \cdots \circ \pi^{(r)} : X \rightarrow \p^2$ such that $e^{(0)} = \pi^* \mathcal{O}_{\p^2}(1)$ and $e^{(i)}$ is the class of the total transform (under $\pi^{(i+1)} \circ \cdots \circ \pi^{(r)}$) of the exceptional curve of $\pi^{(i)}$ for $i=1,\dots,r$. 
In this case we say that $(e^{(0)}, \dots, e^{(r)})$ is the geometric basis of $\Pic(X)$ corresponding to $\pi$.
\end{definition}

\begin{definition}[space of initial conditions for a non-autonomous mapping; without blow- downs] \label{defspaceofICsNA}
A space of initial conditions for a non-autonomous mapping $\varphi_n : \p^2 \dashrightarrow \p^2$ consists of sequences $(X_n)_n$ and $(\pi_n)_n$, where each $X_n$ is a basic rational surface and $\pi_n$ is a birational morphism $\pi_n : X_n \rightarrow \p^2$ written as a sequence of blow-ups $\pi_n = \pi_n^{(1)} \circ \cdots \circ \pi_n^{(r)}$, such that the following conditions hold:
\begin{itemize}
\item The mappings $\varphi_n$ become isomorphisms $\tilde{\varphi}_n \defeq \pi_{n+1}^{-1} \circ \varphi_n \circ \pi_n$ as in Figure \ref{fig:spaceofICs1}.
\item Let $e_n = (e_n^{(0)}, \dots, e_n^{(r)})$ be the geometric basis for $\Pic(X_n)$ corresponding to $\pi_n$. Then the matrices of $\tilde{\varphi}_n^{*} : \Pic(X_{n+1}) \rightarrow \Pic(X_n)$ with respect to these bases do not depend on $n$.
\item The set of effective classes in $\Pic(X_n)$ in terms of the basis $(e_n^{(0)},\dots,e_n^{(r)}) $ does not depend on $n$, i.e. if $\sum_{i} a^{(i)} e_n^{(i)} \in \Pic(X_n)$ is effective, then  $\sum_{i} a^{(i)} e_k^{(i)} \in \Pic(X_k)$ is effective for any $k$.
\end{itemize}

\end{definition}

\begin{figure}[htb]
\begin{center}
\begin{tikzcd}[sep=1cm]
\cdots  \arrow[r] & X_{n-1}  \arrow[r, "\tilde{\varphi}_{n-1}"] \arrow[d, "\pi_{n-1}"] & X_{n} \arrow[r, "\tilde{\varphi}_{n}"]  \arrow[d, "\pi_{n}"] & X_{n+1} \arrow[r]  \arrow[d, "\pi_{n+1}"] &\cdots \\
\cdots \arrow[r,dashed] &\p^2 \arrow[r, "\varphi_{n-1}",dashed] &\p^2 \arrow[r, "\varphi_n",dashed] &\p^2 \arrow[r,dashed] & \cdots ,
\end{tikzcd}
\end{center}
 	\caption{Space of initial conditions for a non-autonomous mapping; without blow-downs}
	\label{fig:spaceofICs1}
\end{figure}

While all the examples we will study in this paper will have spaces of initial conditions as defined above, some of our results will rely on the notion of minimality of a space of initial conditions, which requires us to give a more general definition allowing blow-downs in the construction of $X_n$ from $\p^2$.

\begin{definition}[space of initial conditions for a non-autonomous mapping; allowing blowdowns] \label{spaceofICsdefblowdowns}

A space of initial conditions for a non-autonomous mapping $\varphi_n$, allowing blow-downs, consists of sequences $(X_n)_n$, $(Y_n)_n$, $(\pi_n)_n$, $(\rho_n)_n$, and $(\sigma_n)_n$, with 
\begin{itemize}
\item $X_n$ a basic  rational surface with birational morphism $\pi_n : X_n \rightarrow \p^2$, written as a sequence of blow-ups as  $\pi_n = \pi_n^{(r)} \circ \cdots \circ \pi_n^{(1)}$,
\item $Y_n$ a basic rational surface with birational morphism $\rho_n : Y_n \rightarrow \p^2$, written as a sequence of blow-ups as $\rho_n = \rho_n^{(s)} \circ \cdots \circ \rho_n^{(1)}$,
\item $\sigma_n : Y_n \rightarrow X_n$ a birational morphism written as a sequence of blow-ups as $\sigma_n = \sigma_n^{(s-r)} \circ \cdots \circ \sigma_n^{(1)}$,
\end{itemize}
such that the following conditions hold:
\begin{itemize}
\item The mappings $\tilde{\varphi}_n \defeq \sigma_{n+1} \circ \rho_{n+1}^{-1} \circ \varphi_n \circ \rho_n \circ \sigma_n^{-1}: X_n \rightarrow X_{n+1}$ are isomorphisms as in Figure \ref{fig:spaceofICs2}.
\item Let $f_n = (f_n^{(0)}, \dots, f_n^{(s)})$ be the geometric basis for $\Pic(Y_n)$ corresponding to $\rho_n$. 
Then the set of effective classes in $\Pic(Y_n)$ in terms of the basis $f_n$ does not depend on $n$.
\item Let $e_n = (e_n^{(0)}, \dots, e_n^{(r)})$ be the geometric basis for $\Pic(X_n)$ corresponding to $\pi_n$. 
Then the set of effective classes in $\Pic(X_n)$ in terms of the basis $e_n$ does not depend on $n$.
\item Let $F_n^{(j)} \in \Pic(Y_n)$ be the class of (the total transform of) the exceptional curve contracted by $\sigma_n^{(j)}$, for $j=1,\dots, s-r$. 
Then the expression for $F_n^{(j)}$ in terms of the basis $f_n$ does not depend on $n$.
\item The expression for $\sigma_n^{*} (e_n^{(i)})$ in the basis $f_n$ for $\Pic(Y_n)$ does not depend on $n$. 
\item The matrices of $\tilde{\varphi}_n^{*} : \Pic(X_{n+1}) \rightarrow \Pic(X_n)$ with respect to the bases $e_{n+1}$ and $e_{n}$ do not depend on $n$.

\end{itemize}

\end{definition} 

\begin{figure}[htb]
\begin{center}
\begin{tikzcd}[row sep=.5cm,column sep=.5cm]
 \cdots \arrow[r,dashed] & Y_{n-1} \arrow[rr,  dashed] \arrow[dr, "\sigma_{n-1}" ] \arrow[dd, "\rho_{n-1}", near end] &     & Y_n \arrow[dr, "\sigma_{n}"] \arrow[dd, "\rho_{n}", near end] \arrow[rr,  dashed] &    &  Y_{n+1} \arrow[dr, "\sigma_{n+1}"] \arrow[dd, "\rho_{n+1}", near end] \arrow[r, dashed] & \cdots  &   \\
 & \cdots \arrow[r]  & X_{n-1}\arrow[rr, "\tilde{\varphi}_{n-1}",near start] \arrow[dd,"\pi_{n-1}", near end] &   &  X_{n} \arrow[dd, "\pi_n", near end] \arrow[rr, near start, "\tilde{\varphi}_{n}"]  &  &  X_{n+1} \arrow[dd,"\pi_{n+1} ", near end] \arrow[r] & \cdots \\
\cdots \arrow[r,dashed] & \p^2 \arrow[rr, "\varphi_{n-1}", near start, dashed] &     & \p^2 \arrow[rr, "\varphi_{n}", near start, dashed] &   &  \p^2 \arrow[r, dashed] & \cdots &  \\
 & \cdots \arrow[r,dashed]  & \p^2 \arrow[rr, dashed] &   &  \p^2 \arrow[rr, dashed]&  & \p^2 \arrow[r,dashed] & \cdots
\end{tikzcd}
\end{center}
 	\caption{Space of initial conditions for a non-autonomous mapping; allowing blow-downs}
	\label{fig:spaceofICs2}
\end{figure}

\begin{remark}
The key feature of the more general Definition \ref{spaceofICsdefblowdowns} of a space of initial conditions is that $\sigma_n$ may contract curves on $Y_n$ which $\rho_n$ does not, in which case the mapping $\pi_{n+1} \circ \tilde{\varphi}_n \circ \pi_{n}^{-1} : \p^2 \dashrightarrow \p^2$ will not coincide with $\varphi_n$, but rather be conjugate to it by a birational coordinate change on $\p^2$.  
\end{remark}

\begin{definition}[Picard lattice for a non-autonomous mapping with a space of initial conditions] \label{def:picardlattice}
Let $\varphi_n$ have a space of initial conditions in the sense of Definition \ref{spaceofICsdefblowdowns}.
Using the geometric bases $e_n$ we may identify all $\Pic(X_n)$ with a single $\Z$-module
\begin{equation}
\Pic(\X) = \sum_{i=0}^r \Z e^{(i)},
\end{equation}
via 
\begin{equation}
\begin{gathered}
\iota_n : \Pic(\X) \longrightarrow \Pic(X_n), \\
\iota_n(e^{(i)}) = e_n^{(i)},
\end{gathered}
\end{equation} 
The intersection form on $\Pic(X_n)$ can be used to equip $\Pic(\X)$ with the symmetric bilinear form defined by
\begin{equation}
 e^{(i)} \cdot e^{(j)}  = 
\begin{cases}
+1 &(i=j=0) \\
-1 &(i=j\neq0) \\
0   & (i\neq j).
\end{cases}
\end{equation}
We will refer to $\Pic(\X)$ equipped with this bilinear form as the Picard lattice of the space of initial conditions $X_n$. 
\end{definition}

The canonical class in each $\Pic(X_n)$ is identified with the element
\begin{equation}
\K_{\X} \defeq \iota_n^{-1}(\K_{X_n})  = - 3 e^{(0)} + e^{(1)} + \dots + e^{(r)} \in \Pic(\X).
\end{equation}

\begin{definition}[action on the Picard lattice of a non-autonomous mapping with a space of initial conditions]
\label{def:actiononpicardlattice}
Denote the map induced by $\varphi_n^*$, for any $n$, by $\Phi : \Pic(\X) \rightarrow \Pic(\X)$, which is well-defined because of the conditions in Definition \ref{defspaceofICsNA}:
\begin{center}
\begin{tikzcd}[sep=1cm]
& \Pic(\X) \arrow[d, "\iota_{n}"] 	& \Pic(\X) \arrow[d, "\iota_{n+1}"] \arrow[l, "\Phi"] \\
& \Pic(X_n) 				& \Pic(X_{n+1}) \arrow[l, "\tilde{\varphi}_n^*"] .
\end{tikzcd}
\end{center}
We call $\Phi$ the action of the non-autonomous mapping $\varphi_n$ on the Picard lattice $\Pic(\X)$ . 
\end{definition}
Then $\Phi$ is a lattice automorphism of $\Pic(\X)$ that fixes the element $\K_{\X}$ and preserves effectiveness, in the sense that if $\F \in \Pic(\X)$ is such that $\iota_{n+1}(\F)$ is an effective class on $X_{n+1}$, then $\iota_{n} \circ \Phi (\F)$ is an effective class on $X_{n}$. 

\subsubsection{Minimal spaces of initial conditions}

Allowing for blow-downs in the definition of a space of initial conditions means we can introduce the notion of a minimal space of initial conditions for a given non-autonomous mapping.

\begin{definition}[contraction of a space of initial conditions] \label{def:contraction}
Let $\varphi_n : \p^2 \dashrightarrow \p^2$ be a non-autonomous mapping with a space of initial conditions allowing blow-downs
as in Definition \ref{spaceofICsdefblowdowns}, consisting of sequences $(X_n)_n$, $(Y_n)_n$, $(\pi_n)_n$, $(\rho_n)_n$, and $(\sigma_n)_n$.
A contraction of the space of initial conditions consists of sequences $(X_n')_n$ and $(\mu_n)_n$ with 
\begin{itemize}
\item $X_n'$ a basic rational surface
\item a birational morphism $\mu_n : X_n \rightarrow X_n'$ which is a sequence of blow-ups $\mu_n = \mu_n^{(1)}\circ \cdots \circ \mu_n^{(r')} $ 
\end{itemize}
such that the following conditions hold:
\begin{itemize}
\item 
For every $n$, the map $\tilde{\varphi}_n' \defeq \sigma_{n+1}^{-1} \circ \tilde{\varphi}_n \circ \sigma_n : X_n' \rightarrow X_{n+1}'$
is an isomorphism.
\item 
Let $e_n = (e_n^{(0)},\dots , e_n^{(r)})$ be the geometric basis for $\Pic(X_n)$ corresponding to $\pi_n$, and let $E_n^{(k)} \in \Pic(X_n)$ be the class of (the total transform of) the exceptional divisor contracted by $\mu_n^{(k)}$ for $k=1,\dots,r'$. 
Then the expression for $E_n^{(k)}$ in terms of the basis $e_n$ is independent of $n$.
\item There exist birational morphisms $\pi_n' : X_n' \rightarrow \p^2$ such that the geometric bases corresponding to $\pi_n'$ satisfy the requirements for $X_n'$ to form a space of initial conditions for $\varphi_n$ in the sense of Definition \ref{spaceofICsdefblowdowns}.
\end{itemize}
\end{definition}

We give a diagram showing the birational mappings involved in contraction of a space of initial conditions in Figure \ref{fig:minimisationdiagram}.
 
\begin{figure}[htb]
\begin{center}
\begin{tikzcd}[row sep=.5cm,column sep=.4cm]
& \cdots \arrow[r, dashed] & Y_{n-1} \arrow[rrr,dashed] \arrow[ddd,"\rho_{n-1}"] \arrow[dr, "\sigma_{n-1}"]& &  & Y_{n} \arrow[rrr,dashed] \arrow[ddd,"\rho_{n}"] \arrow[dr, "\sigma_{n}"]& & & Y_{n+1} \arrow[r, dashed] \arrow[ddd,"\rho_{n+1}"] \arrow[dr, "\sigma_{n+1}"]& \cdots & & & &\\
& & \cdots \arrow[r] & X_{n-1} \arrow[rrr]  \arrow[ddd,"\pi_{n-1}"] \arrow[dr, "\mu_{n-1}"]& &  & X_{n} \arrow[rrr]  \arrow[ddd,"\pi_{n}"] \arrow[dr, "\mu_{n}"]& & & X_{n+1} \arrow[r]  \arrow[ddd,"\pi_{n+1}"] \arrow[dr, "\mu_{n+1}"]& \cdots & & & \\
& & & \cdots \arrow[r] & X_{n-1}' \arrow[rrr] \arrow[ddd,"\pi_{n-1}'"]& &  & X_n' \arrow[rrr]  \arrow[ddd,"\pi_{n}'"]& & & X_{n+1}' \arrow[r]  \arrow[ddd,"\pi_{n+1}'"]& \cdots & & \\
& \cdots \arrow[r, dashed] & \p^2 \arrow[rrr,dashed, "\varphi_{n-1}"] & &  & \p^2 \arrow[rrr,dashed, "\varphi_{n}"] & & & \p^2 \arrow[r, dashed] & \cdots & & & &\\
& & \cdots \arrow[r, dashed] & \p^2 \arrow[rrr,dashed] & &  & \p^2 \arrow[rrr,dashed] & & & \p^2 \arrow[r, dashed] & \cdots & & & \\
& & & \cdots \arrow[r, dashed] & \p^2 \arrow[rrr,dashed] & &  & \p^2 \arrow[rrr,dashed] & & & \p^2 \arrow[r, dashed] & \cdots & & 
\end{tikzcd}
\end{center}
 	\caption{Contraction of a space of initial conditions}
	\label{fig:minimisationdiagram}
\end{figure}

\begin{definition}[minimal space of initial conditions] \label{def:minimal}
For a space of initial conditions $X_n$, we say that a contraction is trivial if the maps $\mu_n$ are isomorphisms. If there does not exist any nontrivial contraction, then we say the space of initial conditions is minimal.
\end{definition}

Whether or not contraction of a space of initial conditions is possible is characterised by the following.

\begin{lemma}[\cite{MASE}] \label{minimisationlemma}
If there exist a finite set of elements of $\Pic(\X)$ corresponding to mutually disjoint exceptional curves of the first kind which are permuted by $\Phi$, then blowing down these curves provides a nontrivial contraction of the space of initial conditions.
\end{lemma}

\subsection{Degree growth and entropy}

Hereafter, the definition of integrability of a birational mapping of the plane we will work with is that its algebraic entropy vanishes, or equivalently that its dynamical degree is 1.
For a birational map $f$ from $\p^2$ to itself, written in homogeneous coordinates as
\begin{equation}
\begin{gathered}
f : \p^2 \dashrightarrow \p^2, \\
(x_0 : x_1: x_2) \mapsto \left( f_0(x_0,x_1,x_2) :  f_1(x_0,x_1,x_2) : f_2(x_0,x_1,x_2) \right),
\end{gathered}
\end{equation}
where $f_0,f_1,f_2$ are homogeneous polynomials of the same degree with no common factors, its degree $\deg f$ is defined as the common degree of the polynomials $f_i$.

\begin{definition}[dynamical degree, algebraic entropy {\cite{algebraicentropy}}]
For a possibly non-autonomous mapping $\varphi_n : \p^2 \dashrightarrow \p^2$, the numbers
\begin{equation}
\lambda_1(\varphi_n) = \lim_{k \to \infty}  \left(\deg \varphi^{(k)}\right)^{1/k}, \quad \text{ and } \quad \log \lambda_1(\varphi_n) = \lim_{k \to \infty} \frac{1}{k} \log \left(\deg \varphi^{(k)}\right),
\end{equation}
where $\varphi^{(k)} = \varphi_{n+k-1} \circ \cdots \circ \varphi_{n+1} \circ \varphi_n$, are called the dynamical degree and algebraic entropy \cite{algebraicentropy} respectively, in the case the limits exist.
Note that in the autonomous case $\varphi_n = \varphi$ and $\varphi^{(k)} = \varphi^k$, and the definition of $\lambda_1(\varphi)$ coincides with that of the (first) dynamical degree of a dominant rational self-map of $\p^2$, which is invariant under birational conjugation of $\varphi$.
For a non-autonomous mapping with a space of initial conditions as defined in this paper, the limit $\lambda_1(\varphi_n)$ is well-defined.


\end{definition}

For an autonomous mapping, the dynamical degree $\lambda_1(\varphi)$ is bounded from above by the algebraic degree of $\varphi$, with equality if and only if $\varphi$ is algebraically stable as a birational self-map of $\p^2$ \cite{blanccantat}, i.e. the pushforward action on $\Pic(\p^2)$ satisfies $(\varphi_*)^k = (\varphi^k)_*$ for all integers $k>0$.
This is equivalent to the non-existence of any point of indeterminacy $p \in \mathcal{I}(\varphi^{-1})$ of the inverse map $\varphi^{-1}$ such that $\varphi^{k}(p) \in \mathcal{I}(\varphi)$ for some $k\geq 0$ \cite{dillerfavre}. 
In \cite{jaume} this is referred to as the absence of degree-lowering hypersurfaces.

The possible rates of degree growth of a mapping with a space of initial conditions are provided by the following classifications.

\subsubsection{Autonomous case}

In \cite{dillerfavre}, Diller and Favre classified bimeromorphic self-maps of compact surfaces according to the rate of growth of their spectral radii under iteration. 
Restating some of their results in the language of birational maps, as studied in discrete integrable systems, yields the following.

\begin{proposition}[Diller-Favre \cite{dillerfavre}] \label{classificationpropauto}
An autonomous birational mapping $\varphi$ of the plane belongs to one of the following types:
\begin{description}

\item[type (1)] 
The degree of $\varphi^k$ is bounded.\\
This type of mapping has a space of initial conditions.

\item[type (2)] 
The degree of $\varphi^k$ grows linearly.\\
This type of mapping does not have a space of initial conditions.

\item[type (3)] 
The degree of $\varphi^k$ grows quadratically.\\
This type of mapping has a space of initial conditions, which is a rational elliptic surface.

\item[type (4)] 
The degree of $\varphi^k$ grows exponentially and $\varphi$ has a space of initial conditions.

\item[type (5)] 
The degree of $\varphi^k$ grows exponentially and $\varphi$ does not have a space of initial conditions.

\end{description}

\end{proposition}

\subsubsection{Non-autonomous case}

If a mapping, autonomous or not, has a space of initial conditions then the dynamical degree is encoded in the induced dynamics on its Picard lattice.

\begin{lemma}[Takenawa {\cite{takenawaDDS}}] \label{lem:takenawa}
If a non-autonomous mapping $\varphi_n$ has a space of initial conditions, the degree of $\varphi^{(k)}$ is given by 
\begin{equation}
\deg \varphi^{(k)} = \Phi^k e^{(0)} \cdot e^{(0)},
\end{equation}
where $\Phi$ is the action of the mapping on the Picard lattice.
Then the dynamical degree of the mapping is given by the largest eigenvalue of $\Phi$, and the algebraic entropy by its logarithm.
Similarly if an autonomous mapping $\varphi$ has a space of initial conditions then the dynamical degree is given by the largest eigenvalue of $\varphi^*$.
\end{lemma}

Using the above fact and the work of Diller-Favre \cite{dillerfavre}, one of the authors obtained the following classification of non-autonomous mappings with spaces of initial conditions.

\begin{proposition}[{\cite{MASE}}]\label{classificationpropnonauto}
A non-autonomous mapping of the plane with a space of initial conditions belongs to one of the following three types according to the Jordan normal form of $\Phi$:
\begin{description}
\item[type (a)] 
$$
\Phi \sim \left(\begin{array}{ccc}\mu_1 &   &   \\  & \ddots &   \\  &   & \mu_{r+1}\end{array}\right),
$$
where $\mu_i$ are all roots of unity. In particular $\Phi^l = \operatorname{id}$ for some $l>0$ so the degree growth of the mapping is bounded. 
\item[type (b)]
$$
\Phi \sim \left(\begin{array}{cccccc}1 & 1 &   &   &   &   \\  & 1 & 1 &   &   &   \\  &   & 1 &   &   &   \\  &   &   & \mu_1 &   &   \\  &   &   &   & \ddots &   \\  &   &   &   &   & \mu_{r-2}\end{array}\right),
$$
where $\mu_i$ are all roots of unity. In this case the degree grows quadratically.
\item[type (c)] 
$$
\Phi \sim \left(\begin{array}{ccccc}\lambda &  &   &   &   \\  & \frac{1}{\lambda} &  &   &   \\  &   & \mu_1 &   &   \\  &   &   & \ddots &   \\  &   &   &   & \mu_{r-1}\end{array}\right),
$$
where $\lambda > 1$ is a reciprocal quadratic integer or a Salem number, and $|\mu_i|=1$. In this case the degree grows exponentially and the dynamical degree is $\lambda$.
\end{description}
\end{proposition}

We can partially distinguish between the different types above according to $\operatorname{rank}\Pic(\X)$ as follows.

\begin{proposition}[\cite{MASE}]  \label{picardnumberprop}
If $\operatorname{rank} \Pic(\X) \leq 10$, then the mapping belongs to either type (a) or type (b) of Proposition \ref{classificationpropnonauto}.
\end{proposition}

For type (b), when the degree growth is quadratic, one of the authors also showed \cite{MASE} that a minimal space of initial conditions is formed of generalised Halphen surfaces as defined by Sakai as follows.
\begin{definition}[generalised Halphen surface \cite{SAKAI2001}, Sakai surface] \label{def:generalisedhalphensurface}
A rational surface $X$ is called a generalised Halphen surface if it has an effective anticanonical divisor $D \in |-\K_X|$ of canonical type, i.e if $D = \sum_{i} m_i D_i$, $m_i > 0$, is its decomposition into irreducible components then $[D_i] \cdot \K_X = 0$ for all $i$.
A generalised Halphen surface has $\dim  |-\K_X|$ equal to either 1 or 0. 
In the first case $X$ is a rational elliptic surface and in the second case $X$ has a unique effective anticanonical divisor. The latter is the type of surface which forms the spaces of initial conditions for discrete Painlev\'e equations, and we call such $X$ a Sakai surface.
\end{definition}

From Proposition \ref{classificationpropnonauto} we see that if a non-autonomous mapping has a space of initial conditions and unbounded degree growth, then it can be shown to be either a discrete Painlev\'e equation (type (b)) or non-integrable (type (c)). 

\begin{remark} \label{basicsurfacejustification}
The restriction to basic rational surfaces in the Definitions \ref{defspaceofICsNA} and \ref{spaceofICsdefblowdowns} of a space of initial conditions is justified for surfaces associated with mappings of type (b) or (c), since in both cases the surfaces have infinitely many exceptional curves of the first kind and so by a theorem of Nagata \cite{nagata} must be basic rational surfaces.
\end{remark}

\subsection{Deautonomisation}

We now define what it means for a non-autonomous mapping to be a deautonomisation of an autonomous one with reference to the space of initial conditions and dynamics on the Picard lattice.

\begin{definition}[deautonomisation] \label{deautodef}
Consider a non-autonomous mapping $\varphi_n : \p^2 \dashrightarrow \p^2$ with a space of initial conditions as in Definition \ref{spaceofICsdefblowdowns} (allowing blow-downs). 
Consider also an autonomous mapping $\psi : \p^2 \dashrightarrow \p^2$ with a space of initial conditions in the sense of Definition \ref{defspaceofICsAuto} being a basic rational surface $Z$ with $\operatorname{rank} \Pic(Z) = \operatorname{rank} \Pic(\X)$ to which $\psi$ lifts to an automorphism $\tilde{\psi} : Z \rightarrow Z$. 
Consider a birational morphism $Z \rightarrow \p^2$ (given by a sequence of blow-ups) and denote the geometric basis corresponding to that morphism by $(f^{(0)},\dots, f^{(r)})$. 
We say that $\varphi_n$ is a deautonomisation of $\psi$ if under the identification
\begin{equation}
\begin{gathered}
\kappa : \Pic(\X) \rightarrow \Pic(Z),\\
e^{(i)} \mapsto f^{(i)},
\end{gathered}
\end{equation}
the following two conditions hold:
\begin{itemize}
\item If $\F \in \Pic(\X)$ represents an effective class $\iota_n(\F)$ in $\Pic(X_n)$, then $\kappa(\F) \in \Pic(Z)$ is also effective.
\item The linear maps $\Phi$ and $\tilde{\psi}^*$ coincide under the identification, i.e. the following diagram commutes:
\end{itemize}
\begin{center}
\begin{tikzcd}
& \Pic(X_n)									& \Pic(X_{n+1}) \arrow[l, "\tilde{\varphi}_n^*"] \\
& \Pic(\X)   \arrow[d,"\kappa"]	\arrow[swap]{u}{\iota_{n}}	& \Pic(\X) \arrow[l, "\Phi"]  	\arrow[swap]{d}{\kappa} \arrow[u, "\iota_{n+1}"] 	\\
& \Pic(Z) 										& \Pic(Z)	\arrow[l,"\tilde{\psi}^*"]	 .
\end{tikzcd}
\end{center}
\end{definition}

\begin{remark}
This definition is consistent with that of \cite{fiberdependent} formulated in the context of deautonomising QRT mappings to obtain discrete Painlev\'e equations. 
The notion of deautonomisation there involves a choice of fibre from the elliptic fibration of the space of initial conditions for a QRT mapping, which then becomes the unique anticanonical divisor of the surface after allowing locations of centres of blow-ups to move. 
In particular, a different choice of singular fibre will lead to a different choice of which classes stay effective after deautonomisation, and therefore corresponds to a different way of injecting $\Pic^+(X_n)$ into $\Pic^+(Z)$ via $\kappa \circ \iota_n^{-1}$, but this does not change $\Phi$ nor the degree growth of the deautonomised mappings.
\end{remark}

\begin{remark}
If $(\varphi_n,X_n)$ is a deautonomisation of $(\psi,Z)$, then if a nontrivial contraction of the space of initial conditions for $\varphi_n$ is possible, the same will be true of $\psi$, i.e. there will exist a birational morphism $\pi :Z \rightarrow Z'$ such that $ \pi \circ \tilde{\psi} \circ \pi^{-1}$ is an automorphism of $Z'$ but $\pi$ is not an isomorphism. 
So after contraction $(\varphi_n,X_n)$, we do not have a deautonomisation of $(\psi,Z)$ itself, but rather of the contracted version $(\pi \circ \tilde{\psi} \circ \pi^{-1},Z')$.
\end{remark}

The crucial fact that we will use for the remainder of the paper is the following
\begin{lemma}
If $\varphi_n$ is a deautonomisation of an autonomous mapping in the sense of Definition \ref{deautodef}, then it has exactly the same degree growth, dynamical degree and algebraic entropy.
\end{lemma}

\section{Parameter evolution and the period map} \label{section3}

We are now ready to extend the period map construction to the case of a non-integrable mapping which possesses a space of initial conditions and preserves a 2-form.

\subsection{Period map for rational surfaces with effective anticanonical divisor}

In this subsection, let $X$ be a surface with an effective anticanonical divisor $D = \sum_{i} m_i D_i \in | -\K_{X}|$.
Denote the sublattice of $\Pic(X)$ spanned by the classes $\mathcal{D}_i = [D_i]$ of the irreducible components of $D$ by 
\begin{equation}
Q = \sum_{i} \Z \,\mathcal{D}_i.
\end{equation} 
Its orthogonal complement is another sublattice, which we denote 
\begin{equation}
Q^{\perp} = \left\{ \F \in \Pic(X)~|~ \F \cdot \mathcal{D}_i = 0~~\forall~i~\right\},
\end{equation}
on which the period map will give a $\C$-valued function.
When $X$ is a Sakai surface, $Q^{\perp}$ is isomorphic to the root lattice of an affine root system and the values of the period map on a basis of simple roots are called the root variables. These appear as $n$-dependent parameters in discrete Painlev\'e equations \cite{SAKAI2001}. 

Begin by taking a rational $2$-form $\omega$ on $X$ such that $- \div \omega = D$.
Let the support of $D$ be $D_{\red} = \bigcup_{i} D_i$, so $\omega$ defines a holomorphic symplectic form on $X - D_{\red}$. 
This gives the mapping
\begin{equation}
\begin{gathered}
\hat{\chi} : H_2(X - D_{\red} ; \Z) \rightarrow \C, \\
\Gamma \mapsto \int_{\Gamma} \omega.
\end{gathered}
\end{equation}
In order to use this to define a period map on $Q^{\perp}$, we begin with the long exact sequence of relative singular homology for the pair $(X, X - D_{\red})$, which includes the following:
\begin{equation*}
H_{3}(X ; \Z) \rightarrow H_{3}(X, X-D_{\red} ; \Z) \xrightarrow{\partial^*} H_{2}(X - D_{\red};\Z) \xrightarrow{i^*} H_2(X;\Z) \xrightarrow{j^*} H_2(X,X-D_{\red};\Z). 
\end{equation*}
Here $i^*$ comes from the natural injection of cycles on $X-D_{\red}$ into $X$, while $j^*$ is induced by the quotient of cycles on $X$ by cycles on $X - D_{\red}$. We manipulate this exact sequence using the following facts:

\begin{itemize}
\item Lefschetz duality for the pair $(X, X - D_{\red})$ gives 
\begin{equation}
H_k(X, X-D_{\red} ; \Z) \cong H^{4-k}(D_{\red} ; \Z).
\end{equation}
\item Poincar\'e duality for $D_{\red}$ gives
\begin{equation}
H_1(D_{\red} ; \Z) \cong H^1(D_{\red} ; \Z).
\end{equation}
\item Under Poincar\'e duality for $X$, we have 
\begin{equation}
H_3(X;\Z) \cong H^1(X;\Z) = 0,
\end{equation}
because $b_1 = \operatorname{rank} H^1(X;\Z) = \sum_{p+q=1} h^{p,q} = h^{1,0} + h^{0,1} =  2 h^{0,1} = 0$, since $X$ being rational implies that $h^{0,1}=0$.
\item Lefschetz duality as above gives
$H_2(X, X-D_{\red} ; \Z) \cong H^2(D_{\red} ; \Z) \cong Q$,
so under the Poincar\'e duality $H_2(X;\Z) \cong H^2(X;\Z)$, in the long exact sequence above we have 
\begin{equation}
\ker j^* \cong Q^{\perp},
\end{equation}
which is why the orthogonal complement of $Q$ is significant in this construction. 
\end{itemize}
Using these facts, from the sequence above we obtain
\begin{equation}
0 \rightarrow H_{1}(D_{\red} ; \Z) \rightarrow H_{2}(X - D_{\red}; \Z) \rightarrow Q^{\perp} \rightarrow 0,
\end{equation}
which through the mapping $\hat{\chi}$ on $H_2(X - D_{\red};\Z)$ gives the following.

\begin{definition}[period map \cite{looijenga,SAKAI2001}] \label{def:periodmap}
For $X$ with effective anticanonical divisor $D$ and choice of a rational 2-form $\omega$ on $X$ with $- \div \omega = D$, the above construction gives
\begin{equation}
\chi : Q^{\perp} \rightarrow \C ~\mod \hat{\chi}(H_1(D_{\red};\Z)),
\end{equation}
which we call the period map of $X$.
One can define $\chi$ as a $\C$-valued function by making a normalisation of the value of $\hat{\chi}$ on $H_1(D_{\red};\Z)$ if this is nontrivial.
\end{definition}

\subsubsection{Computation of the period map}

The computation of the period map proceeds along the same lines as in \cite[Chapter I, Section 5]{looijenga} and \cite[Lemma 21]{SAKAI2001}, which we recall here for completeness. 
The first step in the procedure is guaranteed to be possible for generalised Halphen surfaces, but for surfaces associated with non-integrable mappings does not come automatically either because the existence of an effective anticanonical divisor is not guaranteed or, even if this does exist, it is not obvious how to describe $Q^\perp$ as a root lattice; we will address this in subsection \ref{subsec:rootbases} below.

For an element $\alpha \in Q^{\perp}$ that can be expressed as the difference of the classes of two exceptional curves of the first kind,
\begin{equation}
\alpha = [C^1] - [C^0],
\end{equation}
the computation of $\chi(\alpha)$ is done as follows:
\begin{itemize}
\item Find the unique component $D_k$ of the anticanonical divisor which intersects both $C^1$ and $C^0$, with multiplicity one, i.e. $D_k$ such that 
\begin{equation}
C^1 \cdot D_k = C^0 \cdot D_k = 1, \quad \text{ and } \quad C^1 \cdot D_j = C^0 \cdot D_j = 0\quad  \text{ for } j \neq k.
\end{equation}
This unique component is guaranteed to exist by Riemann-Roch, Serre duality and the genus formula, as explained in \cite[Chapter I, Section 0]{looijenga}.
\item Denote the points of intersection of these curves with $D_k$ as $D_k \cap C^1$ and $D_k \cap C^0$. The value of $\chi(\alpha)$ is then computed using the residue formula (see \cite[Chapter I, Section 5]{looijenga} for details) as
\begin{equation}
\chi(\alpha) = 2 \pi i \int_{D_k \cap C^0}^{D_k \cap C^1} \operatorname{Res}_{D_k} \omega. \label{residueintegral}
\end{equation}

\end{itemize}
\begin{remark} \label{rem:coefficientsrootvars}
It is important to note that the integral \eqref{residueintegral} will be computed in coordinates and therefore detects parameters in the mapping, which is fundamental to understanding the mechanism behind the reflection of dynamics on $\Pic(\X)$ in confinement conditions.
In particular, if $X$ is obtained as part of the construction of a space of initial conditions for some mapping, the locations of centres of blow-ups giving rise to $C^1$ and $C^0$ will be expressed in coordinates in terms of coefficients in the mapping. 
Then the value $\chi(\alpha)$ will in general depend on these coefficients via the integration limits $D_k\cap C^1$ and $D_k \cap C^0$, which relates coefficients in the mapping to values of the period map.
This will be illustrated in the examples in Sections \ref{section5} and \ref{section6}.
\end{remark}

\subsubsection{Dynamics of the values of the period map}
The crucial property of the period map for our purposes is how it interacts with the linear dynamics on $\Pic(\X)$. 
Let $\varphi_n$ be a non-autonomous mapping with space of initial conditions formed of surfaces $X_n$ and let $\Phi$ be its action on the Picard lattice $\Pic(\X)$.
Take some linearly independent subset $\left\{\alpha_i \right\} \subset Q^{\perp} \subset \Pic(\X)$ whose span is closed under $\Phi$, and let $M_{i,j}$ be the matrix representation of $\Phi$ with respect to this, i.e.
\begin{equation}
\Phi(\alpha_i) = \sum_{j} M_{i,j} \alpha_j.\label{Phirepresentation}
\end{equation}
For each $n$, take $\omega_n$ to be a rational 2-form on $X_n$ 
chosen such that $\tilde{\varphi}_n^*(\omega_{n+1}) = \omega_n$.
Then denote by $\chi_n$ and $\chi_{n+1}$ the period maps defined by $\omega_n$ and $\omega_{n+1}$ on $X_n$ and $X_{n+1}$ respectively, and the values of the period map of $X_n$ by
\begin{equation}
a_i(n) = \chi_n\left(\iota_n^{-1} (\alpha_i)\right).
\end{equation}
The correspondence between the action of the mapping on $\Pic(\X)$ and parameter evolution $a_i(n)\mapsto a_i(n+1)$ is provided by the following crucial lemma, which is already implicit in the original paper of Sakai \cite{SAKAI2001} (see also \cite[Remark 3.1]{dzhamaytakenawa}).

\begin{lemma} \label{rootvarsparamevolution}
\begin{equation}
\begin{aligned}
a_i(n+1) = \sum_{j} M_{i,j} a_j(n) \label{arepresentation}.
\end{aligned}
\end{equation}
\end{lemma}
\begin{proof}
Since 
$\chi_{n+1}\left( \iota_{n+1}^{-1}(\alpha_i) \right) 
= \chi_{n} \left( \tilde{\varphi}_n^* \circ \iota_{n+1}^{-1} (\alpha_i) \right)
= \chi_n \left(\iota_n^{-1} \circ \Phi (\alpha_i) \right) 
= \sum_{j} M_{i,j} \chi_n\left(\iota_n^{-1} (\alpha_i)\right)$.
The second equality comes from the fact that $\tilde{\varphi}_n^*(\omega_{n+1}) = \omega_n$, and that $\tilde{\varphi}_n$ is biholomorphic and acts as a change of variables for the integrals.
\end{proof}
The fact that the representation in \eqref{arepresentation} is given by exactly the same matrix as in \eqref{Phirepresentation}
 provides the sought after bridge between parameter evolution and dynamics on $\Pic(\X)$.

\subsection{The question of existence of an effective anticanonical divisor}

For Sakai surfaces associated with discrete Painlev\'e equations the existence of an effective anticanonical divisor is built into the definition.
However, to carry out the period map construction in the non-integrable case, we must deal with the question of the existence of an effective anticanonical divisor on the surface $X_n$ from a space of initial conditions for a mapping of type (c) in Proposition \ref{classificationpropnonauto}.

We will need the following fact relating to a space of initial conditions for a non-integrable mapping, the proof of which is found in \cite{MASE}.

\begin{lemma}[{\cite[Lemma 4.15]{MASE}}] \label{lem:negdef}
For a mapping of type (c) in Proposition \ref{classificationpropnonauto}, denote by $v \in \Pic_{\R}(\X)$ the dominant eigenvector corresponding to the eigenvalue $\lambda > 1$ of $\Phi$. 
Then the intersection product is negative definite on $v^{\perp} \cap \Pic(\X)$.
\end{lemma}
Using techniques similar to those in \cite{dillerlin} and \cite{MASE}, we have the following, 

\begin{theorem} \label{effectivenesstheorem}
Let $\varphi_n : \p^2 \dashrightarrow \p^2$ be a non-autonomous, non-integrable mapping which has a space of initial conditions, i.e. a mapping of type (c) in Proposition \ref{classificationpropnonauto}.
Suppose there exists a sequence of rational 2-forms $\omega_n$ on $\p^2$ and constants $c_n \in \C^*$ such that 
\begin{equation} \label{symplecticformpreservedP2}
\varphi_n^*\, \omega_{n+1} = c_n\, \omega_n.
\end{equation}
Denoting the lift of $\omega_n$ to $X_n$ by 
$
\tilde{\omega}_n = (\rho_{n} \circ \sigma_{n}^{-1} )^* \omega = (\sigma_{n})_* \circ (\rho_{n})^* \omega,
$
with its divisor having irreducible decomposition
\begin{equation}
\div(\tilde{\omega}_n) = \sum_{i} m_i^{(n)} D_i^{(n)}, 
\end{equation}
then if under the identification of all $\Pic(X_n)$ into the Picard lattice $\Pic(\X)$ we have, for all $n$,
\begin{equation} \label{D_iidentificationcondition}
\iota_n^{-1} (D_i^{(n)}) = \iota_{n+1}^{-1} (D_i^{(n+1)}),
\end{equation}
 there exists a minimal space of initial conditions formed by surfaces which have effective anticanonical divisors. 
\end{theorem}
\begin{remark}
In particular, if an autonomous mapping $\psi : \p^2 \dashrightarrow \p^2$ is such that $\psi^* \omega = c \,\omega$ for some rational 2-form $\omega$ on $\p^2$ and some $c \in \C^*$, then any deautonomisation of it satisfies the conditions in Theorem \ref{effectivenesstheorem}. 
\end{remark}
\begin{remark}
Note that even before condition \eqref{D_iidentificationcondition} is imposed, the fact $\varphi_n$ preserves $\omega_n$ already implies that $\iota_n^{-1}(D_i^{(n)}) = \iota_{n+1}^{-1}(D_{\ell(i)}^{(n+1)})$ for some permutation $\ell$ that preserves the multiplicities $m_i^{(n)}$ of $D_i^{(n)}$ in $\div(\tilde{\omega}_n)$.
Condition \eqref{D_iidentificationcondition} is required to remove this ambiguity. 
\end{remark}
\begin{proof}[Proof of Theorem \ref{effectivenesstheorem}]

Let the space of initial conditions for the mapping $\varphi_n$ be as in Definition \ref{spaceofICsdefblowdowns}. 
Let the decomposition of $\div(\tilde{\omega}_n)$ on $X_n$ into effective and anti-effective parts be
\begin{equation}
\div(\tilde{\omega}_n) = D_n^{+} - D_n^{-},
\end{equation}
where $D^{+}$ and $D^{-}$ are effective divisors. Note that $- \div(\tilde{\omega}_n)$ being effective is equivalent to $D_n^{+} = 0$. 

The idea now is to show that if $-\div(\tilde{\omega}_n)$ is not effective then the space of initial conditions is not minimal and a nontrivial contraction is possible. 
Indeed suppose $D_n^{+}\neq 0$ and let its decomposition into irreducible components be 
\begin{equation}
D_n^{+} = \sum_{i} m_i C_i^{(n)},
\end{equation}
where $m_i > 0$ and $C_i^{(n)}$ are irreducible curves on $X_n$. 
Denote the classes of these by $[C_i^{(n)}] \in \Pic(X_n)$. 
Then according to our assumptions, the expressions for $[C_i^{(n)}]$ in terms of the geometric basis for $\Pic(X_n)$ corresponding to $\pi_n$ do not depend on $n$ and we have the following well-defined elements of $\Pic(\X)$:
\begin{equation}
\mathcal{D}_i^{+} = \iota_n^{-1}( [C_i^{(n)}] ), \qquad \mathcal{D}^{+} = \sum_i m_i \iota_n^{-1} ([D]^{+}_i ).
\end{equation}
We will show that there exists some $\mathcal{D}_{i}^{+}$ that corresponds to the class of an exceptional curve of the first kind and which is periodic under $\Phi$, and that its orbit gives a collection of classes of exceptional curves of the first kind which are mutually disjoint, so by Lemma \ref{minimisationlemma} the space of initial conditions is not minimal.

The fact that $\varphi_n^*$ preserves the 2-form $\omega_n$ in the sense of \eqref{symplecticformpreservedP2} implies that $\tilde{\varphi}_{n}^* \tilde{\omega}_{n+1}= c_n \,\tilde{\omega}_n$, and $\tilde{\varphi}_n^*\left( \div( \tilde{\omega}_{n+1})\right) = \div( \tilde{\omega}_n)$.
Further, since $\tilde{\varphi}_n$ is an isomorphism, for any irreducible curve $C$ on $X_n$ we have
\begin{equation}
\operatorname{ord} ( C, \div(\tilde{\omega}_n)) = \operatorname{ord} ( \tilde{\varphi}_n(C), \div(\tilde{\omega}_{n+1})),
\end{equation}
which allows us to deduce that $\Phi$ must permute the $\mathcal{D}^{+}_i$ in a way that preserves their multiplicity in the canonical class.
In particular there exists a permutation $\sigma$ such that $\Phi(\mathcal{D}_i^+) = \mathcal{D}_{\sigma(i)}^+$, so all $\mathcal{D}_i^+$ are in the part of $\Pic(\X)$ periodic under $\Phi$.
Denoting by $\ell$ the order of the permutation $\sigma$, we have
\begin{equation}
\mathcal{D}_{i}^{+} \cdot v =  \Phi^{\ell} (\mathcal{D}_{i}^{+}) \cdot \Phi^{\ell} v = \lambda^{\ell} \left( \mathcal{D}_{i}^{+} \cdot v \right),
\end{equation}
and since $\lambda>1$ we deduce that $\mathcal{D}_i^{+} \in v^{\perp} \cap \Pic(\X)$ for every $i$, so by Lemma \ref{lem:negdef} the intersection form is negative definite on the sublattice $\operatorname{span}_{\Z}\left\{ \mathcal{D}_i^{+} \right\} \subset \Pic(\X)$.
This allows us to deduce that at least one of the $\mathcal{D}_i^{+}$ must correspond to an exceptional curve of the first kind, as follows. 
Taking the intersection of $\mathcal{D}^+$ with the canonical class in $\Pic(\X)$ we have 
\begin{equation*}
\mathcal{D}^+ \cdot \K_{\X}	
						= \mathcal{D}^{+} \cdot \mathcal{D}^{+} - \mathcal{D}^{+} \cdot \mathcal{D}^{-}
						\leq \mathcal{D}^{+} \cdot \mathcal{D}^+  < 0,
\end{equation*}
where the last inequality comes from the intersection form being negative definite on $\operatorname{span}_{\Z}\left\{ \mathcal{D}_i^{+} \right\}$.
Since $\mathcal{D}^+$ corresponds to an effective class we then have that some $\mathcal{D}^+_i$ satisfies $\mathcal{D}_i^+ \cdot \K_{X} < 0$. 
Since $\mathcal{D}_i^+ \in v^{\perp} \cap \Pic(X)$ we also deduce $\mathcal{D}_i^+ \cdot \mathcal{D}_i^+ < 0$.
The genus formula then allows us to deduce that $\mathcal{D}_i^+$ must correspond to an exceptional curve of the first kind.

It remains to be shown that the elements of $\Pic(\X)$ in the orbit of this $\mathcal{D}_i^+$ under $\Phi$ are mutually orthogonal. 
Letting $k \in \Z$ be such that $\Phi^k \mathcal{D}_i^+ = \mathcal{D}_j^+ \neq \mathcal{D}_i^+$, we again use the negative definiteness of the intersection form on $\operatorname{span}_{\Z}\left\{ \mathcal{D}_i^{+} \right\}$ to show that
\begin{equation}
0 > \left(\mathcal{D}_i^+ + \Phi^k \mathcal{D}_i^+ \right)^2 = - 2 + 2 \mathcal{D}_i^+ \cdot \Phi^k \mathcal{D}_i^+ 
\quad \implies \quad
1 > \mathcal{D}_i^+ \cdot \Phi^k \mathcal{D}_i^+ \geq 0,
\end{equation}
so $\mathcal{D}_i^+ \cdot \Phi^k \mathcal{D}_i^+ = 0$ and the orbit of $\mathcal{D}_i^+$ gives a collection of mutually disjoint exceptional curves of the first kind which are permuted by the mapping.

Finally we note that after a contraction, performed by blowing down these curves, we still have a space of initial conditions for the same $\varphi_n$ in the sense of Definition \ref{spaceofICsdefblowdowns}, and we can apply the above procedure as many times as is necessary to obtain a minimal space of initial conditions formed of anticanonical surfaces.
Note that the procedure is guaranteed to terminate since a space of initial conditions for a non-integrable mapping must have $\operatorname{rank} \Pic(\X) >10$ according to Proposition \ref{picardnumberprop}.

\end{proof}

\subsection{Bases for the orthogonal sublattice} \label{subsec:rootbases}

The next problem which must be addressed is ensuring that for surfaces associated with non-integrable mappings, one can find appropriate elements of $Q^{\perp}$ on which to compute the period map.
For the remainder of this Section we assume all surfaces have effective anticanonical divisors.

\begin{definition}[root basis of $Q^{\perp}$] \label{rootbasisdef}
Suppose a $\Z$-basis $\{\alpha_i\}$ for $Q^{\perp}$ is such that each $\alpha_i$ can be expressed as the difference of two classes of exceptional curves of the first kind $\alpha_i = [C^1_i] - [C^0_i]$ and the numbers $c_{i,j} = 2 \frac{\alpha_i \cdot \alpha_j}{\alpha_i \cdot \alpha_i}$ form the entries of a generalised Cartan matrix as defined in \cite{KAC}. 
That is, the diagonal entries $c_{i,i} = 2$ for all $i$, the off-diagonal entries $c_{i,j}$, $i\neq j$ are nonpositive integers, and $a_{i,j}=0$ implies $a_{j,i}=0$.
Then we call this a root basis for $Q^{\perp}$. 
Following \cite{SAKAI2001}, we will refer to the values of the period map on a root basis as root variables.
\end{definition}

For each type of surface in his classification, Sakai \cite{SAKAI2001} found a root basis with a corresponding generalised Cartan matrix that is of affine type.
In \cite{takenawaDDS}, Takenawa constructed the space of initial conditions for a non-integrable mapping (the Hietarinta-Viallet mapping \cite{HVequation}), which has an effective anticanonical divisor. 
Takenawa was also able to find a root basis in this case, with generalised Cartan matrix
\begin{equation}
\left(\begin{array}{ccc}2 & -2 & -2 \\-2 & 2 & -2 \\-2 & -2 & 2\end{array}\right),
\end{equation}
which is of hyperbolic type $H_{71}^{(3)}$ \cite{WAN}.
Takenawa also constructed an example with root basis corresponding to the matrix
\begin{equation}
\left(\begin{array}{ccc}2 & -3 & -3 \\-3 & 2 & -3 \\-3 & -3 & 2\end{array}\right),
\end{equation}
which is of neither finite, affine, nor hyperbolic type.

For non-integrable mappings it is not immediately clear that we can always find a root basis, though for our purposes it will be sufficient to show that one can always construct a linearly independent subset of $Q^{\perp}$ that is closed under $\Phi$ and whose span includes the dominant eigenvector.
To this end we define the following.

\begin{definition}[sufficient subset of $Q^{\perp}$] \label{sufficientsubset} 
For a space of initial conditions for a mapping with unbounded degree growth, suppose a linearly independent subset $\{\beta_j\}$ of $Q^{\perp}$ is such that the following two conditions hold:
\begin{itemize} 
\item Each $\beta_j$ can be expressed as the difference of two classes of exceptional curves of the first kind $C_j^{1}$ and $C_j^{0}$, i.e. $\iota_n(\beta_j) = [C^1_j] - [C^0_j]$,
\item $\Phi$ preserves $\operatorname{span}_{\Z} \{\beta_j\}$.
\end{itemize}
If additionally the following requirement is met, which depends on whether the mapping is of type (b) or (c) of Proposition \ref{classificationpropnonauto}, then we call $\{\beta_j\}$ a sufficient subset:
\begin{itemize}
\item type (c) : The dominant eigenvector $v \in \Pic_{\R}(\X)$ of $\Phi$ is contained in $\operatorname{span}_{\R} \{\beta_j\}$,
\item type (b) : Some element in $\Pic_{\R}(\X)$ which is a generalised eigenvector for the eigenvalue $1$ of $\Phi$, but not an eigenvector,
is contained in $\operatorname{span}_{\R} \{\beta_j\}$.
\end{itemize}
\end{definition}

Note that a root basis is automatically a sufficient subset. For non-autonomous mappings of type (b) of Proposition \ref{classificationpropnonauto}, it was shown in \cite{MASE} that they correspond to Sakai surfaces, for which we always have a sufficient subset in $Q^\perp$, for example the root bases constructed by Sakai for each type in the list \cite{SAKAI2001}.
However, in order to show that we can always find a sufficient subset for mappings of type (c) in Proposition \ref{classificationpropnonauto} we will require the following Lemma consisting of several results from \cite{SCasIC}. 

\begin{lemma}[{\cite{SCasIC}}] \label{SCasIClemma} 
Denote the periodic part of $\Pic(\X)$ under $\Phi$ by
\begin{equation}
\operatorname{P}_{\Phi} = \left\{ \F \in \Pic(\X) ~|~ {\exists}~  m \in \Z\setminus\{0\} ~\text{ such that }~ \Phi^m (\F) = \F\right\},
\end{equation}
and let $\overline{\mu}_{\Phi}(t)\in \Q[t]$ denote the minimal polynomial of $\Phi$ as a linear transformation of $\Pic(\X) / \operatorname{P}_{\Phi}$.
If there exists $F \in \Pic(\X) / \operatorname{P}_{\Phi}$, $F \neq 0$ and a polynomial $\psi(t) \in \Q[t]$ such that
\begin{equation}
\psi[\Phi] F = 0 ~\operatorname{mod}~ \operatorname{P}_{\Phi},
\end{equation} 
then $\psi(t)$ shares a factor with $\overline{\mu}_{\Phi}(t)$.
For mappings of type (c) in Proposition \ref{classificationpropnonauto}, $\overline{\mu}_{\Phi}(t)$ is precisely the minimal polynomial of $\lambda$, the dominant eigenvalue of $\Phi$, over $\Q$. We denote this minimal polynomial as $\mu_{\lambda}(t)$.
\end{lemma}

\begin{proof}
The first part follows from the well-known fact in linear algebra that if for an endomorphism $f$ on a finite-dimensional vector space $V$ over a field $K$, a nonzero element $v \in V$, and a polynomial $\psi(t) \in K[t]$ it holds that $\psi(f) v = 0$, then $\psi(t)$ and the minimal polynomial of $f$ have a common factor.
By Proposition~\ref{classificationpropnonauto}, the minimal polynomial $\overline{\mu}_{\Phi}(t)$ coincides with $\mu_{\lambda}(t)$ in the case of type (c), which leads to the second statement.
\end{proof}

We will also require the following result, which follows from the classification of the anticanonical divisors for spaces of initial conditions for non-integrable mappings (Theorem \ref{theorem:classification}) which is given in Section \ref{section4}. 

\begin{proposition} \label{propositionA1}
	Consider a mapping $\varphi_n$ which is non-integrable, i.e. of type (c) in Proposition \ref{classificationpropnonauto}, with a minimal space of initial conditions provided by surfaces $X_n$ which are anticanonical, i.e. $X_n$ has an effective anticanonical divisor $D = \sum_j m_j D_j$.
	Then, there exists $j$ such that
	\begin{itemize}
		\item
		$m_j = 1$,
		\item
		$X_n$ has an exceptional curve of the first kind $C$ such that $C \cdot D_i = \delta_{i, j}$.

	\end{itemize}
\end{proposition}

We are now ready to prove the following. 

\begin{theorem} \label{sufficientsubsetprop}
For a minimal space of initial conditions for a non-integrable mapping, i.e. of type (c) in Proposition \ref{classificationpropnonauto}, one can always find a sufficient subset of $Q^{\perp}$.
\end{theorem}

\begin{proof}
Let the surfaces forming the space of initial conditions be $X_n$, with $D = \sum_j m_j D_j$ the effective anticanonical divisor of $X_n$. 
Then let $[C] \in \Pic(X_n)$ be the class of the exceptional curve $C$ of the first kind as in Proposition \ref{propositionA1}, which intersects exactly one component $D_j$ of the anticanonical divisor $D = \sum_j m_j D_j$.
The element of $\mathcal{D}_k = \iota_n^{-1} [D_k] \in \Pic(\X)$ corresponding to this class must lie in $\operatorname{P}_{\Phi}$, so let $\ell \in \Z_{\geq 0}$ be the smallest nonnegative integer such that $\Phi^{\ell} (\mathcal{D}_k) =\mathcal{D}_k$.
Furthermore, let $\mathcal{C} = \iota_n^{-1} [C]$ be the element of $\Pic(\X)$ corresponding to the exceptional curve of the first kind $C$, which cannot lie in $\operatorname{P}_{\Phi}$ because of the minimality of the space of initial conditions.
Then we can construct a sufficient subset as follows.
First let 
\begin{equation}
\beta_0 =  (\Phi^{\ell} - 1) \mathcal{C},
\end{equation}
which by construction corresponds to the difference of two classes of exceptional curves of the first kind. 
Then let $\beta_1 = \Phi \beta_0$, $\beta_2 = \Phi \beta_1$ and so on until the process terminates and some $\beta_{h+1}$ is a linear combination of $\beta_{0}, \dots , \beta_{h}$, say $\Phi \beta_h + \sum_{i=0}^{h} a_i \beta_i = 0$.
Then $\{ \beta_0, \dots, \beta_{h} \}$ by construction satisfies the first two conditions in Definition \ref{sufficientsubset}, so it remains only to show that $v \in \operatorname{span}_{\R} \{\beta_0, \dots , \beta_h\}$.
From the construction of this subset we have that
\begin{equation}
\psi[\Phi](\Phi^{\ell} - 1) \mathcal{C} = 0,
\end{equation}
where $\psi(t) = (t^{h+1} + \sum_{i=0}^{h} a_i t^i) \in \Q[t]$.
Since we are in the third case of Proposition \ref{classificationpropnonauto}, Lemma \ref{SCasIClemma} implies that $\psi$ has $\mu_{\lambda}(t)$, the minimal polynomial over $\Q$ of the dominant eigenvalue of $\Phi$, as a factor.
Therefore the characteristic polynomial of the restriction of $\Phi$ to $\operatorname{span}_{\R}\left\{\beta_{0}, \dots , \beta_{h}\right\}$ has $\mu_{\lambda}(t)$ as a factor, so the dominant eigenvector $v$ is contained in $\operatorname{span}_{\R}\left\{\beta_{0}, \dots , \beta_{h}\right\}$ as required.

\end{proof}

\subsection{Sufficient deautonomisation}

We are now ready to apply the above results to describe the mechanism by which the dynamical degree is reflected in parameter evolution for a non-autonomous mapping which has a space of initial conditions and preserves a 2-form.

\begin{definition}[sufficient deautonomisation] \label{sufficientdeautonomisationdef}
Consider a deautonomisation $\varphi_n$ of an autonomous mapping $\psi$ in the sense of Definition \ref{deautodef}, with space of initial conditions given by surfaces $X_n$ with effective anticanonical divisors. 
Further assume that the mapping has unbounded degree growth.
If a sufficient subset $\{ \beta_1, \dots , \beta_r \}$ exists, consider the $\Z$-linear map 
\begin{equation}
\begin{gathered}
\boldsymbol{\chi} : \operatorname{span}_{\Z}\{ \beta_1, \dots , \beta_r \} \rightarrow \C~ \operatorname{mod} \hat{\chi}(H_1 (D_{\operatorname{red}} ; \Z)), \\
\beta \mapsto \chi_{X_n}\left( \iota_n (\beta)\right),
\end{gathered}
\end{equation}
where $\chi_{X_n}$ is the period map of $X_n$. 
If $\boldsymbol{\chi}$ is injective, then we call the deautonomisation sufficient. 
\end{definition}
If a deautonomisation is sufficient in the sense of Definition \ref{sufficientdeautonomisationdef}, this means that the evolution of values of the period map on the sufficient subset is guaranteed to give the same characteristic polynomial as that of $\Phi$ restricted to $\operatorname{span}_{\Z}\{\beta_1,\dots,\beta_r \}$ and we have Theorem \ref{bigsummarytheorem} and Corollary \ref{bigsummarycorollary}.
In the case of type (c) mappings this characteristic polynomial will have $(t-\lambda)$ as a factor (where $\lambda$ is the dominant eigenvalue of $\Phi$), whereas for type (b) mappings it has $(t-1)^2$ as a factor.
As noted in Remark \ref{rem:coefficientsrootvars}, the values of the period map are related to the coefficients in the rational functions giving the mapping in coordinates.

The construction of a sufficient deautonomisation depends significantly on the form in coordinates of the autonomous mapping. 
We give sufficient deautonomisations for two families of examples in Sections \ref{section5} and \ref{section6}, as well as for three more in a follow-up paper \cite{part2}. There we also give procedures for constructing sufficient deautonomisations for mappings that take the same form as scalar three-point mappings belonging to classes I-VI in the QRT-type classification.
The idea of the construction is to ensure that the coefficients in the mapping which parametrise locations of the centres of the final blow-ups performed in the construction of $X_n$ are allowed to depend on $n$.
We conjecture that sufficient deautonomisations always exists for a non-integrable autonomous mapping which has a space of initial conditions and preserves a 2-form.
\begin{conjecture}
Let $\psi$ be an autonomous mapping with a space of initial conditions which is of type (c) in Proposition \ref{classificationpropnonauto} and which preserves a 2-form $\omega$ on $\p^2$.
Then there exists a sufficient deautonomisation of $\psi$.
\end{conjecture}

\section{Anticanonical divisors for non-integrable mappings with spaces of initial conditions}  \label{section4}

%
%
%

In this Section we provide a classification of anticanonical divisors (Theorem~\ref{theorem:classification}) for spaces of initial conditions that correspond to non-integrable mappings, i.e. mappings with dynamical degree greater than 1.
As mentioned in \cite{MASE}, the theory of spaces of initial conditions for mappings of the plane allows for two description that are, however, essentially the same.
While the description introduced in the main part of this paper is more intuitive, the other one, by a rational surface and a Cremona isometry, is more compatible with the general theory needed in such a classification.
Therefore, in this section, we shall use the description by Cremona isometries.
The following are our assumptions and notations that we shall use in this section:
\begin{itemize}
	\item
	$X$ is a rational surface.

	\item
	$\sigma$ is a Cremona isometry.
	That is,
	$\sigma$ is a $\mathbb{Z}$-linear transformation of $\Pic(X)$ that satisfies
	\begin{itemize}
		\item
		$\sigma \K_X = \K_X$,
		
		\item
		$(\sigma F_1) \cdot (\sigma F_2) = F_1 \cdot F_2$ for any $F_1$, $F_2 \in \Pic(X)$,

		\item
		if $F \in \Pic(X)$ is effective, then so is $\sigma F$.

	\end{itemize}

	\item
	$\lambda > 1$:
	the maximum eigenvalue of $\sigma$, which is irrational.

	\item
	$v \in \Pic_{\R}(X)$:
	the dominant eigenvector of $\sigma$, which we may assume is nef (see \cite{dillerfavre,MASE}).

	\item
	$X$ is minimal as a space of initial conditions.

	\item
	The anticanonical class $- \K_X$ is effective.
	That is, we have
	\[
		D = \sum^N_{j = 1} m_j D_j \in |- \K_X|,
	\]
	where $D_j \subset X$ is irreducible, $m_j \in \mathbb{Z}_{> 0}$ and $D_i \ne D_j$ ($i \ne j$).

	\item
	$D_{\operatorname{red}} = D_1 \cup \cdots \cup D_N$.

\end{itemize}

\begin{lemma}
	$D$ is connected.
\end{lemma}

\begin{proof}
	The proof of \cite[Proposition~6]{SAKAI2001} is still valid in our case.
\end{proof}

\begin{lemma}\label{lemma:genus}
	$g_a(D_j) \le 1$ for any $j$.
	Moreover, if $D$ has a component $D_j$ such that $g_a(D_j) \ge 1$, then $D = D_j$.
	In particular, if $D$ is not irreducible, then each component of $D$ is a smooth rational curve.
\end{lemma}

\begin{proof}
	Take a geometric basis $\mathcal{H}, \mathcal{E}_1, \ldots, \mathcal{E}_K \in \Pic(X)$ and let $\pi \colon X \to \mathbb{P}^2$ be the corresponding blow-down.
	Using this basis, each $[D_j]$ is expressed as
	\[
		[D_j] = c_{j 0} \mathcal{H} + c_{j 1} \mathcal{E}_1 + \cdots + c_{j K} \mathcal{E}_N.
	\]
	Here, $c_{j 0} \ge 0$ since $D_j$ is effective.
	Moreover, since $- \K_X = 3 \mathcal{H} - \mathcal{E}_1 - \cdots - \mathcal{E}_K$, we have
	\[
		3 = \sum_j m_j c_{j 0},
	\]
	which implies $c_{j 0} \le 3$.

	From here on, we fix $j$.
	If $c_{j 0} = 0$, then $D_j$ is the strict transform of some exceptional curve of $\pi$.
	Therefore, $D_j$ is rational and $g_a(D_j) = 0$.
	Let us consider the remaining cases, i.e.\ $c_{j 0} = 1, 2, 3$.
	In these cases, $C := \pi(D_j) \subset \mathbb{P}^2$ is a curve in $\mathbb{P}^2$ and its class is $[C] = \mathcal{O}_{\mathbb{P}^2}(c_{j 0})$.
	Therefore, by the genus formula on $\mathbb{P}^2$, we have $g_a(C) = \frac{1}{2} (c_{j 0}-2) (c_{j 0} - 1) \le 1$.
	Since taking the strict transform of a curve with respect to a birational morphism does not increase its arithmetic genus, we conclude that $g_a(D_j) \le g_a(C) \leq 1$.

	If $g_a(D_j) = 1$, then $c_{j 0} = 3$ and $m_j = 1$.
	In this case, by the genus formula, we have
	\[
		D_j \cdot \sum_{i \ne j} m_i D_i = D_j \cdot (- \K_X - [D_j]) = 0,
	\]
	which implies $D = D_j$ since $D$ is connected.
\end{proof}

\begin{lemma}\label{lemma:irreducible}
	If $D$ is irreducible, then $g_a(D) = 1$.
\end{lemma}

\begin{proof}
	Since $- \K_X - [D] = 0$ in this case, we have $g_a(D) = 1$ by the genus formula.
\end{proof}

\begin{lemma}\label{lemma:genus_simple}
	If $D$ is not irreducible and $m_j = 1$, then $D_j \cdot \sum_{i \ne j} m_i D_i = 2$.
\end{lemma}

\begin{proof}
	In this case, $g_a(D_j) = 0$ since $D$ is not irreducible.
	By the genus formula, we have
	\[
		2 = D_j \cdot (- \K_X - [D_j]) = D_j \cdot \sum_{i \ne j} m_i D_i.
	\]
\end{proof}

\begin{lemma}
	$D_j$ is orthogonal to $v$ for all $j$.
	Moreover, there exists $\ell \in \mathbb{Z}_{> 0}$ such that $\sigma^{\ell} [D_j] = [D_j]$ for all $j$.
\end{lemma}

\begin{proof}
	Since $v$ is nef and $D_j$ is effective, we have
	\[
		v \cdot [D_j] \ge 0.
	\]
	Multiplying both sides by $m_j$ and taking the sum over $j$, we have
	\[
		v \cdot (- \K_X) \ge 0.
	\]
	However, Since $v \cdot (- \K_X) = 0$, $v \cdot D_j$ must be $0$ for all $j$.
	Therefore, $D_j \in v^{\perp} \cap \Pic(X)$.
	Since $\sigma$ preserves the intersection form on the lattice $v^{\perp} \cap \Pic(X)$, which by \cite[Lemma 4.15]{MASE} is negative definite, there exists $\ell > 0$ such that $\left( \sigma \Big|_{v^{\perp} \cap \Pic(X)} \right)^{\ell} = \operatorname{id}_{v^{\perp} \cap \Pic(X)}$.
\end{proof}

\begin{lemma}
	$D^2_j < 0$.
\end{lemma}

\begin{proof}
	Clear from $D_j \in v^{\perp} \cap \Pic(X)$.
\end{proof}

\begin{lemma}
	$D_j$ is not an exceptional curve of the first kind.
	In particular, $D^2_j \le - 2$.
\end{lemma}

\begin{proof}
	Since $X$ is minimal, there does not exist an exceptional curve of the first kind $C$ such that $C \cdot v = 0$ \cite[Proposition~4.17]{MASE}.
	Since $D_j$ is orthogonal to $v$, $D_j$ is not an exceptional curve of the first kind.
\end{proof}

\begin{lemma}\label{lemma:negative_definite}
	For any nontrivial linear combination $F = \sum_j c_j [D_j]$ with $c_j \in \mathbb{R}$, we have $F^2 < 0$.
\end{lemma}

\begin{proof}
	The intersection form is negative definite on the lattice $v^{\perp} \cap \Pic(X)$.
	Therefore, the form is negative definite on the real vector space $V := \left( v^{\perp} \cap \Pic(X) \right) \otimes \mathbb{R}$.
	Since $F \in V \setminus \{ 0 \}$, we have $F^2 < 0$.
\end{proof}

\begin{lemma}
	For any $a_j \in \mathbb{Z}_{\ge 0}$, we have
	\[
		h^0 \left( \sum_j a_j D_j \right) = 1,
	\]
	which implies that $D$ is the only effective divisor that belongs to the class $- \K_X$.
	In particular, $\sigma$ acts on the set $\{ D_1, \ldots, D_N \}$ as a permutation.
\end{lemma}

\begin{proof}
	By Lemma~\ref{lemma:negative_definite}, the matrix $\left( D_i \cdot D_j \right)_{i, j}$ is negative definite.
	Hence, we have $h^0 \left( \sum_j a_j D_j \right) = 1$ \cite[Proposition~A.9]{MASE}.
\end{proof}

\begin{lemma}
	Let $E \subset X$ be an exceptional curve of the first kind.
	Then, there exists $i$ such that
	$m_j D_j \cdot E = \delta_{i, j}$.
	In particular, $D$ has a component of multiplicity $1$.
\end{lemma}

\begin{proof}
	By the genus formula, we have $E \cdot (- \K_X) = 1$.
	Since
	\[
		1 = \sum_i m_i D_i \cdot E
	\]
	and $m_i D_i \cdot E \ge 0$, there exists $i$ such that
	$m_j D_j \cdot E = \delta_{i, j}$.
\end{proof}

\begin{lemma}\label{lemma:classification1}
	Suppose $j,k$ are such that $j\neq k$, $m_j = 1$ and $D_j \cdot D_k > 0$.
	Then, only one of the following holds:
	\begin{enumerate}
		\item
		$D_j \cdot D_k = 1$,
		$m_k = 1$.

		\item
		$D_j \cdot D_k = 1$,
		$m_k = 2$.

		\item
		$D_j \cdot D_k = 2$,
		$m_k = 1$.

	\end{enumerate}
\end{lemma}

\begin{proof}
	Clear from
	\[
		0 < m_k D_j \cdot D_k
		\le D_j \cdot \sum_{i \ne j} m_i D_i
		= 2,
	\]
	which follows from Lemma \ref{lemma:genus_simple}.
\end{proof}

\begin{lemma}\label{lemma:classification_a1}
	If $D_1 \cdot D_2 = 2$ and $m_1 = m_2 = 1$, then $D = D_1 + D_2$.
	In this case, there are only two possible configurations of $D = D_1 + D_2$:
	\begin{enumerate}
		\item
		$D_1$ and $D_2$ intersect at a point with multiplicity $2$.

		\item\label{enum:classification_a1q}
		$D_1$ and $D_2$ intersect transversely at two points.

	\end{enumerate}
\end{lemma}

\begin{proof}
	By Lemma~\ref{lemma:classification1}, we have
	\[
		2 = D_1 \cdot D_2
		\le \sum_{i \ne 1} m_i D_1 \cdot D_i
		= 2,
	\]
	which implies that $D_2 \cdot D_i = 0$ for all $i \ne 1, 2$.
	Replacing $D_2$ with $D_1$, we have $D_1 \cdot D_i = 0$ for all $i \ne 1, 2$.
	Therefore, $D_1 + D_2$ is a connected component of $D$.
	Since $D$ is connected, $D$ must coincide with $D_1 + D_2$.
\end{proof}

\begin{lemma}\label{lemma:qclassification}
	If $D_1 \cdot D_2 = 1$ and $m_1 = m_2 = 1$, then only the following types of configurations are possible:
	\begin{enumerate}
		\item\label{enum:notcycle}
		$D = D_1 + D_2 + D_3$ and the configuration has only one intersection:
		\begin{center}
			\begin{picture}(200, 200)
				\put(20, 180){\line(1, - 1){160}}
				\put(0, 100){\line(1, 0){200}}
				\put(20, 20){\line(1, 1){160}}
				\put(0, 180){$D_1$}
				\put(0, 80){$D_2$}
				\put(0, 20){$D_3$}
			\end{picture}
		\end{center}
		
		\item\label{enum:cycle}
		$N \ge 3$ and $D = D_1 + \cdots + D_N$ forms a cycle of $\mathbb{P}^1$ (after reordering the indices):
		\begin{center}
			\begin{picture}(200, 200)
				\put(40, 160){\line(1, 0){120}}
				\put(80, 180){\line(-1, - 1){80}}
				\put(120, 180){\line(1, - 1){80}}
				\put(20, 140){\line(0, -1){60}}
				\put(180, 140){\line(0, -1){60}}
	
				\multiput(25, 75)(10, -10){4}{\circle*{2}}
				\multiput(175, 75)(-10, -10){4}{\circle*{2}}
				\multiput(67, 35)(16, 0){5}{\circle*{2}}
	
				\put(90, 145){$D_1$}
				\put(50, 130){$D_N$}
				\put(25, 100){$D_{N - 1}$}
				\put(135, 130){$D_2$}
				\put(160, 100){$D_3$}
			\end{picture}
		\end{center}

	\end{enumerate}
\end{lemma}

\begin{proof}
	By Lemma~\ref{lemma:classification1}, we have
	\[
		2 = D_2 \cdot \sum_{i \ne 2} m_i D_i
		= D_2 \cdot \left( D_1 + \sum_{i \ne 1, 2} m_i D_i \right)
		= 1 + D_2 \cdot \sum_{i \ne 1, 2} m_i D_i,
	\]
	which implies that $D$ has only one component $D_{\ell}$ other than $D_1, D_2$ such that $D_2 \cdot m_{\ell} D_{\ell} = 1$.
	In particular, $m_{\ell} = 1$ and $D_2 \cdot D_{\ell} = 1$.
	We may assume that $\ell = 3$.
	
	If $D_1, D_2, D_3$ intersect at one point, then none of $D_1, D_2, D_3$ intersects $D - D_1 - D_2 - D_3$.
	Since $D$ is connected, $D$ must be $D_1 + D_2 + D_3$ in this case.

	Let us consider the other possibilities, i.e. when\ $D_1 \cap D_2 \cap D_3 = \emptyset$.
	In this case, $D_2$ does not intersect any $D_j$ other than $D_1, D_3$.
	Using Lemma~\ref{lemma:classification1} for $j = 3$, there exists a component $D_{\ell}$ other than $D_2, D_3$ such that $D_3 \cdot D_{\ell} = 1$ and $m_{\ell} = 1$.
	If $\ell = 1$, then $D = D_1 + D_2 + D_3$.
	If not, we may assume that $\ell = 4$.

	Using Lemma~\ref{lemma:classification1} for $j = 4$, there exists a component $D_{\ell}$ other than $D_3, D_4$ such that $D_4 \cdot D_{\ell} = 1$ and $m_{\ell} = 1$.
	Repeating this procedure, one can find a sequence $D_1, D_2, \ldots, D_N$ that forms a cycle.
\end{proof}

\begin{lemma}\label{lemma:cycle}
	Suppose that $D$ has a cycle $D_{j_1}, \ldots, D_{j_{\ell}}$ for $\ell \ge 3$, i.e.
	\[
		D_{j_i} \cdot D_{j_k} \begin{cases}
			> 0 & (i = k \pm 1 \mod \ell) \\
			< 0 & (i = j) \\
			= 0 & (\text{otherwise})
		\end{cases}
	\]
	for $i, k = 1, \ldots, \ell$.
	Then, $D = D_{j_1} + \cdots + D_{j_{\ell}}$ and the configuration of $D$ must be one of case \eqref{enum:cycle} in Lemma~\ref{lemma:qclassification}.
	Moreover, if $\operatorname{rank} H_1(D_{\operatorname{red}}; \mathbb{Z}) \ge 1$ then the configuration of $D$ must be either \eqref{enum:classification_a1q} of Lemma~\ref{lemma:classification_a1} or one of \eqref{enum:cycle} in Lemma~\ref{lemma:qclassification}.
\end{lemma}

\begin{proof}
	If $m_{j_i} = 1$ for some $i$, then, by Lemma~\ref{lemma:qclassification}, the configuration of $D$ must be one of \eqref{enum:cycle} in Lemma~\ref{lemma:qclassification}.
	Let us assume that $m_{j_1}, \ldots, m_{j_{\ell}} \ge 2$ and deduce a contradiction.

	Take a geometric basis $\mathcal{H}, \mathcal{E}_1, \ldots, \mathcal{E}_K \in \Pic(X)$ and the corresponding blow-down $\pi \colon X \to \mathbb{P}^2$.
	Using this basis, each $[D_{j_i}]$ is expressed as
	\[
		[D_{j_i}] = c_{i 0} \mathcal{H} + c_{i 1} + \mathcal{E}_1 + \cdots + c_{i K} \mathcal{E}_K.
	\]
	Since $D_{j_i}$ is effective, $c_{i 0}$ must be nonnegative.
	Since
	\[
		[m_1 D_{j_1} + \cdots + m_{\ell} D_{j_{\ell}}]
		\le [D]
		= - \K_X
		= 3 \mathcal{H} - \mathcal{E}_1 - \cdots - \mathcal{E}_N,
	\]
	$m_1, \ldots, m_{\ell} \ge 2$ and
	\[
		[m_1 D_{j_1} + \cdots + m_{\ell} D_{j_{\ell}}]
		= (m_1 c_{1 0} + \cdots + m_{\ell} c_{\ell 0}) \mathcal{H} + \left( \text{$\mathcal{E}$-part} \right),
	\]
	at most one of $c_{1 0}, \ldots, c_{\ell 0}$ is $1$ and the others are $0$.

	Let us consider the case where $c_{1 0} = 1$, $c_{2 0} = \cdots = c_{\ell 0} = 0$.
	Let $H = \pi (D_{j_1})$, which is a line in $\mathbb{P}^2$.
	In this case, $\pi(D_{j_2}), \ldots, \pi(D_{j_{\ell}})$ are points on $H$ and any intersection among $D_{j_1}, \ldots, D_{j_{\ell}}$ is transversal.
	Therefore, when considered as a graph, $D_{j_1}, \ldots, D_{j_{\ell}} \subset X$ form a tree with $D_{j_1}$ at its root, which contradicts the choice of $D_{j_1}, \ldots, D_{j_{\ell}}$
	since if a tree graph consists of copies of $\mathbb{P}^1$ and every intersection has multiplicity $1$, then it does not have a cycle.
	
	If $c_{1 0} = \cdots = c_{\ell 0} = 0$, then $\pi(D_{j_1}), \ldots, \pi(D_{j_{\ell}})$ are the same point of $\mathbb{P}^2$.
	Therefore, $D_{j_1}, \ldots, D_{j_{\ell}}$ form a tree and any intersection has multiplicity $1$, which leads to a contradiction in the same way as above.
\end{proof}

\begin{lemma}\label{lemma:classification_double}
	If $D_1 \cdot D_2 \ge 2$, then $D_1 \cdot D_2 = 2$ and $m_1 = m_2 = 1$.
	In particular, the configuration of $D$ is one of those in Lemma~\ref{lemma:classification_a1}.
\end{lemma}

\begin{proof}
	Take a geometric basis $\mathcal{H}, \mathcal{E}_1, \ldots, \mathcal{E}_K \in \Pic(X)$ and the corresponding blow-down $\pi \colon X \to \mathbb{P}^2$.
	Using this basis, $[D_1]$ and $[D_2]$ are expressed as
	\[
		[D_i] = c_{i 0} \mathcal{H} + c_{i 1}  \mathcal{E}_1 + \cdots + c_{i K} \mathcal{E}_K
	\]
	for $i = 1, 2$.
	We may assume that $c_{1 0} \ge c_{2 0}$.
	Since $c_{i 0} \ge 0$ and $m_1 c_{1 0} + m_2 c_{2 0} \le 3$, we have $(c_{1 0}, c_{2 0}) = (0, 0)$,
	$(1, 0)$,
	$(1, 1)$,
	$(2, 1)$,
	$(2, 0)$ or
	$(3, 0)$.

	We show that only the cases $(c_{1 0}, c_{2 0}) = (3, 0)$ and $(2, 1)$ are possible.
	If $c_{1 0} = c_{2 0} = 0$, then $D_1$ and $D_2$ are both the strict transforms of some exceptional curves of $\pi$.
	Therefore, $D_1 \cdot D_2 = 0$ or $1$, which is a contradiction.
	If $c_{1 0} = c_{2 0} = 1$, then $\pi(D_1)$ and $\pi(D_2)$ are both lines in $\mathbb{P}^2$.
	Therefore, $\pi(D_1) \cdot \pi(D_2) = 1$, which contradicts the fact that $\pi(D_1) \cdot \pi(D_2) \ge D_1 \cdot D_2 = 2$.
	If $c_{1 0} = 1, 2$ and $c_{2 0} = 0$, then $\pi(D_1)$ is a smooth curve in $\mathbb{P}^2$ and $D_2$ is the strict transform of some exceptional curve of $\pi$.
	Therefore, $D_1 \cdot D_2 = 0$ or $1$, which is a contradiction.
	Therefore, we have $(c_{1 0}, c_{2 0}) = (3, 0), (2, 1)$.
	In particular, $m_1$ must be $1$.

	Next, we show that $D_1 \cdot D_2 = 2$.
	If $(c_{1 0}, c_{2 0}) = (2, 1)$, then $D_1 \cdot D_2 \le \pi(D_1) \cdot \pi(D_2) = 2$.
	Since $D_1 \cdot D_2 \ge 2$ by the assumption, we have $D_1 \cdot D_2 = 2$.
	Let us consider the case $(c_{1 0}, c_{2 0}) = (3, 0)$.
	Let $C = \pi(D_1) \subset \mathbb{P}^2$ and let $P = \pi(D_2) \in \mathbb{P}^2$.
	Since $c_{1 0} = 3$, the arithmetic genus of $C$ is $1$.
	Then, $C$ is either an elliptic curve, a rational curve with a node or a rational curve with a cusp.
	In particular, for any $Q \in C$, the multiplicity of $C$ at $Q$ is at most $2$.
	On the other hand, the multiplicity of $C$ at $P$ is at least $D_1 \cdot D_2 \ge 2$.
	Therefore, $C$ has multiplicity $2$ at $P$ and $D_1 \cdot D_2 = 2$.

	Since $D_1 \cdot D_2 = 2$ and $m_1 = 1$, it follows from Lemma~\ref{lemma:classification1} that $m_k = 1$.
	Hence, the configuration of $D_1$ and $D_2$ must be one of those in Lemma~\ref{lemma:classification_a1}.
\end{proof}

\begin{lemma}\label{lemma:classification_triple}
	If $D_1, D_2, D_3$ intersect at one point, then $m_1 = m_2 = m_3 = 1$ and $D = D_1 + D_2 + D_3$.
	In particular, the configuration of $D$ is \eqref{enum:notcycle} in Lemma~\ref{lemma:qclassification}.
\end{lemma}

\begin{proof}
	By Lemma~\ref{lemma:classification_double}, we have
	\begin{equation} \label{intersectioncondapp}
		D_1 \cdot D_2 = D_1 \cdot D_3 = D_2 \cdot D_3 = 1.
	\end{equation}
	Take a geometric basis $\mathcal{H}, \mathcal{E}_1, \ldots, \mathcal{E}_K \in \Pic(X)$ and the corresponding blow-down $\pi \colon X \to \mathbb{P}^2$.
	Using this basis, $[D_1], [D_2], [D_3]$ are expressed as
	\[
		[D_i] = c_{i 0} \mathcal{H} + c_{i 1}  \mathcal{E}_1 + \cdots + c_{i K} \mathcal{E}_K
	\]
	for $i = 1, 2, 3$.
	We may assume that $c_{1 0} \ge c_{2 0} \ge c_{3 0}$.

	First, we show that $m_1$ or $m_2$ is $1$.
	Since $c_{i 0} \ge 0$ and $m_1 c_{1 0} + m_2 c_{2 0} + m_3 c_{3 0} \le 3$, we have
	\[
		(c_{1 0}, c_{2 0}, c_{3, 0})
		= (0, 0, 0),
		(1, 0, 0),
		(1, 1, 0),
		(1, 1, 1),
		(2, 0, 0),
		(2, 1, 0)
		\text{ or }
		(3, 0, 0).	
	\]
	If $c_{1 0} = 0$, then $D_1, D_2, D_3$ are all the strict transforms of some exceptional curves of $\pi$, which contradicts the assumption that $D_1, D_2, D_3$ intersect at a point.
	To see this, write the classes of these strict transforms as $[D_i] = \mathcal{E}_{k_i} - \sum_{j \neq k_i} d_{i j} \mathcal{E}_j$, for $i=1,2,3$, where $d_{i,j} \in \Z_{\geq 0}$. 
	In order for condition \eqref{intersectioncondapp} to hold, we must have $k_1, k_2,k_3$ distinct and
	\begin{align*}
	[D_1] &= \mathcal{E}_{k_1} - \mathcal{E}_{k_2} - \mathcal{E}_{k_3} - \sum_{j\neq k_1,k_2,k_3} d_{1j} \mathcal{E}_j, \\
	[D_2] &= \mathcal{E}_{k_2} - \mathcal{E}_{k_1} - \mathcal{E}_{k_3} - \sum_{j\neq k_1,k_2,k_3} d_{2j} \mathcal{E}_j, \\
	[D_3] &= \mathcal{E}_{k_3} - \mathcal{E}_{k_1} - \mathcal{E}_{k_2} - \sum_{j\neq k_1,k_2,k_3} d_{3j} \mathcal{E}_j.
	\end{align*}
	However this would lead to a contradiction to the assumption that $D_i$ are effective divisors, since for example $D_1 + D_2$ has class being a sum of negative multiples of $\mathcal{E}_i$'s.
	
	If $(c_{1 0}, c_{2 0}, c_{3 0}) = (1, 0, 0)$, then $H := \pi(D_1) \subset \mathbb{P}^2$ is a line and $P := \pi(D_2) = \pi(D_3) \in H \subset \mathbb{P}^2$ is a point.
	In this case, since $H$ has multiplicity $1$ at $P$, it is impossible for $D_1, D_2, D_3$ to intersect at one point, which is again a contradiction.
	Therefore, we have $c_{10} + c_{20} \ge 2$, which implies that $m_j = 1$ for at least one $j\in \{1,2\}$.	
	
	Using Lemma~\ref{lemma:genus_simple}, we have
	\begin{align*}
		2 &= \sum_{i \ne j} m_i D_i \cdot D_j \\
		&= \sum_{k \in \{1,2,3\}, k\neq j} m_k D_k \cdot D_j + \sum_{i \ne 1, 2, 3} m_i D_i \cdot D_j \\
		&= \sum_{k \in \{1,2,3\}, k\neq j} m_k + \sum_{i \ne 1, 2, 3} m_i D_i \cdot D_j,
	\end{align*}
	which implies that $m_k = 1$ for $k\in \{1,2,3\}, k\neq j$.
	Therefore $m_1=m_2=m_3=1$, and the fact that $D_1,D_2,D_3$ intersect at a point means that by Lemma~\ref{lemma:qclassification} we have $D = D_1 + D_2 + D_3$.
\end{proof}

\begin{theorem}\label{theorem:classification}
Let $X$ be a surface with Cremona isometry $\sigma$ with maximal eigenvalue $\lambda > 1$ and effective anticanonical divisor $D = \sum_{j=1}^{N} m_j D_j$, and suppose $X$ is minimal as a space of initial conditions. Then
	\[
		\operatorname{rank} H_1(D_{\operatorname{red}}; \mathbb{Z}) \le 2.
	\]
	Moreover:
	\begin{enumerate}
		\item
		If $\operatorname{rank} H_1(D_{\operatorname{red}}; \mathbb{Z}) = 2$, then $D$ is irreducible and is an elliptic curve.

		\item
		If $\operatorname{rank} H_1(D_{\operatorname{red}}; \mathbb{Z}) = 1$ and $D$ is irreducible, then $D$ is a rational curve with one nodal singularity.

		\item\label{enum:classification_q}
		If $\operatorname{rank} H_1(D_{\operatorname{red}}; \mathbb{Z}) = 1$ and $D$ is not irreducible, then each component of $D$ is a smooth rational curve.
		In this case, there are two possible configurations of $D$:
		\begin{enumerate}
			\item
			$D = D_1 + D_2$ has two intersection points and their multiplicities are both $1$.

			\item
			$D = D_1 + \cdots + D_N$ forms a cycle as in \eqref{enum:cycle} of Lemma~\ref{lemma:qclassification}.

		\end{enumerate}
		
		\item
		If $\operatorname{rank} H_1(D_{\operatorname{red}}; \mathbb{Z}) = 0$ and $D$ is irreducible, then $D$ is a rational curve with one cusp.

		\item
		If $\operatorname{rank} H_1(D_{\operatorname{red}}; \mathbb{Z}) = 0$ and $D$ is not irreducible, then each component of $D$ is a smooth rational curve.
		In this case, there are three possible configurations of $D$:
		\begin{enumerate}
			\item
			$D = D_1 + D_2$ has only one intersection point and its multiplicity is $2$.

			\item
			$D = D_1 + D_2 + D_3$ has only one intersection point as in \eqref{enum:notcycle} of Lemma~\ref{lemma:qclassification}.

			\item
			All the intersections among $D_1, \ldots, D_N$ are transversal and, if considered as a graph, $D_1, \ldots, D_N$ does not have a cycle.

			In this case, if $m_j = 1$, then there exists $k$ such that $m_k = 2$ and
			$D_j \cdot D_i = \begin{cases}
				1 & (i = k) \\
				0 & (i \ne j, k).
			\end{cases}$

		\end{enumerate}

	\end{enumerate}
\end{theorem}

\begin{proof}
	First, we consider the case where $D$ is irreducible.
	In this case, by Lemma~\ref{lemma:irreducible}, the arithmetic genus of $D$ is $1$.
	Therefore, $D$ is either a smooth elliptic curve, a rational curve with one nodal singularity, or a rational curve with one cusp.
	Which case $D$ belongs to is determined by $\operatorname{rank} H_1(D_{\operatorname{red}}; \mathbb{Z}) \in \{ 0, 1, 2 \}$.

	From here on, we assume that $D$ is not irreducible.
	In this case, by Lemma~\ref{lemma:genus}, all the components are smooth rational curves.
	Therefore, by Lemma~\ref{lemma:cycle}, the rank of $H_1(D_{\operatorname{red}}; \mathbb{Z})$ is $0$ or $1$.
	
	Let us consider the case where there exist $k \ne j$ such that $D_j \cdot D_k \ge 2$.
	In this case, by Lemma~\ref{lemma:classification_double}, the configuration of $D$ is one in Lemma~\ref{lemma:classification_a1}.
	Which case $D$ belongs to is determined by $\operatorname{rank} H_1(D_{\operatorname{red}}; \mathbb{Z}) \in \{ 0, 1 \}$.

	If there exist $j, k, \ell$ such that $D_j \cap D_k \cap D_{\ell} \ne \emptyset$, then, by Lemma~\ref{lemma:classification_triple}, the configuration of $D$ is \eqref{enum:notcycle} of Lemma~\ref{lemma:qclassification}.

	Let us consider the remaining cases, i.e.\
	all the intersections among $D_1, \ldots, D_N$ are transversal.
	If $\operatorname{rank} H_1(D_{\operatorname{red}}; \mathbb{Z}) = 1$, then $D$ has a cycle $D_{j_1}, \ldots, D_{j_{\ell}}$ for $\ell \ge 3$.
	Therefore, by Lemma~\ref{lemma:cycle}, the configuration of $D$ is one of \eqref{enum:cycle} in Lemma~\ref{lemma:qclassification}.
	If $\operatorname{rank} H_1(D_{\operatorname{red}}; \mathbb{Z}) = 0$, then $D_1, \ldots, D_N$ do not have a cycle as a graph.
	The last statement follows from Lemmas~\ref{lemma:classification1} and \ref{lemma:qclassification}.
\end{proof}

\begin{remark} \label{rem:realisationintheliterature}
With regards to whether Theorem \ref{theorem:classification} is a classification in the sense that each of the possible types of anticanonical divisors are realised by a rational surface $X$ admitting a Cremona isometry $\sigma$ with largest eigenvalue $\lambda>1$, we note the following.
The family of examples in Section \ref{section5} correspond to subcase (b) of case (5), which has also appeared in \cite{UEHARA}. 
The examples in Section \ref{section6} fall under subcase (c) of case (5), while subcase (a) will appear in the follow-up paper \cite{part2}.
The case (4) of a rational curve with one cusp has appeared in several contexts in the literature, in particular \cite{dillercremona, mcmullen, UEHARACUSPIDAL}.
The case (3) corresponds to rational surfaces with an anticanonical cycle as considered by Looijenga \cite{looijenga}, and while automorphisms of such surfaces with dynamical degree greater than 1 were not explicitly constructed there, an example with cycle of length 3 will be presented in the follow-up paper \cite{part2}.
Examples where $D$ has two components, corresponding to subcase (a) of case (3) and subcase (a) of case (5), were constructed in \cite{jacksonthesis}.
As far as we are aware, there has not been an example corresponding to case (1) where $D$ is an elliptic curve. 
However such examples may be possible to construct as late confining versions of elliptic discrete Painlev\'e equations. 
Similarly examples corresponding to case (2) should be constructable as late confining versions of additive discrete Painlev\'e equations with symmetry type $E_8^{(1)}$.
\end{remark}

\section{Example: a family of mappings of QRT Class III form}  \label{section5}

The first example we will illustrate in detail deals with a family of mappings that are of the form of three-point mappings in QRT class III as appearing in \cite{ancillary}, i.e.
\begin{equation} \label{QRTclassIII}
\left( x_{n+1} + x_n \right)\left( x_n + x_{n-1} \right) = f(x_n),
\end{equation}
where $f$ is rational. 
From this point on we will make use of notation such as  
\begin{equation}
x_{n-1} = \ubar{x}, \qquad x_{n} = x, \qquad x_{n+1} = \bar{x},
\end{equation}
 to denote up- and down-shifts in $n$ of both variables and parameters.

%
We consider the family of equations 
\begin{equation} \label{classIIIexampleauto}
\left( \bar{x} + x \right)\left( x + \ubar{x} \right) = \frac{ \prod_{i=0}^{m} (x^2 - a_i^2) }{ \prod_{j=1}^{m} (x^2 - b_j^2) },
\end{equation}
where 
 $a_i \in \C \backslash \{0\}$, $i = 0, \dots m$ and $b_j \in \C \backslash \{0\}$, $j=1, \dots m$, for $m \in \Z_{\geq 1}$ with all $a_i$'s and $b_j$'s distinct.
Singularities of the mapping correspond to zeroes of the numerator or denominator of the rational function on the right-hand side of \eqref{classIIIexampleauto}, and calculating along the same lines as in the examples in Section \ref{section1} we find that all of these singularities are confined and we have the singularity patterns
\begin{equation}
\left\{ \pm a_i, \mp a_i \right\}, ~\text{ and }~ \left\{ \pm b_j, \infty, \mp b_j \right\}.
\end{equation}

\subsection{Deautonomisation by singularity confinement}

In this case it turns out that the following deautonomisation suffices to detect the dynamical degree of the mapping 

\begin{equation} \label{classIIIexamplenonauto}
\left( \bar{x} + x \right)\left( x + \ubar{x} \right) = \frac{ \prod_{i=0}^{m} \left(x - a_i \right) \left(x - d_i \right) }{ \prod_{j=1}^{m} \left(x  - b_j \right) \left(x - c_j \right)}. 
\end{equation}
Here $a_i$, $b_j$, $c_j$, $d_i$ are now $n$-dependent and will be required to evolve such that the structure of confined singularities persists. We shall now derive these evolutions explicitly.

First consider the singularities that correspond to zeroes of the numerator of the right-hand side of \eqref{classIIIexamplenonauto}.
In order to have the same kind of confinement behaviour of these singularities as the patterns $\left\{ \pm a_i, \mp a_i \right\}$ in the autonomous case, we require singularities which appear when some iterate $x$ takes a value for which the numerator of the right-hand side vanishes, to lead to a next iterate $\bar{x}$ that is a zero of the numerator of the up-shifted version of the equation.
Considering first the singularity $x = a_k$, for some $k\in\{ 0, \dots , m\}$ with $\ubar{x}=u$ free, we introduce a small parameter $\varepsilon$ and compute
\begin{equation}
\ubar{x} = u, \quad x = a_k + \varepsilon, \quad \bar{x} = - a_k + \mathcal{O}(\varepsilon),
\end{equation}
and similarly for the singularity $x= d_k$,
\begin{equation}
\ubar{x} = u, \quad x = d_k + \varepsilon, \quad \bar{x} = - d_k + \mathcal{O}(\varepsilon).
\end{equation}
For these singularities to be confined in the same way after deautonomisation we impose that $x=a_k$ leads to $\bar{x} = \bar{d}_k$, and $x=d_k$ leads to $\bar{x} = \bar{a}_k$, 
which requires the confinement conditions
\begin{equation} \label{classIIIconfconds1}
\bar{d}_k = - a_k,  \quad  \bar{a}_k = - d_k,
\end{equation}
for which it can be verified by direct calculation that we indeed have the confined singularity patterns $\{a_k , \bar{d}_k\}$ and $\{d_k, \bar{a}_k\}$.
This is tantamount to requiring that the lines blown down by the mapping $(\ubar{x},x) \mapsto (x, \bar{x})$ are sent to indeterminacies of the next iteration $(x, \bar{x}) \mapsto (\bar{x}, \bar{\bar{x}})$.
%
%

For the singularities $x=b_k$ and $x=c_k$ we have 
\begin{equation}\label{threepnine}
\begin{aligned}
\ubar{x} &= u, \\
 x &= b_k + \varepsilon, \\
 \bar{x} &= \frac{\prod_{i} (b_k - a_i)(b_k - d_i)}{(b_k + u) \prod_{j\neq k}(b_k - b_j)\prod_j(b_k - c_j)} \varepsilon^{-1} + F(u) + \mathcal{O}(\varepsilon), \\
 \bar{\bar{x}} &= - b_k  + \sum_{j=1}^{m} (\bar{b}_j + \bar{c}_j) + \sum_{i=0}^{m}(\bar{a}_i + \bar{d}_i) + G(u) \varepsilon +  \mathcal{O}(\varepsilon^2),
\end{aligned}
\end{equation}
where $F$ and $G$ are known fractional-linear functions of $u$ with coefficients expressed in terms of parameters.
An evolution for $x=c_k+\varepsilon$ similar to \eqref{threepnine} is obtained from by interchanging the roles of $b_k$ and $c_k$.
We require that these singularities
 are confined in the same way as in the autonomous case, namely through $\bar{\bar{x}}$ taking a value which is a root of the denominator of the right-hand side in the twice up-shifted version of \eqref{classIIIexamplenonauto}, corresponding to an indeterminacy $(\infty, \bar{\bar{b}}_j)$ or $(\infty, \bar{\bar{c}}_j)$ of the mapping $(\bar{x},\bar{\bar{x}})\mapsto (\bar{\bar{x}},\bar{\bar{\bar{x}}})$.
Which of these indetermacies the singularities $x=b_k$, $x=c_k$ are confined through can be chosen, without loss of generality, to correspond to the patterns $\{b_k, \infty, \bar{\bar{c}}_k\}$ and $\{c_k, \infty, \bar{\bar{b}}_k\}$ via the conditions 
\begin{equation}\label{classIIIconfconds2}
\begin{aligned}
\bar{\bar{c}}_k 		&= - b_k  + \sum_{j=1}^{m} (\bar{b}_j + \bar{c}_j) + \sum_{i=0}^{m}(\bar{a}_i + \bar{d}_i), \\
\bar{\bar{b}}_k 		&= - c_k +\sum_{j=1}^{m} (\bar{b}_j + \bar{c}_j) + \sum_{i=0}^{m}(\bar{a}_i + \bar{d}_i).
\end{aligned}
\end{equation}
For a sufficient deautonomisation in this case we can actually take the parameters $a_i$, $d_i$ in the numerator of the right-hand side of \eqref{classIIIexamplenonauto} to be constant: 
\begin{equation}
\bar{a}_i=a_i, \quad \bar{d}_i=d_i,
\end{equation}
see Remark \ref{rem:sec3sufficientdeauto} below.
With these parameters constant, the conditions \eqref{classIIIconfconds1} require
\begin{equation} \label{classIIIconstantparams}
a_i = - d_i \in \C \backslash \{0\}, \text{ for all } n.
\end{equation}
Then the remaining confinement conditions \eqref{classIIIconfconds2} give a linear system for the evolution of $\boldsymbol{b} = (b_1 ,\dots, b_m)^{T}$, $\boldsymbol{c} = (c_1 ,\dots, c_m)^{T}$:
\begin{equation} \label{classIIIconfcondsystem}
\left(\begin{array}{c}\bar{\boldsymbol{b}} \\ \bar{\boldsymbol{c}} \\ \boldsymbol{b} \\ \boldsymbol{c} \end{array}\right) = \left(\begin{array}{cccc}\mathbf{1}_m & \mathbf{1}_m & 0 & -\mathrm{I}_m \\\mathbf{1}_m & \mathbf{1}_m & -\mathrm{I}_m & 0 \\\mathrm{I}_m & 0 & 0 & 0 \\0 & \mathrm{I}_m & 0 & 0\end{array}\right)
 \left(\begin{array}{c} \boldsymbol{b} \\ \boldsymbol{c} \\ \ubar{\boldsymbol{b}} \\ \ubar{\boldsymbol{c}}\end{array}\right),
\end{equation}
where $\mathrm{I}_m$ is the $m \times m$ identity matrix, and $\mathbf{1}_m$ is the $m\times m$ matrix with all entries being equal to one. 
The matrix on the right-hand side is similar under column permutations to  
\begin{equation}
M = \left(\begin{array}{cc}\mathbf{1}_{2m} & - \mathrm{I}_{2m} \\ \mathrm{I}_{2m} & 0\end{array}\right), 
\end{equation}
so the characteristic polynomial of the linear system \eqref{classIIIconfcondsystem} coincides, up to cyclotomic factors, with that of $M$, which can be computed as a block determinant to be 
\begin{equation}
\det (M - t \, \mathrm{I}_{4m} ) = \left( t^2 + 1\right)^{2m-1}(t^2 - 2m t + 1).
\end{equation}
Therefore we see that the characteristic polynomial has as a root 
\begin{equation}
\lambda = m + \sqrt{m^2 - 1},
\end{equation}
which by the results of Section \ref{section2} we can conclude to be the dynamical degree of the original equation \eqref{classIIIexampleauto}, as we will go on to demonstrate in detail.


\subsection{Space of initial conditions}

We now illustrate how the mechanism by which the dynamical degrees of the family of mappings \eqref{classIIIexampleauto} were obtained above, fits into the framework of the results of Sections \ref{section2}-\ref{section4}, beginning with the construction of the space of initial conditions for the autonomous mapping. 
While our treatment of spaces of initial conditions in Section \ref{section2} treated equations as mappings of $\p^2$, often it is convenient to instead regard them as mappings initially on $\p^1 \times \p^1$, and we will do so here.
In Appendix \ref{app:compactification} we give details relating the picture starting from $\p^1 \times \p^1$ to the one using $\p^2$, so that the examples here and in Section \ref{section6} can still serve as demonstrations of the general results above.

We begin by taking the mapping initially on $\p^1 \times\p^1$ and perform a sequence of blow-ups to regularise it as an automorphism.
We do this in the usual way by introducing $X = 1/x$, $Y=1/y$ so $\p^1 \times\p^1$ is covered by the four affine charts $(x,y)$, $(X,y)$, $(x,Y)$ and $(X,Y)$ and the equation \eqref{classIIIexampleauto} defines via $y = \ubar{x}$ the mapping
\begin{equation}
\begin{gathered}
\varphi : \p^1 \times \p^1 \dashrightarrow \p^1 \times \p^1, \\
(x,y) \mapsto (\bar{x}, \bar{y}), \\
\bar{x} = - x + \frac{ \prod_{i=0}^{m} (x^2 - a_i^2 ) }{(x+y) \prod_{j=1}^{m} (x^2  - b_j^2 )}, \quad \bar{y}=x.
\end{gathered}
\end{equation}
 In order to regularise this as an automorphism we require blow-ups of points given in coordinates by 
 \begin{equation}
 \mathfrak{a}_i^{\pm} : (x,y) = (\pm a_i, \mp a_i), \qquad \mathfrak{b}_i^{\pm} : (x,y) = (\infty, \pm b_i), \qquad \tilde{\mathfrak{b}}_i^{\pm} : (x,y) = (\pm b_i, \infty),
 \end{equation}
 after which we have the rational surface $X$ as shown in Figure \ref{fig:surfaceclassIIIauto}. 
  
  \begin{figure}[htb]
 \begin{tikzpicture}[scale=.95,>=stealth,basept/.style={circle, draw=red!100, fill=red!100, thick, inner sep=0pt,minimum size=1.2mm}]
	\begin{scope}[xshift = -4cm]
			\draw [black, line width = 1pt] 	(4.1,2.5) 	-- (0,2.5)		node [left]  {$y=\infty$} node[pos=0, right] {$$};
			\draw [black, line width = 1pt] 	(3.6,3) -- (3.6,-.5)		node [below]  {$x=\infty$} node[pos=0, above, xshift=7pt] {};
			\draw [black, line width =1pt]     (0, -.2) -- (4,2.8)       node [below left, pos=0] {$x+y=0$};
			\node (btp1) at (.3,2.5) [basept,label={$\tilde{\mathfrak{b}}_{1}^{+}$}] {};
			\node (btm1) at (1,2.5) [basept,label={$\tilde{\mathfrak{b}}_{1}^{-}$}] {};
			\node (btilddots) at (1.75,2.5) [label={$\cdots$}] {};
			\node (btpm) at (2.5,2.5) [basept,label={$\tilde{\mathfrak{b}}_{m}^{+}$}] {};
			\node (btmm) at (3.2,2.5) [basept,label={$\tilde{\mathfrak{b}}_{m}^{-}$}] {};
			\node (bp1) at (3.6,-.2) [basept,label=right:{$\mathfrak{b}_{1}^{+}$}] {};
			\node (bm1) at (3.6,.4) [basept,label=right:{$\mathfrak{b}_{1}^{-}$}] {};
			\node (bdots) at (3.6,1) [label=right:{$\vdots$}] {};
			\node (bpm) at (3.6,2) [basept,label=right:{$\mathfrak{b}_{m}^{-}$}] {};
			\node (bpp) at (3.6,1.4) [basept,label=right:{$\mathfrak{b}_{m}^{+}$}] {};
			\node (ap1) at (.4,.1) [basept,label=above:{$\mathfrak{a}_{0}^{+}$}] {};
			\node (am1) at (.8,.4) [basept,label=above:{$\mathfrak{a}_{0}^{-}$}] {};
			\node (ap1) at (2,1.3) [basept,label=above:{$\mathfrak{a}_{m}^{+}$}] {};
			\node (am1) at (2.4,1.6) [basept,label=above:{$\mathfrak{a}_{m}^{-}$}] {};
			\node (adots) at (1.4,.85) [label=above:{$\iddots$}] {};
			\draw (1.8,-1) node [below,anchor=north] {$\p^1 \times\p^1$};
	\end{scope}
	
		\draw [->] (3,1.5)--(1.25,1.5) node[pos=0.5, below] {$\pi$};
	
	\begin{scope}[xshift = 4cm, yshift= 0cm]
			\draw [black, line width = 1pt] 	(4.1,2.5) 	-- (0,2.5)		node [left]  {} node[pos=1, left] {$D_2$};
			\draw [black, line width = 1pt] 	(3.6,3) -- (3.6,-.5)		node [below]  {$D_3$} node[pos=0, above, xshift=7pt] {};
			\draw [black, line width =1pt]     (0, -.2) -- (4,2.8)       node [below left, pos=0] {$D_1$};
			\draw [red, line width=1pt]		(.3, 2.3) --(.3,3) node [above] {$\tilde{B}_1^{+}$};
			\draw [red, line width=1pt]		(1, 2.3) --(1,3) node [above] {$\tilde{B}_1^{-}$};
			\node (btilddots) at (1.75,2.5) [label={$\cdots$}] {};
			\draw [red, line width=1pt]		(2.5, 2.3) --(2.5,3) node [above] {$\tilde{B}_m^{+}$};
			\draw [red, line width=1pt]		(3.2, 2.3) --(3.2,3) node [above] {$\tilde{B}_m^{-}$};
			\draw [red, line width =1pt]	(3.4,-.2) -- (4.1,-.2) node [right] {$B_1^{+}$};
			\draw [red, line width =1pt]	(3.4,.4) -- (4.1,.4) node [right] {$B_1^{-}$};
			\node (bdots) at (3.6,1) [label=right:{$\vdots$}] {};
			\draw [red, line width =1pt]	(3.4,2) -- (4.1,2) node [right] {$B_m^{+}$};
			\draw [red, line width =1pt]	(3.4,1.4) -- (4.1,1.4) node [right] {$B_m^{-}$};
			\draw [red, line width =1pt]	(.1,.5) -- (.7,-.3) node [xshift=3pt,yshift=-7pt] {$A_0^{+}$};
			\draw [red, line width =1pt]	(.5,.8) -- (1.1,0) node [xshift=5pt,yshift=-5pt] {$A_0^{-}$};
			\draw [red, line width =1pt]	(1.7,1.7) -- (2.3,.9) node [xshift=3pt,yshift=-7pt] {$A_m^{+}$};
			\draw [red, line width =1pt]	(2.1,2) -- (2.7,1.2) node [xshift=5pt,yshift=-5pt] {$A_m^{-}$};
			\node (am1) at (2.4,1.6){};
			\node (adots) at (1.1,.55) [label=above:{$\iddots$}] {};
			\draw (1.5,-1) node [below,anchor=north] {$X$};
	\end{scope}
	\end{tikzpicture}
 	\caption{Space of initial conditions for Class III example (autonomous).}
	\label{fig:surfaceclassIIIauto}
\end{figure}

Denoting the composition of the blow-ups by $\pi : X \rightarrow \p^1 \times \p^1$ and the exceptional divisors by $A_i^{\pm} = \pi^{-1} (\mathfrak{a}_i^{\pm})$, $B_j^{\pm} = \pi^{-1} (\mathfrak{b}_j^{\pm})$, $\tilde{B}_j^{\pm} = \pi^{-1} (\tilde{\mathfrak{b}}_j^{\pm})$,  we have 
 \begin{equation}
 \Pic(X) = \Z \mathcal{H}_x + \Z \mathcal{H}_y + \sum_{i=0}^m (\Z \mathcal{A}^+_i+ \Z \mathcal{A}^{-}_i )+ \sum_{j=1}^m (\Z \mathcal{B}^+_j + \Z \mathcal{B}^{-}_j +\Z \tilde{\mathcal{B}}^+_j + \Z \tilde{\mathcal{B}}^{-}_j),
 \end{equation} 
 where $\mathcal{H}_x = \pi^* (\mathcal{O}_{\p^1}(1) \times 1)$ and $\mathcal{H}_y = \pi^* (1 \times \mathcal{O}_{\p^1}(1) )$ correspond to classes of lines of constant $x$ and $y$ respectively, and we have used calligraphic script to indicate classes of exceptional divisors.
 Note that, as guaranteed by Theorem \ref{effectivenesstheorem}, the surface $X$ has effective anticanonical divisor given by 
$
D = D_1 + D_2 +D_3,
$
 where $D_1, D_2, D_3$ are the proper transforms of the curves given by $x+y=0$, $y=\infty$, and $x=\infty$ respectively.
This is the pole divisor of the 2-form on $X$ given in the original coordinates by $\frac{dx \wedge dy}{x+y}$, which is preserved by the mapping.
This corresponds to subcase (b) of case (5) in Theorem \ref{theorem:classification}.

Through calculations in charts, we find that the mapping acts by pullback on $\Pic(X)$ as follows.
\begin{equation} \label{PhiactionclassIII}
\varphi^* : \left\{ 
\begin{aligned}
\mathcal{H}_x 				&\mapsto 
(2m+1)\mathcal{H}_x + \mathcal{H}_y - \sum_{i=0}^{m} \mathcal{A}_{i}^{+}  - \sum_{i=0}^{m} \mathcal{A}_{i}^{-}  - \sum_{j=1}^{m} \tilde{\mathcal{B}}_{j}^{+} - \sum_{j=1}^{m} \tilde{\mathcal{B}}_{j}^{-},   \\
\mathcal{H}_y				&\mapsto		\mathcal{H}_x,  
\quad
\mathcal{A}_i^{\pm} 			\mapsto 		\mathcal{H}_x - \mathcal{A}_i^{\mp}, 
\quad 
\mathcal{B}_j^{\pm} 		\mapsto 		\mathcal{H}_x - \tilde{\mathcal{B}}_j^{\pm}, 
\quad
\tilde{\mathcal{B}}_j^{\pm} 	\mapsto 		\mathcal{B}_j^{\mp}. 
\end{aligned}
\right.
\end{equation}
While it is possible to compute the characteristic polynomial of the matrix of $\varphi^*$ with respect to this basis for $\Pic(X)$, it will be more convenient to do this in terms of the root basis we will introduce after deautonomisation.

\subsection{Space of initial conditions for deautonomised version}


Consider the deautonomised equation \eqref{classIIIexamplenonauto} defining the mapping
\begin{equation}
\begin{gathered}
\varphi_n : \p^1 \times \p^1 \dashrightarrow \p^1 \times \p^1, \\
(x,y) \mapsto (\bar{x}, \bar{y}), \\
\bar{x} = - x + \frac{ \prod_{i=0}^{m} (x - a_i)(x-d_i) }{(x+y) \prod_{j=1}^{m} (x  - b_j )(x - c_j)}, \quad \bar{y}=x,
\end{gathered}
\end{equation}
subject to the confinement conditions derived above, namely
\begin{equation} \label{classIIIexampleconfcondsall}
\begin{gathered}
\bar{a}_i = - d_i,  \quad \bar{d}_i = - a_i, \\
\bar{b}_k 		= - \ubar{c}_k +\sum_{j=1}^{m} (b_j + c_j) + \sum_{i=0}^{m}(a_i + d_i), \\
\bar{c}_k 	= - \ubar{b}_k  + \sum_{j=1}^{m} (b_j + c_j) + \sum_{i=0}^{m}( a_i + d_i).
\end{gathered}
\end{equation}

To construct the space of initial conditions we introduce the surface $X_n$ obtained by blowing up $\p^1 \times \p^1$ at the following points, as shown in Figure \ref{fig:surfaceclassIIInonauto}:
 \begin{equation}
 \begin{aligned}
\mathfrak{a}_i &: (x,y) = (a_i, - a_i), 				&&\mathfrak{e}_i : (x,y) = (d_i, - d_i), 			&&&i=0,\dots,m,	\\
 \mathfrak{b}_j &: (x,y) = (\infty, \ubar{b}_j), 		&&\mathfrak{c}_j : (x,y) = (\infty, \ubar{c}_j), 		&&&\\ 
 \tilde{\mathfrak{b}}_j &: (x,y) = (b_j, \infty), 		&&\tilde{\mathfrak{c}}_j : (x,y) = (c_j, \infty),		&&&j=1,\dots,m.
 \end{aligned}
 \end{equation}
 Denoting the composition of the blow-ups by $\pi_n : X_n \rightarrow \p^1 \times \p^1$, and the exceptional divisors by $A_i = \pi_n^{-1} (\mathfrak{a}_i)$, $E_i = \pi_n^{-1} (\mathfrak{e}_i)$, $B_j = \pi_n^{-1} (\mathfrak{b}_j)$, $C_j = \pi_n^{-1} (\mathfrak{c}_j)$, $\tilde{B}_j = \pi_n^{-1} (\tilde{\mathfrak{b}}_j)$, $\tilde{C}_j = \pi_n^{-1} (\tilde{\mathfrak{c}}_j)$, we identify all $\Pic(X_n)$ into a single $\Z$-module
\begin{equation}
 \Pic(\X) = \Z \mathcal{H}_x + \Z \mathcal{H}_y + \sum_{i=0}^m (\Z \mathcal{A}_i + \Z \mathcal{E}_i) + \sum_{j=1}^m (\Z \mathcal{B}_j + \Z \mathcal{C}_j + \Z \tilde{\mathcal{B}}_j + \Z \tilde{\mathcal{C}}_j),
\end{equation}
where $ \iota_n(\mathcal{H}_x) =  \pi_n^*(\mathcal{O}_{\p^1} (1) \times 1) $, $\iota_n (\mathcal{A}_i) = [A_i]$ and so on. 
The intersection product gives the symmetric bilinear form on $ \Pic(\X)$ defined by $\mathcal{H}_x \cdot \mathcal{H}_y = 1$, $\mathcal{H}_x \cdot \mathcal{H}_x = \mathcal{H}_y \cdot \mathcal{H}_y = 0$, with all generators that correspond to exceptional divisors of blow-ups being of self-intersection $-1$ and orthogonal to all other generators.
Each surface $X_n$ has an effective anticanonical divisor given by 
$
D = D_1 + D_2 +D_3 \in |-\K_{X_{n}} |,
$
 where $D_1, D_2, D_3$ are the proper transforms of the curves given by $x+y=0$, $y=\infty$, and $x=\infty$ respectively.

\begin{figure}[htb]
 \begin{tikzpicture}[scale=.95,>=stealth,basept/.style={circle, draw=red!100, fill=red!100, thick, inner sep=0pt,minimum size=1.2mm}]
	\begin{scope}[xshift = -4cm]
			\draw [black, line width = 1pt] 	(4.1,2.5) 	-- (0,2.5)		node [left]  {$y=\infty$} node[pos=0, right] {$$};
			\draw [black, line width = 1pt] 	(3.6,3) -- (3.6,-.5)		node [below]  {$x=\infty$} node[pos=0, above, xshift=7pt] {};
			\draw [black, line width =1pt]     (0, -.2) -- (4,2.8)       node [below left, pos=0] {$x+y=0$};
			\node (btp1) at (.3,2.5) [basept,label={$\tilde{\mathfrak{b}}_{1}$}] {};
			\node (btm1) at (1,2.5) [basept,label={$\tilde{\mathfrak{c}}_{1}$}] {};
			\node (btilddots) at (1.75,2.5) [label={$\cdots$}] {};
			\node (btpm) at (2.5,2.5) [basept,label={$\tilde{\mathfrak{b}}_{m}$}] {};
			\node (btmm) at (3.2,2.5) [basept,label={$\tilde{\mathfrak{c}}_{m}$}] {};
			\node (bp1) at (3.6,-.2) [basept,label=right:{$\mathfrak{b}_{1}$}] {};
			\node (bm1) at (3.6,.4) [basept,label=right:{$\mathfrak{c}_{1}$}] {};
			\node (bdots) at (3.6,1) [label=right:{$\vdots$}] {};
			\node (bpm) at (3.6,2) [basept,label=right:{$\mathfrak{c}_{m}$}] {};
			\node (bpp) at (3.6,1.4) [basept,label=right:{$\mathfrak{b}_{m}$}] {};
			\node (ap1) at (.4,.1) [basept,label=above:{$\mathfrak{a}_{0}$}] {};
			\node (am1) at (.8,.4) [basept,label=above:{$\mathfrak{e}_{0}$}] {};
			\node (ap1) at (2,1.3) [basept,label=above:{$\mathfrak{a}_{m}$}] {};
			\node (am1) at (2.4,1.6) [basept,label=above:{$\mathfrak{e}_{m}$}] {};
			\node (adots) at (1.4,.85) [label=above:{$\iddots$}] {};
			\draw (1.8,-1) node [below,anchor=north] {$\p^1 \times\p^1$};
	\end{scope}
	
		\draw [->] (3,1.5)--(1.25,1.5) node[pos=0.5, below] {$\pi_n$};
	
	\begin{scope}[xshift = 4cm, yshift= 0cm]
			\draw [black, line width = 1pt] 	(4.1,2.5) 	-- (0,2.5)		node [left]  {} node[pos=1, left] {$D_2$};
			\draw [black, line width = 1pt] 	(3.6,3) -- (3.6,-.5)		node [below]  {$D_3$} node[pos=0, above, xshift=7pt] {};
			\draw [black, line width =1pt]     (0, -.2) -- (4,2.8)       node [below left, pos=0] {$D_1$};
			\draw [red, line width=1pt]		(.3, 2.3) --(.3,3) node [above] {$\tilde{B}_1$};
			\draw [red, line width=1pt]		(1, 2.3) --(1,3) node [above] {$\tilde{C}_1$};
			\node (btilddots) at (1.75,2.5) [label={$\cdots$}] {};
			\draw [red, line width=1pt]		(2.5, 2.3) --(2.5,3) node [above] {$\tilde{B}_m$};
			\draw [red, line width=1pt]		(3.2, 2.3) --(3.2,3) node [above] {$\tilde{C}_m$};
			\draw [red, line width =1pt]	(3.4,-.2) -- (4.1,-.2) node [right] {$B_1$};
			\draw [red, line width =1pt]	(3.4,.4) -- (4.1,.4) node [right] {$C_1$};
			\node (bdots) at (3.6,1) [label=right:{$\vdots$}] {};
			\draw [red, line width =1pt]	(3.4,2) -- (4.1,2) node [right] {$C_m$};
			\draw [red, line width =1pt]	(3.4,1.4) -- (4.1,1.4) node [right] {$B_m$};
			\draw [red, line width =1pt]	(.1,.5) -- (.7,-.3) node [xshift=3pt,yshift=-7pt] {$A_0$};
			\draw [red, line width =1pt]	(.5,.8) -- (1.1,0) node [xshift=7pt,yshift=-7pt] {$E_0$};
			\draw [red, line width =1pt]	(1.7,1.7) -- (2.3,.9) node [xshift=3pt,yshift=-7pt] {$A_m$};
			\draw [red, line width =1pt]	(2.1,2) -- (2.7,1.2) node [xshift=7pt,yshift=-7pt] {$E_m$};
			\node (am1) at (2.4,1.6){};
			\node (adots) at (1.1,.55) [label=above:{$\iddots$}] {};
			\draw (1.5,-1) node [below,anchor=north] {$X_n$};
	\end{scope}
	\end{tikzpicture}
 	\caption{Space of initial conditions for Class III example (deautonomised).}
	\label{fig:surfaceclassIIInonauto}
\end{figure}

\begin{proposition} \label{classIIImatrixPhionPic}
With the confinement conditions \eqref{classIIIexampleconfcondsall}, the mapping $\varphi_n$ becomes an isomorphism $\tilde{\varphi}_n = \pi_{n+1}^{-1} \circ \varphi_n \circ \pi_n : X_n \rightarrow X_{n+1}$,
and its pullback induces the following lattice automorphism $\Phi =  \iota_n^{-1} \circ \varphi_n^* \circ \iota_{n+1}$ of $\Pic(\X)$:
\begin{equation}
\Phi : 
\left\{ 
\begin{aligned}
\mathcal{H}_x 				&\mapsto 
(2m+1)\mathcal{H}_x + \mathcal{H}_y - \sum_{i=0}^{m} \mathcal{A}_{i}  - \sum_{i=0}^{m} \mathcal{E}_{i}  - \sum_{j=1}^{m} \tilde{\mathcal{B}}_{j} - \sum_{j=1}^{m} \tilde{\mathcal{C}}_{j}, \qquad
\mathcal{H}_y 				\mapsto		\mathcal{H}_x, \\
\mathcal{A}_i 				&\mapsto 		\mathcal{H}_x - \mathcal{E}_i, \qquad
\mathcal{E}_i 				\mapsto 		\mathcal{H}_x - \mathcal{A}_i, \qquad
\mathcal{B}_j	 			\mapsto 		\mathcal{H}_x - \tilde{\mathcal{B}}_j, \qquad
\mathcal{C}_j	 			\mapsto 		\mathcal{H}_x - \tilde{\mathcal{C}}_j, \\
\tilde{\mathcal{B}}_j	 		&\mapsto 			\mathcal{C}_j, \qquad
\tilde{\mathcal{C}}_j	 		\mapsto 			\mathcal{B}_j. 
\end{aligned}
\right.
\end{equation}
\end{proposition}

\begin{remark}
This means that the mapping \eqref{classIIIexamplenonauto} with the confinement conditions \eqref{classIIIexampleconfcondsall} is a deautonomisation of \eqref{classIIIexampleauto} in the sense of Definition \ref{deautodef} and its adaptation to the $\p^1 \times \p^1$ setup in Appendix \ref{app:compactification}, via 
\begin{equation}
\begin{gathered}
    \kappa : \Pic(\X) \rightarrow \Pic(X), \\
    \A_i \mapsto \A_i^{+}, \quad \E_i \mapsto \A_i^{-}, \quad
    \B_j \mapsto \B_j^{+},  \quad \mathcal{C}_j \mapsto \B_j^{-}, \quad
    \tilde{\B}_j \mapsto \tilde{\B}_j^{+},  \quad \tilde{\mathcal{C}}_j \mapsto \tilde{\B}_j^{-}.
\end{gathered}
\end{equation}
\end{remark}

\subsection{Root basis and period map}

With the space of initial conditions and linear transformation $\Phi$ in hand, we now calculate the period map of $X_n$ and confirm the mechanism by which the dynamical degree appeared in the confinement conditions.
We take the rational 2-form to be that given in the initial affine charts for $\p^1 \times \p^1$ by 
\begin{equation} \label{omegaclassIII}
\omega = k \frac{dx \wedge dy}{x+y} = - k \frac{dX \wedge dy}{X(1+X y)} = - k \frac{dx \wedge dY}{Y(x Y + 1)} = k \frac{dX \wedge dY }{X Y (X+ Y)},
\end{equation}
with $k \in \C^{*}$ being arbitrary at this stage, and denote its lift to $X_n$ under $\pi_n$ by
\begin{equation}
\tilde{\omega}_n = \pi_n^* \omega. 
\end{equation}
We have $- \div (\tilde{\omega}_n) = D_1 + D_2 + D_3$, and the elements of $\Pic(\X)$ corresponding to the irreducible components are 
\begin{equation}
\mathcal{D}_1 = \mathcal{H}_x + \mathcal{H}_y - \sum_{i=0}^{m}\A_i - \sum_{i=0}^{m}\E_i, \quad 
\mathcal{D}_2 = \mathcal{H}_y   - \sum_{j=1}^{m}\tilde{\mathcal{B}}_j  - \sum_{j=1}^{m}\tilde{\mathcal{C}}_j , \quad 
\mathcal{D}_3 = \mathcal{H}_x   - \sum_{j=1}^{m}\mathcal{B}_j  - \sum_{j=1}^{m}\mathcal{C}_j ,
\end{equation}
so we denote their span and its orthogonal complement respectively by
\begin{equation}
Q = \operatorname{span}_{\Z} \left\{ \mathcal{D}_1, \mathcal{D}_2, \mathcal{D}_3\right\}, \quad Q^{\perp} = \left\{ \F \in \Pic(\X)~|~ \F \cdot \mathcal{D}_i  = 0\right\}.
\end{equation}

To find a root basis for $Q^{\perp}$ in the sense of Definition \ref{rootbasisdef}, we note that when $m=1$, $X_n$ becomes a Sakai surface of (additive) type $A_2^{(1) *}$ in Sakai's classification scheme so we can use the root basis there, which is formed of simple roots for an affine root system of type $E_6^{(1)}$ (see \cite[Appendix A]{SAKAI2001} or \cite[8.2.16]{KNY}). 
It turns out that this choice can be generalised to the $m>1$ case as follows.

%

\begin{proposition} \label{classIIIrootbasisprop}
We have a root basis for $Q^{\perp} \subset \Pic(\X)$ given by 
\begin{equation} 
Q^{\perp} = \sum_{i=0}^{2m} \Z \alpha_i + \sum_{j=1}^{2m-1} \Z \beta_j+ \sum_{j=1}^{2m-1} \Z \gamma_j + \Z \varepsilon_1 + \Z \varepsilon_2,
\end{equation}
where 
\begin{equation}
\begin{gathered}
\varepsilon_1 = \mathcal{H}_y - \A_0 - \B_{1}, \qquad \varepsilon_2 = \mathcal{H}_x - \A_0 - \tilde{\B}_{1}, 
\\
\begin{aligned}
\alpha_0 &= \A_{0} - \E_{0}, \\
\alpha_1 &= \E_{0} - \A_{1}, \\
\alpha_2 &= \A_{1} - \E_{1},  \\
&\vdots \\
\alpha_{2m-1} &= \E_{m-1} - \A_{m}, \\
\alpha_{2m} &= \A_{m} - \E_{m},
\end{aligned}
\quad
\begin{aligned}
\beta_1 &= \B_{1} - \mathcal{C}_{1}, \\
\beta_2 &= \mathcal{C}_{1} - \B_{2}, \\
\beta_3 &= \B_{2} - \mathcal{C}_{2}, \\
&\vdots \\
\beta_{2m-2} &= \mathcal{C}_{m-1} - \B_{m}, \\
\beta_{2m-1} &= \B_{m} - \mathcal{C}_{m},
\end{aligned}
\quad
\begin{aligned}
\gamma_1 &= \tilde{\B}_{1} - \tilde{\mathcal{C}}_{1}, \\
\gamma_2 &= \tilde{\mathcal{C}}_{1} - \tilde{\B}_{2}, \\
\gamma_3 &= \tilde{\B}_{2} - \tilde{\mathcal{C}}_{2} , \\
&\vdots \\
\gamma_{2m-2} &= \tilde{\mathcal{C}}_{m-1} - \tilde{\B}_{m}, \\
\gamma_{2m-1} &= \tilde{\B}_{m} - \tilde{\mathcal{C}}_{m}.
\end{aligned}
\end{gathered}
\end{equation}

\end{proposition}

\begin{proof}
In this case $\K_{\X} \not \in Q^{\perp}$ and $\operatorname{rank}Q=3$, so since $\operatorname{rank} \Pic(\X) = 6m + 4$ we need to find $6m+1$ elements to construct a root basis. 
We begin by taking the obvious choices of differences of pairs of classes of exceptional divisors intersecting the components $D_1, D_2, D_3$, which provides $6m-1$ elements and leads to $\alpha_0,\dots, \alpha_{2m}$, $\beta_{1},\dots, \beta_{2m-1}$ and $\gamma_{1}, \dots, \gamma_{2m-1}$ as above.
For the additional two elements required to form a root basis, we take cues from the known choice in the $m=1$ case and set $\varepsilon_1$ to be the difference of the class $\mathcal{H}_y - \A_0$  of the proper transform of the line $y=-a_0$ and the class $\B_1$ of $B_1$, which both intersect $D_3$. 
Similarly we choose $\varepsilon_2$ to be the difference of the class $\mathcal{H}_x - \A_0$ of the proper transform of the line $x=a_0$ and the class $\tilde{\B}_1$ of $\tilde{B}_1$, which both intersect $D_2$.
The fact that these choices lead to a root basis in the sense of Definition \ref{rootbasisdef} is immediate, and in Figure \ref{fig:dynkindiagramclassIII} we present the graph corresponding to the generalised Cartan matrix given by the intersection numbers of the elements, i.e. nodes corresponding to elements, with edges joining them if the pair of elements has intersection $1$.
We remark that this graph already appeared in the work of Looijenga \cite{looijenga} as the intersection configuration of a root basis for a surface with an anticanonical cycle of length three - this root basis will appear again in a way more directly related to Looijenga's work in one of the families of examples in the follow-up paper \cite{part2}.
\end{proof}

\begin{figure}[htb]
\centering
\begin{tikzpicture}[scale=.8, 
				elt/.style={circle,draw=black!100, fill=blue!75, thick, inner sep=0pt,minimum size=1.75mm},
				blk/.style={circle,draw=black!100, fill=black!75, thick, inner sep=0pt,minimum size=1.75mm}, 
				mag/.style={circle, draw=black!100, fill=magenta!100, thick, inner sep=0pt,minimum size=1.75mm}, 
				red/.style={circle, draw=black!100, fill=red!100, thick, inner sep=0pt,minimum size=1.75mm} ]
	\path 
			(-4,0) 	node 	(bleft) [blk, label={[xshift=-0pt, yshift = -22 pt] $\beta_{2m-1}$} ] {}
			(-3,0) 	node 	(bleft-1) [blk, label={[xshift=-0pt, yshift = -22 pt] $$} ] {}
			(-2.45,0)	node 	(dotsleft) {$\dots$}
			(-2,0) 	node 	(d1) [blk, label={[xshift=-0pt, yshift = -22 pt] $\beta_{1}$} ] {}
		        (-1,0) node 	(d2) [blk, label={[xshift=-0pt, yshift = -22 pt] $\varepsilon_1$} ] {}
		        ( 0,0) 	node  	(d3) [blk, label={[xshift=0pt, yshift = -22 pt] $\alpha_{0}$} ] {}
		        ( 1,0) 	node  	(d4) [blk, label={[xshift=0pt, yshift = -22 pt] $\varepsilon_{2}$} ] {}
		        ( 2,0) node 	(d5) [blk, label={[xshift=0pt, yshift = -22 pt] $\gamma_1$} ] {}
			(2.52,0)	node 	(dotsright) {$\dots$}
			(3,0) 	node 	(bright-1) [blk, label={[xshift=-0pt, yshift = -22 pt] $$} ] {}
			(4,0) 	node 	(bright) [blk, label={[xshift=5pt, yshift = -22 pt] $\gamma_{2m-1}$} ] {}
		        ( 0,1) node 	(d6) [blk, label={[xshift=12pt, yshift = -12 pt] $\alpha_{1}$} ] {}
		        ( 0,2) node 	(d0) [blk, label={[xshift=12pt, yshift = -12 pt] $\alpha_{2}$} ] {}
			(0,2.61)	node 	(dotstop) {$\vdots$}
		        ( 0,3) node 	(atop-1) [blk, label={[xshift=12pt, yshift = -12 pt] } ] {}
		        ( 0,4) node 	(atop) [blk, label={[xshift=15pt, yshift = -12 pt] $\alpha_{2m}$} ] {};
		\draw [black,line width=1pt ] (d1) -- (d2) -- (d3) -- (d4)  -- (d5);
		\draw [black,line width=1pt ] (bleft) -- (bleft-1) ;
		\draw [black,line width=1pt ] (bright) -- (bright-1) ;
		\draw [black,line width=1pt ] (d3) -- (d6) -- (d0) ;		
		\draw [black,line width=1pt ] (atop-1) -- (atop) ;
	\end{tikzpicture}
 	\caption{Intersection diagram of root basis for $Q^{\perp}$ in Class III example}
	\label{fig:dynkindiagramclassIII}
\end{figure}


We now compute the period mapping on this basis via the procedure outlined in Section \ref{section2}. 
Note that in this case $D_{\operatorname{red}}$ is the union of three copies of $\p^1$ meeting at a single point so $H_1(D_{\operatorname{red}};\Z) = 0$, and the only normalisation we need in order to define $\chi$ as a $\C$-valued function is a choice of the constant $k$ appearing in $\omega$ as in \eqref{omegaclassIII}. 
A convenient choice leads to the following, in which the period map is that on $X_n$ and we slightly abuse notation by writing $\alpha_i$ etc. when more precisely we mean their corresponding elements $\iota_n(\alpha_i) \in \Pic(X_n)$.


\begin{lemma} \label{rootvariableslemmaclassIII}
The root variables for the root basis in Proposition \ref{classIIIrootbasisprop} are given by
\begin{equation}
\begin{aligned}
\chi(\alpha_0) &=   a_0 - d_0, \\
\chi(\alpha_1) &=  d_{0} - a_{1}, \\
\chi(\alpha_2) &=  a_{1} - d_{1},  \\
\chi(\alpha_3) &=  d_{1} - a_{2},  \\
&\vdots \\
\chi(\alpha_{2m-1}) &= d_{m-1} - a_{m}, \\
\chi(\alpha_{2m}) &= a_{m} - d_{m},
\end{aligned}
\quad
\begin{aligned}
\chi(\varepsilon_1) &= - a_0 -  \ubar{b}_1, \\
\chi(\beta_1) &=  \ubar{b}_{1} - \ubar{c}_{1}, \\
\chi(\beta_2) &=  \ubar{c}_{1} - \ubar{b}_{2}, \\
\chi(\beta_3) &=  \ubar{b}_{2} - \ubar{c}_{2}, \\
&\vdots \\
\chi(\beta_{2m-2}) &=  \ubar{c}_{m-1} - \ubar{b}_{m}, \\
\chi(\beta_{2m-1}) &=  \ubar{b}_{m} - \ubar{c}_{m},
\end{aligned}
\quad
\begin{aligned}
\chi(\varepsilon_2) &= - a_{0} + b_{1}, \\
\chi(\gamma_1) &= - b_{1} + c_{1}, \\
\chi(\gamma_2) &= - c_1 + b_2, \\
\chi(\gamma_3) &= - b_{2} + c_{2}, \\
&\vdots \\
\chi(\gamma_{2m-2}) &= -  c_{m-1} + b_{m}, \\
\chi(\gamma_{2m-1}) &= - b_{m} + c_{m}.
\end{aligned}
\end{equation}
\end{lemma}

\begin{proof}
In order to compute the period mapping using the procedure outlined in Section \ref{section2}, we first note that the residues of $\tilde{\omega}_n$ along the components of the anticanonical divisor $D$ are given in coordinates by 
\begin{equation}
\begin{aligned}
\operatorname{Res}_{D_1} \tilde{\omega}_n &= \operatorname{Res}_{s=0} k \frac{dx \wedge d s}{s} = - k d x, \\
\operatorname{Res}_{D_2} \tilde{\omega}_n &= \operatorname{Res}_{Y=0} - k \frac{dx \wedge d Y}{Y(x Y+1)} = k  d x, \\
\operatorname{Res}_{D_3} \tilde{\omega}_n &= \operatorname{Res}_{X=0} - k \frac{dX \wedge d y}{X(1 + X y)} = - k  d y, 
\end{aligned}
\end{equation}
where $s = x+y$. 
To calculate $\chi(\alpha_i)$ for $i=0,\dots, 2m$ we note that we can express $\A_j - \E_k = [A_j] - [E_k]$ as the difference of two exceptional curves of the first kind which intersect $D_1$, which is the proper transform of the line $x+y=0$. 
Therefore using the residue computed above we have 
\begin{equation}
\chi( \A_j - \E_k) = 2 \pi i \int_{E_k \cap D_1}^{A_j \cap D_1} \operatorname{Res}_{D_1} \tilde{\omega}_n = 2 \pi i \int_{x=d_k}^{x=a_j}  - k d x = - 2 \pi i k ( a_j - d_k).
\end{equation}
Similarly we find 
\begin{equation}
\begin{aligned}
\chi( \tilde{\B}_j - \tilde{\mathcal{C}}_k) &= 2 \pi i \int_{\tilde{C}_k \cap D_2}^{\tilde{B}_j \cap D_2} \operatorname{Res}_{D_2} \tilde{\omega}_n = 2 \pi i \int_{x=c_k}^{x=b_j}  k d x = 2 \pi i k ( b_j - c_k),\\
\chi( \B_j - \mathcal{C}_k) &= 2 \pi i \int_{C_k \cap D_3}^{B_j \cap D_3} \operatorname{Res}_{D_3} \tilde{\omega}_n = 2 \pi i \int_{y=\ubar{c}_k}^{x=\ubar{b}_j}  - k d y = - 2 \pi i k ( \ubar{b}_j - \ubar{c}_k),
\end{aligned}
\end{equation}
which allows us to deduce the values of $\chi$ on all elements of the root basis except for $\varepsilon_1$ and $\varepsilon_2$, but these are computed similarly. 
For example for $\varepsilon_1$, we use its expression as the difference of classes $\mathcal{H}_y - \A_0$ and $\B_1$ of exceptional curves intersecting $D_3$, and note that $\mathcal{H}_y - \A_0$ corresponds to the proper transform of the line $y= - a_0$ so we compute 
\begin{equation}
\chi( \varepsilon_1 ) = 2 \pi i \int_{y=\ubar{b}_1}^{y=-a_0}  - k d y = 2 \pi i k (a_0 + \ubar{b}_1),
\end{equation}
and $\varepsilon_2$ is dealt with similarly. 
Finally we choose $k$ such that $2 \pi i k = -1$, and we have the result.
\end{proof}
\begin{remark}
At this stage we see that in order for this deautonomisation to be sufficient in the sense of Definition \ref{sufficientdeautonomisationdef} it must be nontrivial, i.e. $\bar{c}_i \neq c_i$ and $\bar{b}_i\neq b_i$. 
If not, then for example $\chi(\beta_i)=- \chi(\gamma_i)$ and the injectivity hypothesis breaks down.
\end{remark}
With this Lemma in hand, we can compute how the values of the period map on the root basis evolve with $n$ as dictated by the confinement conditions \eqref{classIIIexampleconfcondsall}. 
Denoting vectors of elements of the root basis by 
$\boldsymbol{\alpha}_e = (\alpha_0, \alpha_2, \dots , \alpha_{2m})^{T}$, $\boldsymbol{\alpha}_o = (\alpha_1, \alpha_3, \dots , \alpha_{2m-1})^{T}$, $\boldsymbol{\beta}_o = (\beta_1, \beta_3, \dots, \beta_{2m-1})^{T}$, $\boldsymbol{\beta}_e = (\beta_2, \beta_4, \dots, \beta_{2m})^{T}$, $\boldsymbol{\gamma}_o = (\gamma_1, \gamma_3, \dots, \gamma_{2m-1})^{T}$, and $\boldsymbol{\gamma}_e = (\gamma_2, \gamma_4, \dots, \gamma_{2m})^{T}$, the following is obtained by direct calculation.
\begin{proposition}
The evolution with $n$ of the root variables for the surface $X_n$ is given by
\begin{gather}
\bar{\boldsymbol{\chi}}=M_\Phi\cdot\boldsymbol{\chi}, \\
\text{where}\hskip.5cm\boldsymbol{\chi} = \left(\chi(\varepsilon_1), \chi(\varepsilon_2),\chi(\boldsymbol{\alpha}_o), \chi(\boldsymbol{\alpha}_e), \chi(\boldsymbol{\beta}_e),\chi(\boldsymbol{\beta}_o), \chi(\boldsymbol{\gamma}_e), \chi(\boldsymbol{\gamma}_o)\,\right)^T,\nonumber
\\
\hskip1.3cm\bar{\boldsymbol{\chi}} = \left(\bar{\chi}(\bar{\varepsilon}_1), \bar{\chi}(\bar{\varepsilon}_2), \bar{\chi}(\bar{\boldsymbol{\alpha}}_o), \bar{\chi}(\bar{\boldsymbol{\alpha}}_e), \bar{\chi}(\bar{\boldsymbol{\beta}}_e), \bar{\chi}(\bar{\boldsymbol{\beta}}_o), \bar{\chi}(\bar{\boldsymbol{\gamma}}_e), \bar{\chi}(\bar{\boldsymbol{\gamma}}_o) \right)^T,\nonumber\\[-2mm]\nonumber\\
\text{and}\hskip.5cm M_\Phi  =
 \left(\begin{array}{cccccccc} 0 & -1 & * & * & * & * & * & * \\1 & 2m & * & * & * & * & * & * \\ 0 & 0  & -I_{m+1} & * & 0  & 0  &  0 & 0  \\ 0 & 0  & 0  & I_{m} & 0  & 0  &  0 & 0  \\  0 &  0 &  0 & 0 & 0  & 0  & -I_m & 0  \\ 0 & 0  & 0  &  0 & 0  & 0  &  0 & -I_m \\ 0 & 0  & 0  & 0  &  I_m & * & 0  & 0  \\  0& 0  & 0  &  0 &  0 & I_m &  0 & 0 \end{array}\right)
\end{gather}
and where $\chi$, $\bar{\chi}$ are the period maps of the surfaces $X_n$, $X_{n+1}$ respectively, and bars on root basis elements  indicate their counterparts in $\Pic(X_{n+1})$ under the identification $\iota_{n+1}$.

\end{proposition}
%
%

As expected, the evolution of the root variables is given by exactly the same matrix as that of $\Phi$ on $Q^{\perp}$ with respect to the root basis.
Indeed, by direct computation using the expressions in Proposition \ref{classIIIrootbasisprop} for the root basis in terms of generators of $\Pic(\X)$ and $\Phi$ as it appears in Proposition \ref{classIIImatrixPhionPic}, we see that the evolution of the root basis for $Q^{\perp}$ under $\Phi$ takes the form
\begin{equation}
\Phi:\ \boldsymbol{\delta} \mapsto M_\Phi\, \boldsymbol{\delta},
\end{equation}
where $\boldsymbol{\delta} = \left(\varepsilon_1, \varepsilon_2, \boldsymbol{\alpha}_o,\boldsymbol{\alpha}_e, \boldsymbol{\beta}_e, \boldsymbol{\beta}_o, \boldsymbol{\gamma}_e, \boldsymbol{\gamma}_o \right)^T$.
Further, it can be verified by direct calculation that all elements of the root basis aside from $\varepsilon_1$ and $\varepsilon_2$ are in the periodic part $\operatorname{P}_{\Phi}$ of $\Pic(\X)$ under $\Phi$, so the minimal polynomial with the dynamical degree of the mapping as a root can be obtained as follows.
\begin{proposition}
The transformation of $Q^{\perp} / \operatorname{P}_{\Phi}$,  induced by $\Phi$ is given with respect to the basis $\left\{ \varepsilon_1 + \operatorname{P}_{\Phi}, \varepsilon_2 + \operatorname{P}_{\Phi} \right\}$ 
by the matrix 
\begin{equation}
\left(\begin{array}{cc}0 & -1 \\1 & 2m\end{array}\right),
\end{equation}
whose characteristic polynomial is $(t^2 - 2 m t + 1)$.
In particular the dynamical degree of the mapping is equal to the largest root $\lambda = m + \sqrt{m^2-1}$ of this polynomial.
\end{proposition}

%

\begin{remark} \label{rem:sec3sufficientdeauto}
The fact that, in order to read off the dynamical degree of the mapping from the confinement conditions \eqref{classIIIexampleconfcondsall}, it was sufficient to keep $a_i$ and $d_i$ constant as in \eqref{classIIIconstantparams} is precisely because the root variables that are fixed  because of this choice, correspond to elements of $Q^{\perp} \cap \operatorname{P}_{\Phi}$. 
Keeping these parameters fixed still gives a sufficient deautonomisation in the sense of Definition \ref{sufficientdeautonomisationdef}.
\end{remark}

\section{A family of mappings of QRT Class I form}\label{section6}

The next example we will consider is the family of equations \eqref{0inf2k0} introduced in the Introduction, namely 
\begin{equation} \label{section5exampleauto}
\bar{x} + \ubar{x} = \frac{1}{x^{2m}},
\end{equation}
where $m \in \Z_{\geq 1}$,
which are of the form of Class I QRT-type mappings \cite{ancillary}.
This mapping admits the confined singularity pattern $\{0,\infty^{2m},0\}$, and if we take the deautonomisation
\begin{equation} \label{section5examplenonauto}
\bar{x} + \ubar{x} = \frac{1}{x^{2m}} - \frac{b}{x},
\end{equation}
the $n$-dependence of $b$ required for the singularity confinement to persist is $\bar{b} - 2 m b + \ubar{b} = 0$,
from which we can read off minimal polynomial $t^2 - 2m t + 1$ for the dynamical degree of the mapping \eqref{section5exampleauto}, namely 
\begin{equation}
\lambda = m + \sqrt{m^2-1}.
\end{equation}

\subsection{Space of initial conditions}
The space of initial conditions for equation \eqref{section5exampleauto} was constructed in \cite{SCasIC}, which we recall here with some added details which will be necessary for deautonomisation and the computation of the period map.
Consider the mapping 
\begin{equation} \label{section5mappingauto}
\begin{gathered}
\varphi : \p^1 \times \p^1 \dashrightarrow \p^1 \times \p^1, \\
(x,y) \mapsto (\bar{x}, \bar{y}), \\
\bar{x} = \frac{1 - y x^{2m}}{x^{2m}}, \quad \bar{y}=x,
\end{gathered}
\end{equation}
defined by equation \eqref{section5exampleauto}.
In contrast to the examples studied above, here, to construct the space of initial conditions we will have to blow up infinitely near points, i.e. those lying on the exceptional divisors of previous blow-ups.
For this reason we will need to introduce notation for charts covering exceptional divisors according to the following convention: after blowing up a point $\mathfrak{p}_i$ given in some affine chart $(x,y)$ by
\begin{equation}
\mathfrak{p}_i : (x,y) = (x_*, y_*),
\end{equation}
the exceptional divisor $E_i$ of the blow-up of $\mathfrak{p}_i$ is covered by two affine coordinate charts $(u_i,v_i)$ and $(U_i,V_i)$ given by \begin{equation}
\begin{aligned}
x&= u_i v_i+ x_{*}  , 			& y &= v_i + y_{*}, \\
u_i &= \frac{x- x_{*}}{y- y_{*} },	 & v_i  &=y-y_{*},
\end{aligned}
\qquad \text{and} \qquad 
\begin{aligned}
x &= V_i+ x_{*}, 			& y&= U_i V_i +y_{*},\\
U_i &= \frac{y-y_*}{x-x_*},		& V_i &= x - x_*.
\end{aligned}
\end{equation}
In particular the exceptional divisor $E_i$ has in these charts the local equations $v_i=0$ and $V_i=0$.

For the case at hand we require $8 m$ blow-ups in total, $4m$ each over the points $(x,y)=(0,\infty)$ and $(x,y)=(\infty,0)$ which we outline now.
Firstly letting $\mathfrak{p}_{1} : (x,y) = (0,\infty)$ and introducing coordinates for the exceptional divisor as above we require first $2m$ blow-ups of points $\mathfrak{p}_{k+1} : (U_{k},V_{k}) = (0,0)$ for $k=1,\dots 2m-1$. 
Note that each of these points $\mathfrak{p}_{k+1}$ lies on the proper transform of the line $y=\infty$ under the blow-ups of $\mathfrak{p}_{1},\dots,\mathfrak{p}_{k}$.
Then we find the point $\mathfrak{p}_{2m+1} : (U_{2m},V_{2m}) = (1,0)$ away from the proper transform of $y=\infty$, and after blowing this up we require another sequence of blow-ups of points $\mathfrak{p}_{2m+2}, \dots, \mathfrak{p}_{4m}$ given in coordinates by $\mathfrak{p}_{2m+k+1} : (u_{k},v_{k}) = (0,0)$, $k=1,\dots, 2m-1$, though these points are away from the proper transforms of $y=\infty$ or any of the previous exceptional divisors.
Similarly for $\tilde{\mathfrak{p}}_{1} : (x,y) = (\infty,0)$ we initially blow up $2m$ points on the proper transform of the line $x=\infty$ given by 
$\tilde{\mathfrak{p}}_{k+1} : (\tilde{u}_{k},\tilde{v}_{k}) = (0,0)$, where we have used tildes to denote the coordinates introduced in blowing up points $\tilde{\mathfrak{p}}_i$ to distinguish them from those coming from blow-ups of $\mathfrak{p}_i$.
Then similarly we find $\tilde{\mathfrak{p}}_{2m+1} : (\tilde{u}_{2m},\tilde{v}_{2m}) = (1,0)$ followed by a sequence $\tilde{\mathfrak{p}}_{2m+2}, \dots, \tilde{\mathfrak{p}}_{4m}$ given by $\tilde{\mathfrak{p}}_{2m+k+1} : (\tilde{u}_{k},\tilde{v}_{k}) = (0,0)$, $k=1,\dots,2m-1$.
The resulting surface is shown in Figure \ref{fig:section5examplesurfaces}, where $E_i$ and $\tilde{E}_i$ are exceptional divisors from the blow-ups of $\mathfrak{p}_i$ and $\tilde{\mathfrak{p}}_i$ respectively, and we have introduced notation $D_i, \tilde{D}_i$ for the curves which will be shown to be the components of an anticanonical divisor of $X$ whose classes in $\Pic(X)$ are 
\begin{equation}
\begin{aligned}
\mathcal{D}_0 &= \mathcal{H}_y - \sum_{i=1}^{2m} \E_i,  				&& \mathcal{D}_i = \E_i - \E_{i+1}, \quad i = 1,\dots, 4m-1, \\
\tilde{\mathcal{D}}_0 &= \mathcal{H}_x - \sum_{i=1}^{2m} \tilde{\E}_i,  	&& \tilde{\mathcal{D}}_i = \tilde{\E}_i - \tilde{\E}_{i+1}, \quad i = 1,\dots, 4m-1.
\end{aligned}
\end{equation}
The fact that the mapping becomes an automorphism of $X$ is verified by calculation in charts.

%

\begin{figure}[htb]
\centering
	\begin{tikzpicture}[scale=.75,>=stealth,basept/.style={circle, draw=red!100, fill=red!100, thick, inner sep=0pt,minimum size=1.2mm}]
		\begin{scope}[xshift = -4cm]
			\draw [black, line width = 1pt] 	(4.1,2.5) 	-- (-0.5,2.5)	node [pos=0, right]  {$y=\infty$} node[pos=0, right] {$$};
			\draw [black, line width = 1pt] 	(0,3) -- (0,-1)			node [below] {$x=0$}  node[pos=0, above, xshift=-7pt] {} ;
			\draw [black, line width = 1pt] 	(3.6,3) -- (3.6,-1)		node [pos=0, above]  {$x=\infty$} node[pos=0, above, xshift=7pt] {};
			\draw [black, line width = 1pt] 	(4.1,-.5) 	-- (-0.5,-0.5)	node [left]  {$y=0$} node[pos=0, right] {$$};

			\node (p1) at (3.6,-.5) [basept,label={[xshift=-8pt, yshift = 0 pt] $\tilde{\mathfrak{p}}_{1}$}] {};
			\node (p2) at (4.3,-1.3) [basept,label={[xshift=0pt, yshift = -20 pt] $\tilde{\mathfrak{p}}_{2}$}] {};
			\node (p3) at (5.3,-1.3) [label={[xshift=0pt, yshift = -20 pt] $\cdots$}] {};
			\node (p8) at (6.3,-1.3) [basept,label={[xshift=0pt, yshift = -20 pt] $\tilde{\mathfrak{p}}_{4m}$}] {};
			\node (p9) at (0,2.5) [basept,label={[yshift=-20pt, xshift=+10pt] $\mathfrak{p}_{1}$}] {};
			\node (p10) at (-0.65,3.3) [basept,label={[xshift=0pt, yshift = 0 pt] $\mathfrak{p}_{2}$}] {};
			\node (p11) at (-1.65,3.3) [label={[xshift=-2pt, yshift = 0 pt] $\cdots$}] {};
			\node (p16) at (-2.65,3.3) [basept,label={[xshift=0pt, yshift = 0 pt] $\mathfrak{p}_{4m}$}] {};

			\draw [line width = 0.8pt, ->] (p2) -- (p1);
			\draw [line width = 0.8pt, ->] (p3) -- (p2);
			\draw [line width = 0.8pt, ->] (p8) -- (p3);

			\draw [line width = 0.8pt, ->] (p10) -- (p9);
			\draw [line width = 0.8pt, ->] (p11) -- (p10);
			\draw [line width = 0.8pt, ->] (p16) -- (p11);
			\draw (1.8,-1.8) node [below,anchor=north] {$\p^1 \times\p^1$};
	\end{scope}
	
		\draw [->] (3.5,1)--(2,1) node[pos=0.5, below] {$\pi$};
		
	\begin{scope}[xshift = 5cm, yshift=-.5cm]
			\draw [black, dashed, line width = 1pt] 	(0.,-.2) -- (0,1.35)			;
			\draw [black, dashed, line width = 1pt] 	(-.2,0) 	-- (1.35,0)	;
			\draw [black, line width = 1pt] 	(2.45,4) 	-- (4.25,4)	node [pos=0.5, above]  {} node[pos=0, right] {};
			\draw [black, line width = 1pt] 	(4,2.45) 	-- (4,4.25)	node [pos=0.5, right]  {} node[pos=0, above, xshift=0pt] {};
			\path [line width = 1pt, bend right]	(-.1, 1.1) edge (.5,2) node[pos=0.5, xshift=15pt, yshift=35pt] {};
			\path [line width = 1pt, bend right]	(.4, 1.9) edge (1.1,2.6) node[midway, xshift=15pt, yshift=35pt] {};
			\node at (1.25,2.75)  [label={[xshift=+3pt, yshift = -10 pt] $\iddots$}] {};
			\path [line width = 1pt, bend right]	(1.6, 3.1) edge (2.6,4.1) node[pos=0.5, xshift=6pt, yshift=0pt] {};
			\draw [line width =1pt]	(2.3,3.3) -- (1.4,4.2) ;
			\draw [line width =1pt]	(1.4,4) -- (2.3,4.9) ;
			\node at (2,4.2)  [label={[xshift=-5pt, yshift = 5 pt] $\ddots$}] {};
			\draw [line width =1pt]	(.7,4.7) -- (1.6,5.6) ;
			\draw [red, line width =1pt]	(.9,4.7) -- (0,5.6) ;

			\path [line width = 1pt, bend left]	(1.1,-.1) edge (2,.5) node[pos=0.5, xshift=15pt, yshift=35pt] {};
			\path [line width = 1pt, bend left]	(1.9,.4) edge (2.6,1.1) node[midway, xshift=15pt, yshift=35pt] {};
			\node at (2.75,1.25)  [label={[xshift=+2pt, yshift = -10 pt] $\iddots$}] {};
			\path [line width = 1pt, bend left]	( 3.1,1.6) edge (4.1,2.6) node[pos=0.5, xshift=6pt, yshift=0pt] {};
			\draw [line width =1pt]	(3.3,2.3) -- (4.2,1.4) ;
			\draw [line width =1pt]	(4,1.4) -- (4.9,2.3) ;
			\node at (4.7,1.5)  [label={[xshift=7pt, yshift = -7 pt] $\ddots$}] {};
			\draw [line width =1pt]	(4.7,.7) -- (5.6,1.6) ;
			\draw [red, line width =1pt]	(4.7,.9) -- (5.6,0) ;
		\node at (3.35,4.3) {\scriptsize $D_0$};
		\node at (0,1.7) {\tiny $D_{1}$};
		\node at (.65,2.55) {\tiny $D_{2}$};

		\node at (2.1, 2.85) {\tiny $D_{2m}$};

		\node at (2.9, 5.1) {\tiny $D_{2m+2}$};
		\node at (2.2, 5.8) {\tiny $D_{4m-1}$};
		\node[red] at (0,5.9) {\tiny $E_{4m}$};

		\node at (4.3,3.35) {\tiny $\tilde{D}_0$};
		\node at (1.7, 0) {\tiny $\tilde{D}_{1}$};
		\node at (2.55, .65) {\tiny $\tilde{D}_{2}$};
		
		\node at (2.75, 2) {\tiny $\tilde{D}_{2m}$};

		\node at (5.4, 2.6) {\tiny $\tilde{D}_{2m+2}$};
		\node at (6.1, 1.9) {\tiny $\tilde{D}_{4m-1}$};
		\node[red] at (6.1, 0) {\tiny $\tilde{E}_{4m}$};

			\draw (1.8,-1.2) node [below,anchor=north] {$X$};

	\end{scope}

	\end{tikzpicture}
	\caption{Space of initial conditions for Class I example}
	\label{fig:section5examplesurfaces}
\end{figure}

\subsection{Space of initial conditions for deautonomised version}

We take the deautonomised mapping \eqref{section5examplenonauto} as
\begin{equation} \label{section5mappingnonauto}
\begin{gathered}
\varphi_n : \p^1 \times \p^1 \rightarrow \p^1 \times \p^1, \\
(x,y) \mapsto (\bar{x}, \bar{y}), \\
\bar{x} = \frac{1 - b x^{2m-1} - y x^{2m}}{x^{2m}}, \quad \bar{y}=x,
\end{gathered}
\end{equation}
with the confinement condition 
\begin{equation} \label{section5exampleconfcondb}
\bar{b} - 2 m b + \ubar{b} = 0.
\end{equation}
The points to be blown up to construct the surface $X_n$ that provides a space of initial conditions for mapping \eqref{section5mappingnonauto}, are given by the same expressions in coordinates as for the autonomous case with the exception of the final points $\mathfrak{p}_{4m}$ and $\tilde{\mathfrak{p}}_{4m}$, which for the non-autonomous mapping are
\begin{equation}
\mathfrak{p}_{4m} : (u_{4m-1}, v_{4m-1}) = ( b, 0 ), \qquad \tilde{\mathfrak{p}}_{4m} : (\tilde{u}_{4m-1}, \tilde{v}_{4m-1}) = ( \ubar{b}, 0 ).
\end{equation}
Denoting again the exceptional divisors arising from the blow-ups by $E_i = \pi_n^{-1}(\mathfrak{p}_i)$ and $\tilde{E}_i = \pi_n^{-1}(\tilde{\mathfrak{p}}_i)$, we identify all $\Pic(X_n)$ into the single $\Z$-module
\begin{equation}
\Pic(\X) = \Z \mathcal{H}_x + \Z \mathcal{H}_y + \sum_{i=1}^{4m} \Z \E_i + \sum_{i=1}^{4m} \Z \tilde{\E}_i,
\end{equation}
where $[E_i] = \iota_n(\E_i)$ and $[\tilde{E}_i] = \iota_n(\tilde{\E}_i)$.
\begin{proposition} \label{section5matrixPhionPic}
With the confinement condition \eqref{section5exampleconfcondb}, the mapping $\varphi_n$ in \eqref{section5mappingnonauto} becomes an isomorphism $\tilde{\varphi}_n = \pi_{n+1}^{-1} \circ \varphi_n \circ \pi_n : X_n \rightarrow X_{n+1}$,
and its pullback induces the following lattice automorphism $\Phi =  \iota_n^{-1} \circ \varphi_n^* \circ \iota_{n+1}$ of $\Pic(\X)$:
\begin{equation}
\Phi : 
\left\{ 
\begin{aligned}
\mathcal{H}_x 				&\mapsto 		2m \mathcal{H}_x + \mathcal{H}_y - \sum_{i=1}^{4m} \mathcal{E}_i, 							\quad
\mathcal{H}_y 				\mapsto		\mathcal{H}_x, \quad
\mathcal{E}_i 				\mapsto 		\tilde{\mathcal{E}}_i, \quad
\tilde{\mathcal{E}}_i 			\mapsto 		\mathcal{H}_x - \mathcal{E}_{4m+1-i}.
\end{aligned}
\right.
\end{equation}
\end{proposition}

\subsection{Root basis and period map}
The rational 2-form $\omega$ on $\p^1 \times \p^1$ preserved by $\varphi_n$ is given by 
\begin{equation}
\omega = k dx \wedge dy  = - k \frac{ dx \wedge dY}{Y^2} = - k \frac{ dX \wedge dy}{X^2}  = k \frac{dX\wedge dY}{X^2 Y^2},
\end{equation}
 where $k\in \C^*$, and we denote its lift to $X_n$ by $\tilde{\omega}_n = \pi_n^* \omega$.
 The pole divisor of $\tilde{\omega}_n$ is effective and is given in terms of the curves $D_j$, $\tilde{D}_j$ by 
 \begin{equation}
 - \div \tilde{\omega}_n = 2 D_0 + \sum_{j=1}^{2m} j D_j + \sum_{j=1}^{2m-1} (2m-j) D_{2m+j} + 2 \tilde{D}_0 + \sum_{j=1}^{2m} j \tilde{D}_j + \sum_{j=1}^{2m-1} (2m-j) \tilde{D}_{2m+j}.
 \end{equation}
 To construct a root basis in this case we perform a similar procedure to that outlined in the proof of Proposition \ref{sufficientsubsetprop}.
The singularity pattern $\{0,\infty^{2m},0\}$ corresponds to the movement of the following exceptional curves of the first kind under the mapping according to 
\begin{equation}
[ x=0  ]  \longrightarrow \tilde{E}_{4m} \longrightarrow  E_{4m} \longrightarrow  [ y = 0],
\end{equation}
where $[ x=0  ] $ and $[ y=0  ]$ indicate proper transforms under the blow-ups. 
Passing to $\Pic(\X)$ we have
\begin{equation}
\mathcal{H}_x - \E_1 \xleftarrow{\quad \Phi\quad }  \tilde{\E}_{4m} \xleftarrow{\quad \Phi\quad } \E_{4m} \xleftarrow{\quad \Phi\quad } \mathcal{H}_y - \tilde{\E}_{1},
\end{equation}
so we take $\mathcal{H}_y - \tilde{\E}_1$ and see that the component of $D$ it intersects corresponds to $\tilde{\mathcal{D}}_1 = \tilde{\E}_1 - \tilde{\E}_2$, which has period $4$ under $\Phi$, so we compute 
\begin{equation}
\Phi^4(\mathcal{H}_y - \tilde{\E}_1) = 2 m\mathcal{H}_x + \mathcal{H}_y - \sum_{i=1}^{4m} \E_i - \tilde{\E}_{1}.
\end{equation} 
and obtain an element of $Q^{\perp}$ given by  
\begin{equation}
\alpha_1 = (\Phi^4 - 1) (\mathcal{H}_y - \tilde{\E}_1) = 2 m\mathcal{H}_x - \sum_{i=1}^{4m} \E_i.
\end{equation}
Applying $\Phi$ to $\alpha_1$ we see that a root basis can be found as follows.
\begin{proposition} \label{rootbasispropsection5}
We have a root basis for $Q^{\perp}$ given by 
\begin{equation}
\alpha_1 = 2 m \mathcal{H}_x - \sum_{i=1}^{4m}\E_i, \quad \alpha_2 = 2 m \mathcal{H}_y - \sum_{i=1}^{4m}\tilde{\E}_i.
\end{equation}
\end{proposition}
In Figure \ref{fig:dynkindiagramclassI} we give the intersection diagram of the irreducible components of the anticanonical divisor as well as the diagram corresponding to the generalised Cartan matrix formed by the numbers $c_{i,j} = 2 \frac{\alpha_i \cdot \alpha_j}{\alpha_i \cdot \alpha_i}$, which here is
\begin{equation}
\left(\begin{array}{cc}2 & -2m \\-2m & 2\end{array}\right).
\end{equation}
In particular in the $m=1$ case $X_n$ is a Sakai surface of type $D_7^{(1)}$, and the the root basis in Proposition \ref{rootbasispropsection5} is a basis of simple roots for an affine root system of type $\underset{|\alpha|^2=4}{A_1^{(1)}}$, where the non-standard length of roots does not affect the fact that the matrix with entries $c_{i,j} = 2 \frac{\alpha_i \cdot \alpha_j}{\alpha_i \cdot \alpha_i}$ becomes the usual generalised Cartan matrix of type $A_1^{(1)}$.
\begin{figure}[htb]
\centering
\begin{tikzpicture}[scale=.9, 
				elt/.style={circle,draw=black!100, fill=blue!75, thick, inner sep=0pt,minimum size=1.75mm},
				blk/.style={circle,draw=black!100, fill=black!75, thick, inner sep=0pt,minimum size=1.75mm}, 
				mag/.style={circle, draw=black!100, fill=magenta!100, thick, inner sep=0pt,minimum size=1.75mm}, 
				red/.style={circle, draw=black!100, fill=red!100, thick, inner sep=0pt,minimum size=1.75mm} ]
	\path 
			(-2,0) 	node 	(d2mleft) [blk, label={[xshift=-15pt, yshift = -10 pt] $\D_{2m}$} ] {}
			(-1,0) 	node 	(d0left) [blk, label={[xshift=-0pt, yshift = -22 pt] $\D_{0}$} ] {}
			(0,0) 	node 		(d0right) [blk, label={[xshift=-0pt, yshift = -22 pt] $\tilde{\D}_0$} ] {}
			(1,0) 	node 		(d2mright) [blk, label={[xshift=18pt, yshift = -10 pt] $\tilde{\D}_{2m}$} ] {}
			(1.75,.75) 	node 		(dupright) [blk, label={[xshift=-0pt, yshift = 0 pt] $\tilde{\D}_{2m+1}$} ] {}
			(2.5,.75) 	node 		(dotstopright) [label={[xshift=6pt, yshift = -7 pt] $\dots$} ] {}
			(3.25,.75) 	node 		(d4mright) [blk, label={[xshift=-0pt, yshift = 0 pt] $\tilde{\D}_{4m-1}$} ] {}
			(1.75,-.75) 	node 		(ddownright) [blk, label={[xshift=-0pt, yshift =-22 pt] $\tilde{\D}_{2m-1}$} ] {}
			(2.5,-.75) 	node 		(dotsbottomright) [label={[xshift=6pt, yshift = -7 pt] $\dots$} ] {}
			(3.25,-.75) 	node 		(d1right) [blk, label={[xshift=0pt, yshift =-22pt] $\tilde{\D}_{1}$} ] {}
			(-2.75,.75) 	node 		(dupleft) [blk, label={[xshift=-0pt, yshift = 0 pt] $\D_{2m+1}$} ] {}
			(-3.5,.75) 	node 		(dotstopleft) [label={[xshift=-4pt, yshift = -7 pt] $\dots$} ] {}
			(-4.25,.75) 	node 		(d4mleft) [blk, label={[xshift=-0pt, yshift = 0 pt] $\D_{4m-1}$} ] {}
			(-2.75,-.75) 	node 		(ddownleft) [blk, label={[xshift=-0pt, yshift =-22 pt] $\D_{2m-1}$} ] {}
			(-3.5,-.75) 	node 		(dotsbottomleft) [label={[xshift=-4pt, yshift = -7 pt] $\dots$} ] {}
			(-4.25,-.75) 	node 		(d1left) [blk, label={[xshift=0pt, yshift =-22pt] $\D_{1}$} ] {};

		\draw [black,line width=1pt ] (d0left) -- (d0right);
		\draw [black,line width=1pt ] (d2mleft) -- (d0left);
		\draw [black,line width=1pt ] (d2mright) -- (d0right);
		\draw [black,line width=1pt ] (dupright) -- (d2mright) -- (ddownright);
		\draw [black,line width=1pt ] (dupleft) -- (d2mleft) -- (ddownleft);
		\draw [black,line width=1pt ] (dupleft) -- (dotstopleft);
		\draw [black,line width=1pt ] (ddownleft) -- (dotsbottomleft);
		\draw [black,line width=1pt ] (dupright) -- (dotstopright);
		\draw [black,line width=1pt ] (ddownright) -- (dotsbottomright);
	\begin{scope}[xshift=6.5cm]
		\path 
		        (-.9,0)     node 		(a1) [blk, label={[xshift=-0pt, yshift = -22 pt] $\alpha_1$} ] {}
		        ( .9,0) 	node  	(a2) [blk, label={[xshift=0pt, yshift = -22 pt] $\alpha_2$} ] {}
			(0,0)		node 	(label) [label={[xshift=0pt, yshift = 3 pt] $2m$} ] {};
			(0,-2)		node 	(label) [label={[xshift=0pt, yshift = 5 pt] $~$} ] {};
		\draw [black,line width=1pt ] (a1) -- (a2);
	\end{scope}
	\end{tikzpicture}
 	\caption{Intersection diagrams of anticanonical divisor components and root basis for Class I example}
	\label{fig:dynkindiagramclassI}
\end{figure}

The evolution of the root basis under $\Phi$ is given by 
\begin{equation} \label{Phionrootbasissection5}
\Phi : \left(\begin{array}{c} \alpha_1 \\ \alpha_2 \end{array}\right) \mapsto \left(\begin{array}{cc}0 & -1 \\1 & 2m\end{array}\right) \left(\begin{array}{c} \alpha_1 \\ \alpha_2 \end{array}\right),
\end{equation}
and we see the factor $(t^2 - 2m t + 1)$ from the characteristic polynomial of $\Phi$ on $\Pic(\X)$ as the characteristic polynomial of its restriction to $Q^{\perp}$, and we can deduce that the dynamical degree of the mapping is $\lambda = m + \sqrt{m^2 -1 }$. 
To see this on the level of the confinement condition on $b$ we use the expressions for the root basis as differences of classes of exceptional curves of the first kind and find the root variables for the root basis in Proposition \ref{rootbasispropsection5} to be 
\begin{equation}
\chi(\alpha_1) = \ubar{b}, \qquad \chi(\alpha_2) = - b.
\end{equation}
Therefore the confinement condition $\bar{b} = 2 m b - \ubar{b}$ corresponds to evolution of the root variables according to the same matrix as $\Phi$ on the root basis as in \eqref{Phionrootbasissection5}, i.e.
\begin{equation}
\bar{\chi}(\bar{\alpha}_1) = b = - \chi(\alpha_2), \qquad
\bar{\chi}(\bar{\alpha}_2) = - \bar{b} = \ubar{b} - 2m b = \chi(\alpha_1) + 2 m \chi(\alpha_2),
\end{equation}
and we see how the parameter evolution reflects the dynamical degree.

\section{Conclusion}

In summary, in this paper we have provided an explanation (as summarised in Theorem \ref{bigsummarytheorem}) of the correspondence between confinement conditions for a vast class of birational mappings of the plane and their dynamical degrees. 
As we explained, this correspondence underlies the efficacy of the method of full deautonomisation by singularity confinement, as a detector of the dynamical degree.
While the problem of effectiveness of the anticanonical divisor class means that the method is not guaranteed to work for birational mappings in general, examples with a space of initial conditions but no effective anticanonical divisor are rare, and Theorem \ref{effectivenesstheorem} can be regarded as an affirmative answer to a non-autonomous counterpart to the conjecture initially posed by Gizatullin, in the case of mappings preserving a rational 2-form.

Another interpretation of the results in this paper is as a generalisation of Takenawa's attempt \cite{takenawaindefinite} towards a version of the Sakai theory of discrete Painlev\'e equations for non-integrable mappings, involving rational surfaces associated via $Q^{\perp}$ to root systems of indefinite type. 
In particular non-integrable mappings from the class to which Theorem \ref{effectivenesstheorem} applies provide a suite of new examples with indefinite type root bases, extending some which appeared in Looijenga's study of rational surfaces with anticanonical cycle \cite{looijenga}.
While we have not confirmed this, it is natural to expect that symmetries of the families of surfaces constructed in this paper can be described in terms of the Weyl group associated to $Q^{\perp}$.
Further, while a classification of the kinds of effective anticanonical divisors of surfaces forming the space of initial conditions for non-integrable mappings cannot be given by a finite list as in the discrete Painlev\'e case, we have given a local classification of the intersection configuration of irreducible components in Section \ref{section4}.


\subsubsection*{Acknowledgements}

The authors thank an anonymous referee for many helpful comments and suggestions that greatly improved the paper.
AS was supported by a Japan Society for the Promotion of Science (JSPS) Postdoctoral Fellowship for Research in Japan and also acknowledges the support of JSPS KAKENHI Grant Numbers 21F21775, 22KF0073 and 24K22843. 
TM and RW would also like to acknowledge support from JSPS, respectively through JSPS KAKENHI Grant Numbers 18K13438 and 23K12996, and 22H01130.

\subsubsection*{Conflict of interest statement}

The authors have no conflicts of interest to declare.

\subsubsection*{Data availability statement} 
There is no associated data.


\appendix
\section{Choice of compactification} \label{app:compactification}

In constructing the space of initial conditions for a mapping coming from a discrete equation 
\begin{equation}  \label{discreteequationgeneralcompactification}
(x_{n+1}, y_{n+1} ) = ( f_n(x_n,y_n), g_n(x_n,y_n) ),
\end{equation}
it is often more convenient to consider the variables $x_n,y_n$ as individual affine coordinates on two copies of $\p^1$, rather than $(x_n,y_n)$ being an affine chart for $\p^2$. 
This amounts to choosing $\p^1 \times \p^1$ rather than $\p^2$ as compactification of $\C^2$, and we note that other choices are possible, for example one of the Hirzebruch surfaces $\mathbb{F}_{\ell}$. 
We use $\p^1 \times \p^1$ when we demonstrate the results of Sections \ref{section2}-\ref{section4} in the examples in Sections \ref{section5} and \ref{section6}, so in this Appendix we give an outline of how the formulation using $\p^2$ translates to the alternative choice of $\p^1 \times \p^1$ as compactification.

\subsection{Autonomous case}
We first consider the case where the functions $f_n$, $g_n$ in the equation \eqref{discreteequationgeneralcompactification} do not depend on $n$ explicitly.
Denoting homogeneous coordinates for $\p^2$ by $[Z_0:Z_1:Z_2]$, the variables $(x,y)=(x_n,y_n)$ can be taken as affine coordinates via $[Z_0:Z_1:Z_2] = [1:x:y]$ on the patch $X_0 \neq 0$.
Regarding $(x_{n+1},y_{n+1})$ as affine coordinates for another copy of $\p^2$ in the same way, the discrete equation \eqref{discreteequationgeneralcompactification} in the autonomous case defines a birational mapping
\begin{equation}
\varphi : \p^2 \dashrightarrow \p^2.
\end{equation}
Alternatively, taking $\p^1\times \p^1$ with coordinates $([X_0:X_1], [Y_0:Y_1])$ and regarding $(x,y)$ as affine coordinates in the two factors via $[X_0:X_1] = [1:x]$ and $[Y_0,Y_1]=[1:y]$ on the patches $X_0\neq0$, $Y_0\neq 0$ respectively, the equation \eqref{discreteequationgeneralcompactification} defines a birational mapping of $\p^1 \times \p^1$ which we denote
\begin{equation}
\varphi' : \p^1 \times \p^1 \dashrightarrow \p^1 \times \p^1. 
\end{equation}
Then $\varphi$ and $\varphi'$ will be birationally conjugate, so the autonomous mapping $\varphi$ has a space of initial conditions if and only if $\varphi'$ does, in the sense that there exists a rational surface $X'$ and birational mapping $\pi' : X' \dashrightarrow \p^1 \times \p^1$ such that $\tilde{\varphi} ' \defeq \pi'^{-1} \circ \varphi' \circ \pi'$ is an automorphism of $X'$.
The birational conjugacy of $\varphi$ and $\varphi'$ is achieved by, for example, the usual elementary birational transformation between $\p^2$ and $\p^1\times \p^1$ which we recall now.

\subsection{Elementary birational transformation between $\p^2$ and $\p^1 \times \p^1$}

Take any two distinct points $p_1$ and $p_2$ on $\p^2$ and choose homogeneous coordinates $[Z_0:Z_1:Z_2]$ such that $p_1 : [Z_0:Z_1:Z_2] = [0:0:1]$ and $p_2 : [Z_0:Z_1:Z_2] = [0:1:0]$.
Blow up these points, denoting the resulting surface by $\operatorname{Bl}_{p_1,p_2}(\p^2)$ and the birational map given by the composition of the two blow-ups $\operatorname{Bl}_{p_1,p_2} : \p^2 \dashrightarrow \operatorname{Bl}_{p_1,p_2}(\p^2)$. 
Then the proper transform under these blow-ups of the line $Z_0=0$ is an exceptional curve of the first kind and can be contracted, which leads to $\p^1 \times \p^1$. 
Take homogeneous coordinates $[X_0:X_1] = [1:x]$ and $[Y_0:Y_1]=[1:y]$ for the two $\p^1$-factors in $\p^1 \times \p^1$ such that the point to which the line $Z_0=0$ was contracted is $q : ([X_0:X_1],[Y_0:Y_1])=([0:1],[0:1])$.
Then the blow-up of this point $\operatorname{Bl}_{q}(\p^1\times \p^1)$ is isomorphic to $\operatorname{Bl}_{p_1,p_2}(\p^2)$, and we have the birational map $\operatorname{Bl}_q^{-1} \circ \operatorname{Bl}_{p_1,p_2} : \p^2 \rightarrow \p^1 \times \p^1$.
In coordinates this is given explicitly by
\begin{equation}
\p^2 \ni [Z_0:Z_1:Z_2] \mapsto ( [X_0:X_1], [Y_0:Y_1]) =( [Z_0:Z_1], [Z_0:Z_2]) \in \p^1 \times \p^1,
\end{equation}
which is exactly compatible with how we took $(x,y)$ as coordinates on an affine subset of $\p^2$ versus $\p^1 \times \p^1$.

Denote the class of a hyperplane divisor on $\p^2$ by $\mathcal{H} \in \Pic(\p^2)$, and the classes of exceptional divisors of the blow-ups of $p_1$ and $p_2$ by $\mathcal{E}_1$ and $\mathcal{E}_2$ respectively.
Similarly denote the classes of lines of constant $x$ and $y$ on $\p^1 \times \p^1$ by $\mathcal{H}_1$ and $\mathcal{H}_2$ respectively, and the class of the exceptional divisor of the blow-up of $q$ by $\mathcal{F}$.
Then we have the following correspondence between $\Pic(\operatorname{Bl}_{p_1,p_2}(\p^2))$ and $\Pic(\operatorname{Bl}_{q}(\p^1\times \p^1))$:
\begin{equation}
\begin{aligned}
\mathcal{H} &= \mathcal{H}_1 + \mathcal{H}_2 - \mathcal{F},\\
\mathcal{E}_1 &= \mathcal{H}_2 - \mathcal{F},\\
\mathcal{E}_2 &= \mathcal{H}_1 - \mathcal{F},
\end{aligned}
\quad \text{ and conversely }
\quad 
 \begin{aligned}
\mathcal{H}_1 &= \mathcal{H} - \mathcal{E}_1,\\
\mathcal{H}_2 &= \mathcal{H} - \mathcal{E}_2,\\
\mathcal{F} &= \mathcal{H} - \mathcal{E}_1 - \mathcal{E}_2.
\end{aligned}
\end{equation}
We give a schematic illustration of this in Figure \ref{fig:elementarytransformation}.

\begin{figure}[htb]
\centering
    \begin{tikzpicture}[basept/.style={circle, draw=red!100, fill=red!100, thick, inner sep=0pt,minimum size=1.2mm},baseptblue/.style={circle, draw=blue!100, fill=blue!100, thick, inner sep=0pt,minimum size=1.2mm},scale=.8]
    
    \begin{scope}[xshift=+5cm, yshift=-8cm]
    \node at (0,-3) {$\mathbb{P}^1 \times \mathbb{P}^1$};
    \draw[thick,black] (-2.5,-2)  -- (+2.5,-2)  node[pos=0,left] {\small $y=0$} ;
    \draw[thick,red] (-2.5,+2) -- (+2.5,+2) node[pos=0,left,black] {\small $y=\infty$} ;
    \draw[thick,black] (-2,-2.5) -- (-2,+2.5) node[pos=0,below] {\small $x=0$} ;
    \draw[thick,red] (+2,-2.5) -- (+2,+2.5)node[pos=0,below,black] {\small$ x=\infty$} ;
    \node (q1) at (2,2) [baseptblue,label={[xshift=10pt, yshift = 0 pt] \small $q$}] {};

    \end{scope}

    \draw[dashed,thick,black,->] (-1.5,-8)--(+1,-8) node[ midway,above] {$\operatorname{Bl}_q^{-1} \circ \operatorname{Bl}_{p_1,p_2}$};

    \begin{scope}[xshift=-4.5cm, yshift=-8cm]
    \node at (0,-3) {$\mathbb{P}^2$};
    \draw[thick,black] (-2.5,-2) -- (+2.5,-2) node [pos=0,left] {\small $Z_2=0$} ;
    \draw[thick,black] (-2.25,-2.5) -- (+.3,+2) node [pos=0,below,black] {\small $Z_1=0$} ;
    \draw[thick,blue] (+2.25,-2.5) -- (-.3,+2) node [pos=1,left,black] {\small $Z_0=0$} ;
    \node (p1) at (0,+1.475) [basept,label={[xshift=12pt, yshift = -10 pt] \small $p_{1}$}] {};
    \node (p2) at (1.96,-2) [basept,label={[xshift=10pt, yshift = 0 pt] \small $p_{2}$}] {};

    \end{scope}
    
    \begin{scope}[yshift=-1cm]
    \node at (0,3.5) {$\operatorname{Bl}_{p_1,p_2} (\p^2) \cong \operatorname{Bl}_{q}(\p^1 \times \p^1) $};
    \draw[thick,black] (-2.5,-2)  -- (+2.5,-2)  ;
    \draw[thick,red] (-2.5,+2) -- (+1,+2)  node [midway, above, red] {$\mathcal{E}_1$};
    \draw[thick,black] (-2,-2.5) -- (-2,+2.5) ;
    \draw[thick,red] (+2,-2.5) -- (+2,+1) node [midway, right, red] {$\mathcal{E}_2$};    
    \draw[thick,blue] (+2.5,0) -- (0,+2.5) node [midway, above right, blue] {$\mathcal{F}$};    
    \end{scope}
	\draw[thick, black, dashed,->] (-4,-5.5) -- (-2.5,-4) node[midway, above left] {$\operatorname{Bl}_{p_1,p_2}$};
	\draw[thick, black, dashed,->] (+4,-5.5) -- (+2.5,-4) node[midway, above right] {$\operatorname{Bl}_{q}$};

\end{tikzpicture}
\caption{Elementary birational transformation $\operatorname{Bl}_q^{-1} \circ \operatorname{Bl}_{p_1,p_2} : \p^2 \dashrightarrow \p^1 \times \p^1$}
\label{fig:elementarytransformation}
\end{figure}

\subsection{Non-autonomous case}

The setup in Section \ref{section2} can be adapted to use $\p^1 \times \p^1$ as compactification instead of $\p^2$, which we outline now. 
We use $'$ to indicate objects associated with the $\p^1 \times \p^1$ choice, for example $\varphi'_n$ for a non-autonomous mapping of $\p^1 \times \p^1$ and $X'_n$ for surfaces forming its space of initial conditions obtained from $\p^1\times\p^1$. 
Again the main technical consideration arises from the need to restrict the $n$-dependence of the mapping and surfaces to prevent pathological examples, for which we will need the following.

\begin{definition}[geometric basis; $\p^1\times\p^1$ version]
Let $X'$ be a rational surface such that there exists a birational morphism $\pi' : X' \rightarrow \p^1 \times \p^1$ and $(h^{(1)},h^{(2)},f^{(1)}, \dots, f^{(r')})$ be a $\Z$-basis for $\Pic(X')$. 
We call $(h^{(1)},h^{(2)},f^{(1)}, \dots, f^{(r')})$ a geometric basis if there exists a composition of blow-ups $\pi' = \pi'^{(1)} \circ \cdots \circ \pi'^{(r')} : X' \rightarrow \p^1 \times \p^1$ such that $h^{(1)} = \pi'^* (\mathcal{O}_{\p^1}(1)\times 1), h^{(2)} = \pi'^* (1 \times \mathcal{O}_{\p^1}(1))$ and $f^{(i)}$ is the class of the total transform (under $\pi'^{(i+1)} \circ \cdots \circ \pi'^{(r')}$) of the exceptional curve of $\pi'^{(i)}$ for $i=1,\dots,r'$. 
\end{definition}
Then we can define a space of initial conditions for a non-autonomous mapping on $\p^1 \times \p^1$ for example as follows.

\begin{definition}[space of initial conditions for a non-autonomous mapping; without blow-downs; $\p^1 \times \p^1$ version] \label{def:spaceofICsP1P1}
A space of initial conditions for a non-autonomous mapping $\varphi_n' : \p^1\times \p^1 \dashrightarrow \p^1 \times \p^1$ consists of sequences $(X_n')_n$ and $(\pi_n')_n$, where each $X_n'$ is a rational surface and $\pi_n'$ is a birational morphism $\pi_n' : X_n' \rightarrow \p^1 \times \p^1$ written as a sequence of blow-ups as $\pi_n' = \pi_n'^{(1)} \circ \cdots \circ \pi_n'^{(r')}$, such that the following conditions hold:
\begin{itemize}
\item The mappings $\varphi_n'$ become isomorphisms $\tilde{\varphi}'_n \defeq \pi_{n+1}'^{-1} \circ \varphi_n' \circ \pi_n'$ as in Figure \ref{fig:spaceofICsP1P1}.

\item Let $e_n' =\left(h_n^{(1)},h_n^{(2)},f_n^{(1)}, \dots, f_n^{(r')}\right)$ be the geometric basis for $\Pic(X_n')$ corresponding to $\pi_n'$. Then the matrices of $\tilde{\varphi}_n'^* : \Pic(X_{n+1}' \rightarrow \Pic(X_n')$ with respect to these bases do not depend on $n$.
\item The set of effective classes in $\Pic(X'_n)$ in terms of the basis $e_n'$ does not depend on $n$.
\end{itemize} 

\end{definition}
\begin{figure}[htb]
\begin{center}
\begin{tikzcd}[sep=1cm]
\cdots  \arrow[r] & X'_{n-1}  \arrow[r, "\tilde{\varphi}'_{n-1}"] \arrow[d, "\pi'_{n-1}"] & X'_{n} \arrow[r, "\tilde{\varphi}'_{n}"]  \arrow[d, "\pi'_{n}"] & X'_{n+1} \arrow[r]  \arrow[d, "\pi'_{n+1}"] &\cdots \\
\cdots \arrow[r,dashed] &\p^1\times \p^1 \arrow[r, "\varphi'_{n-1}",dashed] &\p^1\times \p^1 \arrow[r, "\varphi'_n",dashed] &\p^1\times \p^1 \arrow[r,dashed] & \cdots ,
\end{tikzcd}
\end{center}
 	\caption{Space of initial conditions for a non-autonomous mapping; without blow-downs}
	\label{fig:spaceofICsP1P1}
\end{figure}

It is also possible to adapt the more general Definition \ref{spaceofICsdefblowdowns} allowing blow-downs to the $\p^1 \times \p^1$ setting, but all of the examples in this paper and the follow-up \cite{part2} involve only blow-ups, so to use them as demonstrations of the general theory the above definition will suffice.

The following Lemma allows the examples in Sections \ref{section5} and \ref{section6} to serve as demonstrations of the general theory developed in Sections \ref{section2} to \ref{section4}.

\begin{lemma} \label{lem:spaceofICsP1P1toP2}
Let $\varphi_n' : \p^1 \times \p^1 \dashrightarrow \p^1 \times \p^1$ be a non-autonomous mapping with a space of initial conditions as in Definition \ref{def:spaceofICsP1P1} with 
\begin{itemize}
    \item surfaces $X_n'$,
    \item birational morphisms $\pi_n' : X_n' \rightarrow \p^1 \times \p^1$ written as sequences of blow-ups as $\pi_n' = \pi_n'^{(1)} \circ \cdots \circ \pi_n'^{(r')}$,
    \item geometric bases $e_n' = (h_n^{(1)},h_n^{(2)},f_n^{(1)}, \dots, f_n^{(r')})$ for $Pic(X_n')$ corresponding to $\pi_n'$,
    \item isomorphisms $\tilde{\varphi}_n' = (\pi_{n+1}')^{-1} \circ \varphi_n' \circ \pi_n' : X_n' \rightarrow X_{n+1}'$.
\end{itemize}
Let $\eta_n : \p^1 \times \p^1 \dashrightarrow \p^2$ be the elementary birational transformation outlined above with $q$ being the centre of the blow-up $\pi_n'^{(1)}$.
Note that this involves choosing coordinates for $\p^1 \times \p^1$ and $\p^2$.
Then we have a space of initial conditions for $\varphi_n \defeq \eta_{n+1} \circ \varphi_n' \circ \eta_n^{-1} : \p^2 \dashrightarrow \p^2$ (without blow-downs) consisting of
\begin{itemize}
    \item surfaces $X_n = X_{n}'$,
    \item birational morphisms $\pi_n : X_n\rightarrow \p^2$ written as sequences of blow-ups as $\pi_n = \pi_n^{(1)}\circ \cdots \circ \pi_n^{(r)}$, with 
    \begin{itemize}
        \item $r = r'+1$, 
        \item $\pi_n^{(k)} = \pi_n'^{(k-1)}$ for $k=3,\dots,r=r'+1$,
        \item $\pi_n^{(1)}$, $\pi_n^{(2)}$ being the contractions in $\eta_n$,
    \end{itemize}    
    \item geometric bases $e_n = (e_n^{(0)},e_n^{(1)}, \dots, e_n^{(r)})$ for $Pic(X_n)$ corresponding to $\pi_n$, related to $e_n'$ by 
    \begin{equation}
    \begin{gathered}
                e_n^{(0)} = h_n^{(1)} + h_n^{(2)} - f_n^{(1)}, \quad 
        e_n^{(1)} = h_n^{(2)}-f_n^{(1)}, \quad 
        e_n^{(2)} = h_n^{(1)}-f_n^{(1)}, \\
        e_n^{(3)} = f_n^{(2)}, \quad 
        \dots, \quad 
        e_n^{(r)} = f_n^{(r-1)} = f^{(r')}.
    \end{gathered}
    \end{equation}
\end{itemize}
\end{lemma}

For a non-autonomous mapping $\varphi'_n$ with a space of initial conditions as in Definition \ref{def:spaceofICsP1P1}, we can identify all $\Pic(X_n')$ into the Picard lattice $\Pic(\X') = \Z h^{(1)} + \Z h^{(2)} + \sum_{i=1}^{r'} f^{(i)}$ via the geometric bases $e'_n$. 
We then have the lattice automorphism $\Phi' : \Pic(\X')\rightarrow \Pic(\X')$ induced by $(\tilde{\varphi}_n')^*$, similarly to in Definitions \ref{def:picardlattice} and \ref{def:actiononpicardlattice}.
The Picard lattice and action will be related to those of the space of initial conditions for $\varphi_n = \eta_{n+1}\circ \varphi_n' \circ \eta_n^{-1}$ in Lemma \ref{lem:spaceofICsP1P1toP2} according to 
\begin{equation}
\begin{gathered}
        e^{(0)} = h^{(1)} + h^{(2)} - f^{(1)}, \quad 
        e^{(1)} = h^{(2)}-f^{(1)}, \quad 
        e^{(2)} = h^{(1)}-f^{(1)}, \\
        e^{(3)} = f^{(2)}, \quad 
        \dots, \quad 
        e_n^{(r)} = f_n^{(r-1)} = f^{(r')}.
\end{gathered}
\end{equation}



\end{document}